\documentclass[a4paper,UKenglish,numberwithinsect]{article}

\usepackage[hmargin=3.5cm]{geometry}
\usepackage{hyperref} 
\usepackage{pstricks}
\usepackage{pst-node} 
\usepackage{bookmark} 
\usepackage{lipsum}
\usepackage{multido}
\usepackage{paralist}
\usepackage{amsmath} 
\usepackage{setspace}

\newtheorem{definition}{Definition}[section]
\newtheorem{lemma}{Lemma}[section]
\newtheorem{theorem}{Theorem}[section]

\newtheorem{property}{Property}[section]
\newtheorem{proposition}{Proposition}[section]

\newenvironment{proof}{\emph{Proof.}}{\hfill$\square$\medskip}

\usepackage{amsfonts}
\usepackage{amsmath}
\usepackage{amssymb}
\usepackage{xspace}
\usepackage{stmaryrd} 
\usepackage{pifont}

\newcommand\neworrenewcommand[1]{%
    \let#1\relax
    \newcommand#1%
}

\newcommand{\cpp}{C\nolinebreak\hspace{-.05em}\raisebox{.4ex}{\tiny\bf +}\nolinebreak\hspace{-.10em}\raisebox{.4ex}{\tiny\bf +}}
\newcommand\etal{~\textit{et al.}\xspace}

\newcommand\setcomp[2]{\left\{{#1}\ \left|\ {#2}\right.\right\}}
\newcommand\set[1]{\left\{{#1}\right\}}
\newcommand\tuple[1]{\brac{#1}}
\newcommand\brac[1]{\left({#1}\right)}
\newcommand\apply[2]{{#1}\mathord{\brac{#2}}}
\newcommand\idxi{i}
\newcommand\idxj{j} 
\newcommand\idxz{z}

\newcommand\numof{\ell} 
\neworrenewcommand\ordinal{\alpha}
 
\newcommand\exptower[2]{2 \uparrow_{#1} \brac{#2}}

\newcommand\sizeof[1]{\left|{#1}\right|}

\newcommand\bigo[1]{\apply{\mathcal{O}}{#1}} 
\newcommand\poly{f}
\newcommand\polyof[1]{\apply{\poly}{#1}} 
\newcommand\const{d} 
\newcommand\regset{R}
\newcommand\negation[1]{\overline{#1}}

\newcommand\graph{\mathcal{G}} 
\newcommand\vertices{\mathcal{V}}
\newcommand\edges{\mathcal{E}}

\newcommand\bnode[1]{\rnode{#1}{$\bullet$}}

\newcommand\prestar[2]{\apply{Pre^*_{#1}}{#2}}

\newcommand\pre[3]{\apply{Pre^{#1}_{#2}}{#3}}
\newcommand\cpds{\mathcal{C}}
\newcommand\controls{\mathcal{P}}

\newcommand\cpdsrules{\mathcal{R}}
\newcommand\cpdsord{n}
\newcommand\opord{k}
\newcommand\cops[1]{\mathcal{O}_{#1}}
\newcommand\copsgen[1]{\mathcal{G}_{#1}}
\newcommand\genrules{\cpdsrules_{\copsgen{\cpdsord}}} 
\newcommand\genrulesfull[2]{\cpdsrules^{{#1}, {#2}}_{\copsgen{\cpdsord}}} 
\newcommand\cha{a}
\newcommand\chb{b}
\newcommand\chc{c} 
\newcommand\chd{d}

\newcommand\cpush[2]{push^{#2}_{#1}}

\newcommand\scopy[1]{copy_{#1}}
\newcommand\pop[1]{pop_{#1}}
\newcommand\collapse[1]{collapse_{#1}}
\newcommand\rew[1]{rew_{#1}} 
\newcommand\noop{noop}
\newcommand\cpdsruler{r}
\newcommand\cpdsrule[4]{\tuple{{#1},{#2},{#3},{#4}}}
\newcommand\cpdarule[5]{\tuple{{#1},{#2},{#3},{#4},{#5}}}

\newcommand\control{p}
\newcommand\genop{o} 
\newcommand\optuple{\cpdsruler}
\newcommand\opseq{\overrightarrow{\optuple}} 
\newcommand\configc{c}
\newcommand\config[2]{\langle {#1}, {#2} \rangle}

\newcommand\stackw{w}
\newcommand\stacku{u}
\newcommand\stackv{v}

\newcommand\ctop[1]{top_{#1}}

\newcommand\cpdstran{\longrightarrow}
\newcommand\cpdsrun{\longrightarrow^\ast}
\newcommand\alphabet\Sigma 

\newcommand\ccompose[3]{{#1} :_{#2} {#3}}
\newcommand\sbot{\perp}

\newcommand\annot[2]{{#1}^{#2}}

\newcommand\genlang{{\lang_g}} 
\newcommand\runsegment{\sigma}
\newcommand\altrunsegment{\tau}
\newcommand\oalphabet\Gamma
\newcommand\ocha\gamma
\newcommand\cpdatran[1]{\xrightarrow{#1}} 

\newcommand\controlin{{\control_{\text{in}}}}
\newcommand\controlout{{\control_{\text{out}}}}

\newcommand\cpdalang[6]{\apply{\lang^{{#5},{#6}}_{{#2},{#3}, {#4}}}{#1}}

\newcommand{\numstacks}{m} 
\newcommand\emptystack[1]{\sbot_{#1}}
\newcommand\langcheckcpds{\cpds_\emptyset}
\newcommand\regtuples{\chi}

\newcommand\numphases{\zeta}
 
\newcommand\pbrunseg[1]{\runsegment_{#1}}
\newcommand\pbcontrol[1]{\control_{#1}}
\newcommand\pbsa[2]{\saauta^{#1}_{#2}}
\newcommand\pbstack[1]{\iota_{#1}}
\newcommand\pbcpds[1]{\cpds_{#1}}

\newcommand\numscopes{\zeta} 
\newcommand\salyrst[2]{#1^{#2}}
 
\newcommand\lsashift{\text{\tt Shift}}
\newcommand\lsaenvmove{\text{\tt EnvMove}}
\newcommand\lsasaturate[1]{\text{\tt Saturate}_{#1}}
\newcommand\lsapredecessor[1]{\text{\tt Predecessor}_{#1}}
\newcommand\reachgraph[2]{\graph^{#2}_{#1}}
\newcommand\sbmax{N} 
\newcommand\sbstack[2]{{_{#1}\mathord{#2}}}
\newcommand\sbstackbr[2]{\brac{_{#1}\mathord{#2}}}
\newcommand\sbchar[3]{{_{#1,#2}\mathord{#3}}}
\newcommand\sbpr{\mathfrak{p}}
\newcommand\sbcr{\mathfrak{c}}

\newcommand{\AFStack}[2]{\mathrm{Stacks}_{#1}^{#2}}

\newcommand\sastates{\mathbb{Q}}
\newcommand\sastateset{Q}
\newcommand\sadelta{\Delta}
\newcommand\safinals{\mathcal{F}}
\newcommand\sastate{q}

\newcommand\saauta{A}
\newcommand\saautb{B}
\newcommand\sopen[1]{[_{#1}}
\newcommand\sclose[1]{]_{#1}}
\newcommand\sbrac[2]{[{#1}]_{#2}}
\newcommand\lang{\mathcal{L}}
\newcommand\slang[2]{\apply{\lang_{#1}}{#2}}
\newcommand\satran[1]{\xrightarrow{#1}}
\newcommand\satrancol[2]{\xrightarrow[{#2}]{#1}} 
\newcommand\satranfull[4]{{#1} \xrightarrow[{#3}]{#2} \brac{{#4}}} 
\newcommand\satranfullk[3]{{#1} \xrightarrow{#2} \brac{{#3}}}

\newcommand\branch{{col}}
 
\newcommand\sat{t}

\newcommand\satstep{\Pi} 

\newcommand\auxsat[1]{\satstep_{#1}}

\newcommand\gentaaut{\mathcal{T}}
\newcommand\taaut[3]{\gentaaut^{#1}_{{#2},{#3}}}
\newcommand\tastate{t}
\newcommand\tastates[1]{T_{#1}}
\newcommand\tadelta{\delta}
\newcommand\tatran[1]{\xrightarrow{#1}}
\newcommand\tarun[1]{\xrightarrow{#1}_\ast}
\newcommand\tacha{\ctop{1}}
\newcommand\tacontrol{control}

\newcommand\leftcpda{\cpds^L}
\newcommand\rightcpds{\cpds^R}

\newcommand\shortlong[2]{#2}

\newif\ifdraft\draftfalse 

\newenvironment{namedtheorem}[2]{%
    \expandafter\gdef\csname reftheorem#1\endcsname{%
        Theorem~\ref{#1} (#2)%
    }%
    \begin{theorem}[{#2}] \label{#1}%
}{%
    \end{theorem}%
}
\newcommand\reftheorem[1]{\expandafter\csname reftheorem#1\endcsname}

\newenvironment{namedlemma}[2]{%
    \expandafter\gdef\csname reflemma#1\endcsname{%
        Lemma~\ref{#1} (#2)%
    }%
    \begin{lemma}[{#2}] \label{#1}%
}{%
    \end{lemma}%
}
\newcommand\reflemma[1]{\expandafter\csname reflemma#1\endcsname}

\newcommand\refproperty[1]{\expandafter\csname refproperty#1\endcsname}

\newenvironment{nameddefinition}[2]{%
    \expandafter\gdef\csname refdefinition#1\endcsname{%
        Definition~\ref{#1} (#2)%
    }%
    \begin{definition}[{#2}] \label{#1}%
}{%
    \end{definition}%
}
\newcommand\refdefinition[1]{\expandafter\csname refdefinition#1\endcsname}

\title{Saturation of Concurrent Collapsible Pushdown Systems}
\author{M. Hague} 
\date{Royal Holloway University of London, and LIGM, Marne-la-Vall\'ee \\
\texttt{matthew.hague@rhul.ac.uk}}

\begin{document}

\maketitle

\begin{abstract}
    
Multi-stack pushdown systems are a well-studied model of concurrent computation
using threads with first-order procedure calls.  While, in general, reachability
is undecidable, there are numerous restrictions on stack behaviour that lead to
decidability.  To model higher-order procedures calls, a generalisation of
pushdown stacks called collapsible pushdown stacks are required.  Reachability
problems for multi-stack collapsible pushdown systems have been little studied.
Here, we study ordered, phase-bounded and scope-bounded multi-stack collapsible
pushdown systems using saturation techniques, showing decidability of control
state reachability and giving a regular representation of all configurations
that can reach a given control state.

\end{abstract}

\section{Introduction}

Pushdown systems augment a finite-state machine with a stack and accurately
model first-order recursion.  Such systems then are ideal for the analysis of
sequential first-order programs and several successful tools, such as
Moped~\cite{S02} and SLAM~\cite{BR02}, exist for their analysis.  However, the
domination of multi- and many-core machines means that programmers must be
prepared to work in concurrent environments, with several interacting execution
threads.

Unfortunately, the analysis of concurrent pushdown systems is well-known to be
undecidable.  However, most concurrent programs don't interact pathologically
and many restrictions on interaction have been discovered that give decidability
(e.g.~\cite{BET03,BMT05,SV06,HLMS10,K11c}).

One particularly successful approach is \emph{context-bounding}.  This
underapproximates a concurrent system by bounding the number of context switches
that may occur~\cite{QR05}.  It is based on the observation that most real-world
bugs require only a small number of thread interactions~\cite{Q08}.
Additionally, a number of more relaxed restrictions on stack behaviour have been
introduced.  In particular phase-bounded~\cite{lTMP07},
scope-bounded~\cite{lTN11}, and ordered~\cite{BCCC96} (corrected
in~\cite{ABH08}) systems.  There are also generic frameworks --- that bound the
tree-~\cite{MP11} or split-width~\cite{CGK12} of the interactions between
communication and storage --- that give decidability for all communication
architectures that can be defined within them.

Languages such as \cpp, Haskell, Javascript, Python, or Scala increasingly
embrace higher-order procedure calls, which present a challenge to verification.
A popular approach to modelling higher-order languages for verification is that
of (higher-order recursion) schemes~\cite{D82,O06,KNUW05}.   Collapsible
pushdown systems (CPDS) are an extension of pushdown systems~\cite{HMOS08} with
a ``stack-of-stacks'' structure.  The ``collapse'' operation allows a CPDS to
retrieve information about the context in which a stack character was created.
These features give CPDS equivalent modelling power to schemes~\cite{HMOS08}.

These two formalisms have good model-checking properties.  E.g, it is decidable
whether a $\mu$-calculus formula holds on the execution graph of a
scheme~\cite{O06} (or CPDS~\cite{HMOS08}).  Although, the complexity of such
analyses is high, it has been shown by Kobayashi~\cite{K11c} (and Broadbent\etal
for CPDS~\cite{BCHS13}) that they can be performed in practice on real code
examples.

However concurrency for these models has been little studied.  Work by Seth
considers phase-bounding for CPDS without collapse~\cite{S09} by reduction to a
finite state parity game.  Recent work by Kobayashi and Igarashi studies
context-bounded recursion schemes~\cite{KI13}.

Here, we study global reachability problems for ordered, phase-bounded, and
scope-bounded CPDS.  We use \emph{saturation} methods, which have been
successfully implemented by e.g. Moped~\cite{S02} for pushdown systems and
C-SHORe~\cite{BCHS13} for CPDS.  Saturation was first applied to model-checking
by Bouajjani\etal~\cite{BEM97} and Finkel\etal~\cite{FWW97}.  We presented a
saturation technique for CPDS in ICALP 2012~\cite{BCHS12}.  Here, we present the
following advances.  
\begin{compactenum}
    \item Global reachability for ordered CPDSs (\S\ref{sec:ordered-reach}).
          This is based on Atig's algorithm~\cite{A12} for ordered PDSs and
          requires a non-trivial generalisation of his notion of \emph{extended}
          PDSs (\S\ref{sec:extended-cpds}).  For this we introduce the notion of
          \emph{transition automata} that encapsulate the behaviour of the
          saturation algorithm.  In \shortlong{the full
          article}{Appendix~\ref{sec:phase-appendix}} we show how to use the
          same machinery to solve the global reachability problem for
          phase-bounded CPDSs.

    \item Global reachability for scope-bounded CPDSs (\S\ref{sec:scope-reach}).
        This is a backwards analysis based upon La Torre and Napoli's forwards
        analysis for scope-bounded PDSs, requiring new insights to complete the
        proofs.
\end{compactenum}
Because the naive encoding of a single second-order stack has an undecidable MSO
theory (we show this folklore result in \shortlong{the full
paper}{Appendix~\ref{sec:mso-undec}}) it remains a challenging open problem to
generalise the generic frameworks above (\cite{MP11,CGK12}) to CPDSs, since
these frameworks rely on MSO decidability over graph representations of the
storage and communication structure.

\shortlong{A full version of this paper with all definitions and proofs is
available~\cite{H13}.}{}


\section{Preliminaries} 
\label{sec:preliminaries}

Before defining CPDSs, we define $\exptower{0}{x} = x$ and
$\exptower{\idxi+1}{x} = 2^{\exptower{\idxi}{x}}$.  

\subsection{Collapsible Pushdown Systems (CPDS)}

For a readable introduction to CPDS we defer to a survey by Ong~\cite{O13}.
Here, we can only briefly describe higher-order collapsible stacks and their
operations.  We use a notion of collapsible stacks called \emph{annotated
stacks} (which we refer to as collapsible stacks).  These were introduced in
ICALP 2012, and are essentially equivalent to the classical model~\cite{BCHS12}. 

\subparagraph*{Higher-Order Collapsible Stacks}

An order-$1$ stack is a stack of symbols from a stack alphabet $\alphabet$, an
order-$\cpdsord$ stack is a stack of order-$(\cpdsord-1)$ stacks.  A collapsible
stack of order $\cpdsord$ is an order-$\cpdsord$ stack in which the stack
symbols are annotated with collapsible stacks which may be of any order $\leq
\cpdsord$.  Note, often in examples we will omit annotations for clarity.  We
fix the maximal order to $\cpdsord$, and use $\opord$ to range between
$\cpdsord$ and $1$.  We simultaneously define for all $1 \leq \opord \leq
\cpdsord$, the set $\AFStack{\opord}{\cpdsord}$ of order-$\opord$ stacks whose
symbols are annotated by stacks of order at most $\cpdsord$.  Note, we use
subscripts to indicate the order of a stack.  Furthermore, the definition below
uses a least fixed-point.  This ensures that all stacks are finite.  An
order-$\opord$ stack is a collapsible stack in $\AFStack{\opord}{\cpdsord}$.   

\begin{definition}[Collapsible Stacks]
    The family of sets $(\AFStack{\opord}{\cpdsord})_{1 \leq \opord \leq
    \cpdsord}$ is the smallest family (for point-wise inclusion) such that:
    \begin{compactenum}
        \item for all $2 \leq \opord \leq \cpdsord$,
              $\AFStack{\opord}{\cpdsord}$ is the set of all (possibly empty)
              sequences $\sbrac{\stackw_1 \ldots \stackw_\numof}{\opord}$ with
              $\stackw_1, \ldots, \stackw_\numof\in \AFStack{\opord-1}{\cpdsord}$.

        \item $\AFStack{1}{\cpdsord}$ is all sequences
              $\sbrac{\annot{\cha_1}{\stackw_1} \ldots
              \annot{\cha_\numof}{\stackw_\numof}}{1}$ with $\numof \geq 0$ and
              for all $1 \leq \idxi \leq \numof$, $\cha_\idxi$ is a stack symbol
              in $\alphabet$ and $\stackw_\idxi$ is a collapsible stack in
              $\bigcup\limits_{1 \leq \opord \leq \cpdsord}
              \AFStack{\opord}{\cpdsord}$.
    \end{compactenum}
\end{definition}

An order-$\cpdsord$ stack can be represented naturally as an edge-labelled tree
over the alphabet
$\set{\sopen{\cpdsord-1},\ldots,\sopen{1},\sclose{1},\ldots,\sclose{\cpdsord-1}}
\uplus \alphabet$, with $\alphabet$-labelled edges having a second target to the
tree representing the annotation.  We do not use $\sopen{\cpdsord}$ or
$\sclose{\cpdsord}$ since they would appear uniquely at the beginning and end of
the stack.  An example order-$3$ stack is given below, with only a few
annotations shown (on $\cha$ and $\chc$).  The annotations are order-$3$ and
order-$2$ respectively. 
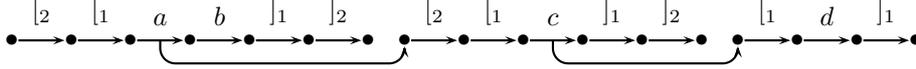
\begin{center}
    \vspace{3ex}
    \begin{psmatrix}[nodealign=true,colsep=2ex,rowsep=2ex]
        \bnode{N1} && \bnode{N2} && \bnode{N3} &\pnode{N34}& \bnode{N4} &&
        \bnode{N5} && \bnode{N6} && \bnode{N7} &
        
        \bnode{N8} && \bnode{N9} && \bnode{N10} &\pnode{N1011}& \bnode{N11} &&
        \bnode{N12} && \bnode{N13} &
        
        \bnode{N14} &\pnode{N1415}& \bnode{N15} && \bnode{N16} && \bnode{N17} \\

        \psset{angle=-90,linearc=.2}
        \ncline{->}{N1}{N2}^{$\sopen{2}$}
        \ncline{->}{N2}{N3}^{$\sopen{1}$}
        \ncline{->}{N3}{N4}^{$\cha$} 
        \ncbar[arm=1.5ex]{->}{N34}{N8}
        \ncline{->}{N4}{N5}^{$\chb$}
        \ncline{->}{N5}{N6}^{$\sclose{1}$}
        \ncline{->}{N6}{N7}^{$\sclose{2}$}

        \ncline{->}{N8}{N9}^{$\sopen{2}$}
        \ncline{->}{N9}{N10}^{$\sopen{1}$}
        \ncline{->}{N10}{N11}^{$\chc$}
        \ncbar[arm=1.5ex]{->}{N1011}{N14}
        \ncline{->}{N11}{N12}^{$\sclose{1}$}
        \ncline{->}{N12}{N13}^{$\sclose{2}$}

        \ncline{->}{N14}{N15}^{$\sopen{1}$}
        \ncline{->}{N15}{N16}^{$\chd$}
        \ncline{->}{N16}{N17}^{$\sclose{1}$}
    \end{psmatrix}
    \vspace{-3ex}
\end{center} 

Given an order-$\cpdsord$ stack
$\stackw = \sbrac{\stackw_1\ldots\stackw_\numof}{\cpdsord}$, we define
$\ctop{\cpdsord+1}(\stackw)=\stackw$ and
\[ 
    \begin{array}{rcll}
        \apply{\ctop{\cpdsord}}{\sbrac{\stackw_1 \ldots
        \stackw_\numof}{\cpdsord}} &=& \stackw_1 & \text{when $\numof > 0$} \\
        
        \apply{\ctop{\cpdsord}}{\sbrac{}{\cpdsord}} &=& \sbrac{}{\cpdsord-1} &
        \text{otherwise} \\ 
        
        \apply{\ctop{\opord}}{\sbrac{\stackw_1 \ldots \stackw_\numof}{\cpdsord}}
        &=& \apply{\ctop{\opord}}{\stackw_1} & \text{when $\opord < \cpdsord$
        and $\numof > 0$} 
    \end{array}        
\]     
noting that $\apply{\ctop{\opord}}{\stackw}$ is undefined if
$\apply{\ctop{\opord'}}{\stackw} = \sbrac{}{\opord'-1}$ for any $\opord' >
\opord$.

We write $\ccompose{\stacku}{\opord}{\stackv}$ --- where $\stacku$ is
order-$(\opord-1)$ --- to denote the stack obtained by placing $\stacku$ on top
of the $\ctop{\opord}$ stack of $\stackv$.  That is, if $\stackv =
\sbrac{\stackv_1 \ldots \stackv_\numof}{\opord}$ then
$\ccompose{\stacku}{\opord}{\stackv} = \sbrac{\stacku \stackv_1 \ldots
\stackv_\numof}{\opord}$, and if $\stackv = \sbrac{\stackv_1 \ldots
\stackv_\numof}{\opord'}$ with $\opord' > \opord$,
$\ccompose{\stacku}{\opord}{\stackv} =
\sbrac{\brac{\ccompose{\stacku}{\opord}{\stackv_1}} \stackv_2 \ldots
\stackv_\numof}{\opord'}$.  This composition associates to the right.  E.g., the
stack $\sbrac{\sbrac{\sbrac{\annot{\cha}{\stackw} \chb}{1}}{2}}{3}$ above can be
written $\ccompose{\stacku}{3}{\stackv}$ where $\stacku$ is the order-$2$ stack
$\sbrac{\sbrac{\annot{\cha}{\stackw} \chb}{1}}{2}$ and $\stackv$ is the empty
order-$3$ stack $\sbrac{}{3}$.  Then
$\ccompose{\stacku}{3}{\ccompose{\stacku}{3}{\stackv}}$ is
$\sbrac{\sbrac{\sbrac{\annot{\cha}{\stackw} \chb}{1}}{2}
\sbrac{\sbrac{\annot{\cha}{\stackw} \chb}{1}}{2}}{3}$.

\subparagraph*{Operations on Order-$\cpdsord$ Collapsible Stacks}

The following operations can be performed on an order-$\cpdsord$ stack where
$\noop$ is the null operation $\apply{\noop}{\stackw} = \stackw$.
\[
    \begin{array}{rcl}
        \cops{\cpdsord} &=& \set{\noop, \pop{1}} \cup \setcomp{\rew{\cha},
        \cpush{\cha}{\opord}, \scopy{\opord}, \pop{\opord}}{\cha \in \alphabet
        \land 2 \leq \opord \leq \cpdsord} 
    \end{array}
\]
We define each $\genop \in \cops{\cpdsord}$ for an order-$\cpdsord$ stack
$\stackw$.  Annotations are created by $\cpush{\cha}{\opord}$, which pushes a
character onto $\stackw$ and annotates it with
$\apply{\ctop{\opord+1}}{\apply{\pop{\opord}}{\stackw}}$.  This, in essence,
attaches a closure to a new character.
\begin{compactenum}
    \item 
        We set $\apply{\pop{\opord}}{\ccompose{\stacku}{\opord}{\stackv}} =
        \stackv$.
        
    \item 
        We set $\apply{\scopy{\opord}}{\ccompose{\stacku}{\opord}{\stackv}} =
        \ccompose{\stacku}{\opord}{\ccompose{\stacku}{\opord}{\stackv}}$.

    \item 
        We set
        $\apply{\collapse{\opord}}{\ccompose{\annot{\cha}{\stacku'}}{1}{\ccompose{\stacku}{(\opord+1)}{\stackv}}}=
        \ccompose{\stacku'}{(\opord+1)}{\stackv}$ when $\stacku$ is
        order-$\opord$ and $1 \leq \opord < \cpdsord$; and
        $\apply{\collapse{\cpdsord}}{\ccompose{\annot{\cha}{\stacku}}{1}{\stackv}}
        =  \stacku$ when $\stacku$ is order-$\cpdsord$.

    \item 
        We set $\apply{\cpush{\chb}{\opord}}{\stackw} =
        \ccompose{\annot{\chb}{\stacku}}{1}{\stackw}$ where $\stacku =
        \apply{\ctop{\opord+1}}{\apply{\pop{\opord}}{\stackw}}$.

    \item
        We set $\apply{\rew{\chb}}{\ccompose{\annot{\cha}{\stacku}}{1}{\stackv}}
        = \ccompose{\annot{\chb}{\stacku}}{1}{\stackv}$.
\end{compactenum}
For example, beginning with $\sbrac{\sbrac{\cha}{1}\sbrac{\chb}{1}}{2}$ and
applying $\cpush{\chc}{2}$ we obtain
$\sbrac{\sbrac{\annot{\chc}{\sbrac{\sbrac{\chb}{1}}{2}} \cha}{1}
\sbrac{\chb}{1}}{2}$.  In this setting, the order-$2$ context information for
the new character $\chc$ is $\sbrac{\sbrac{\chb}{1}}{2}$.  We can then apply
$\scopy{2}; \collapse{2}$ to get
$\sbrac{\sbrac{\annot{\chc}{\sbrac{\sbrac{\chb}{1}}{2}} \cha}{1}
\sbrac{\annot{\chc}{\sbrac{\sbrac{\chb}{1}}{2}} \cha}{1} \sbrac{\chb}{1}}{2}$
then $\sbrac{\sbrac{\chb}{1}}{2}$.  That is, $\collapse{\opord}$ replaces the
current $\ctop{\opord+1}$ stack with the annotation attached to $\chc$.  

\subparagraph*{Collapsible Pushdown Systems}

We are now ready to define collapsible PDS.

\begin{definition}[Collapsible Pushdown Systems]
    An order-$\cpdsord$ \emph{collapsible pushdown system ($\cpdsord$-CPDS)} is
    a tuple $\cpds = \tuple{\controls, \alphabet, \cpdsrules}$ where $\controls$
    is a finite set of control states, $\alphabet$ is a finite stack alphabet,
    and $\cpdsrules \subseteq \brac{\controls \times \alphabet \times
    \cops{\cpdsord} \times \controls}$ is a set of rules.
\end{definition}

We write \emph{configurations} of a CPDS as a pair $\config{\control}{\stackw}
\in \controls \times \AFStack{\cpdsord}{\cpdsord}$.  We have a transition
$\config{\control}{\stackw} \cpdstran \config{\control'}{\stackw'}$ via a rule
$\cpdsrule{\control}{\cha}{\genop}{\control'}$ when $\apply{\ctop{1}}{\stackw} =
\cha$ and $\stackw' = \apply{\genop}{\stackw}$.

\subparagraph*{Consuming and Generating Rules}

We distinguish two kinds of rule or operation: a rule
$\cpdsrule{\control}{\cha}{\genop}{\control'}$ or operation $\genop$ is
\emph{consuming} if $\genop = \pop{\opord}$ or $\genop = \collapse{\opord}$ for
some $\opord$.  Otherwise, it is \emph{generating}.  We write
$\genrulesfull{\controls}{\alphabet}$ for the set of generating rules of the
form $\cpdsrule{\control}{\cha}{\genop}{\control'}$ such that $\control,
\control' \in \controls$ and $\cha \in \alphabet$, and $\genop \in
\cops{\cpdsord}$.  We simply write $\genrules$ when no confusion may arise.

\subsection{Saturation for CPDS}
\label{sec:stackaut}

Our algorithms for concurrent CPDSs build upon the saturation technique for
CPDSs~\cite{BCHS12}.  In essence, we represent sets of configurations $C$ using
a $\controls$-stack automaton $\saauta$ reading stacks.  We define such automata
and their languages $\apply{\lang}{\saauta}$ below.  Saturation adds new
transitions to $\saauta$ --- depending on rules of the CPDS and existing
transitions in $\saauta$ --- to obtain $\saauta'$ representing configurations
with a path to a configuration in $C$.  I.e., given a CPDS $\cpds$ with control
states $\controls$ and a $\controls$-stack automaton $\saauta_0$, we compute
$\prestar{\cpds}{\saauta_0}$ which is the smallest set s.t.
$\prestar{\cpds}{\saauta_0} \supseteq \apply{\lang}{\saauta_0}$ and
$\prestar{\cpds}{\saauta_0} \supseteq \setcomp{\config{\control}{\stackw}}{
\exists \config{\control}{\stackw} \cpdstran \config{\control'}{\stackw'}
\;\textrm{s.t.\;} \config{\control'}{\stackw'} \in \prestar{\cpds}{\saauta_0}}$.

\subparagraph*{Stack Automata} 

Sets of stacks are represented using order-$\cpdsord$ stack automata.  These are
alternating automata with a nested structure that mimics the nesting in a
higher-order {collapsible} stack.  We recall the definition below.

\begin{definition}[Order-$\cpdsord$ Stack Automata]
    An \emph{order-$\cpdsord$ stack automaton} is a tuple
    $
        \saauta = \tuple{ 
                      \sastates_\cpdsord,\ldots,\sastates_1, 
                      \alphabet,
                      \sadelta_\cpdsord,\ldots,\sadelta_1,
                      \safinals_\cpdsord,\ldots,\safinals_1 } $ where
                      $\alphabet$ is a finite stack alphabet,
                      $\sastates_\cpdsord, \ldots, \sastates_1$ are disjoint,
                      and
    \begin{compactenum}
        \item for all $2 \leq \opord \leq \cpdsord$, we have $\sastates_\opord$
              is a finite set of states, $\safinals_\opord \subseteq
              \sastates_\opord$ is a set of accepting states, and
              $\sadelta_\opord \subseteq \sastates_\opord \times
              \sastates_{\opord-1} \times 2^{\sastates_\opord}$ is a transition
              relation such that for all $\sastate$ and $\sastateset$ there is
              \emph{at most one} $\sastate'$ with $\tuple{\sastate, \sastate',
              \sastateset} \in \sadelta_\opord$, and

         \item $\sastates_1$ is a finite set of states, $\safinals_1 \subseteq
               \sastates_1$ is a set of accepting states, and the transition
               relation is $\sadelta_1 \subseteq \bigcup\limits_{2 \leq \opord
               \leq \cpdsord}\brac{\sastates_1 \times \alphabet \times
               2^{\sastates_\opord} \times 2^{\sastates_1}}$.
    \end{compactenum}
\end{definition}

States in $\sastates_\opord$ recognise order-$\opord$ stacks.  Stacks are read
from ``top to bottom''.  A stack $\ccompose{\stacku}{\opord}{\stackv}$ is
accepted from $\sastate$ if there is a transition $\tuple{\sastate, \sastate',
\sastateset} \in \sadelta_\opord$, written $\sastate \satran{\sastate'}
\sastateset$, such that $\stacku$ is accepted from $\sastate' \in
\sastates_{(\opord-1)}$ and $\stackv$ is accepted from each state in
$\sastateset$.  At order-$1$, a stack
$\ccompose{\annot{\cha}{\stacku}}{1}{\stackv}$ is accepted from $\sastate$ if
there is a transition $\tuple{\sastate, \cha, \sastateset_\branch, \sastateset}$
where $\stacku$ is accepted from all states in $\sastateset_\branch$ and
$\stackv$ is accepted from all states in $\sastateset$.  An empty order-$\opord$
stack is accepted by any state in $\safinals_\opord$.  We write $\stackw \in
\slang{\sastate}{\saauta}$ to denote the set of all stacks $\stackw$ accepted
from $\sastate$.  Note that a transition to the empty set is distinct from
having no transition.

We show a part run using $\sastate_3 \satran{\sastate_2} \sastateset_3 \in
\sadelta_3$, $\sastate_2 \satran{\sastate_1} \sastateset_2 \in \sadelta_2$,
$\sastate_1 \satrancol{\cha}{\sastateset_\branch} \sastateset_1 \in \sadelta_1$.
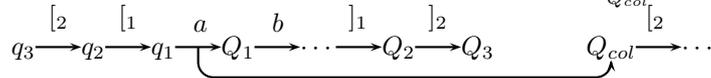
\begin{center}
    \begin{psmatrix}[nodealign=true,colsep=2ex,rowsep=2ex]
        \Rnode{N1}{$\sastate_3$} && \Rnode{N2}{$\sastate_2$} &&
        \Rnode{N3}{$\sastate_1$} &\pnode{N34}& \Rnode{N4}{$\sastateset_1$}
        && \Rnode{N5}{$\cdots$} && \Rnode{N6}{$\sastateset_2$} &&
        \Rnode{N7}{$\sastateset_3$} && && \Rnode{N8}{$\sastateset_\branch$} &&
        \Rnode{N9}{$\cdots$} 

        \psset{angle=-90,linearc=.2}
        \ncline{->}{N1}{N2}^{$\sopen{2}$}
        \ncline{->}{N2}{N3}^{$\sopen{1}$}
        \ncline{->}{N3}{N4}^{$\cha$} 
        \ncbar[arm=1.5ex]{->}{N34}{N8}
        \ncline{->}{N4}{N5}^{$\chb$}
        \ncline{->}{N5}{N6}^{$\sclose{1}$}
        \ncline{->}{N6}{N7}^{$\sclose{2}$}
        \ncline{->}{N8}{N9}^{$\sopen{2}$}
    \end{psmatrix}
\end{center}

\subparagraph*{Long-form Transitions}

We will often use a \emph{long-form} notation (defined below) that captures
nested sequences of transitions.  E.g. we can write
$\satranfull{\sastate_3}{\cha}{\sastateset_\branch}{\sastateset_1,
\sastateset_2, \sastateset_3}$ to represent the use of $\sastate_3
\satran{\sastate_2} \sastateset_3$, $\sastate_2 \satran{\sastate_1}
\sastateset_2$, and $\sastate_1 \satrancol{\cha}{\sastateset_\branch}
\sastateset_1$ for the first three transitions of the run above.  Note that this
latter long-form transition starts at the very beginning of the stack and reads
its $\ctop{1}$ character.  Formally, for a sequence of transitions $\sastate
\satran{\sastate_{\opord-1}} \sastateset_\opord, \sastate_{\opord-1}
\satran{\sastate_{\opord-2}} \sastateset_{\opord-1}, \ldots, \sastate_1
\satrancol{\cha}{\sastateset_\branch} \sastateset_1$ in $\sadelta_\opord$ to
$\sadelta_1$ respectively, we write
$\satranfull{\sastate}{\cha}{\sastateset_\branch}{\sastateset_1,\ldots,\sastateset_\opord}$.

\subparagraph*{$\mathcal{P}$-Stack Automata}

We define $\mathcal{P}$-automata~\cite{BEM97} for CPDSs.  Given control states
$\controls$, an \emph{order-$\cpdsord$ $\controls$-stack automaton} is an
order-$\cpdsord$ stack automaton such that for each $\control \in \controls$
there exists a state $\sastate_\control \in \sastates_\cpdsord$.  We set
$\apply{\lang}{\saauta} = \setcomp{\config{\control}{\stackw}}{\stackw \in
\slang{\sastate_\control}{\saauta}}$.

\subparagraph*{The Saturation Algorithm}

We recall the saturation algorithm.  For a detailed explanation of the
saturation function complete with examples, we refer the reader to our ICALP
paper~\cite{BCHS12}.  Here we present an abstracted view of the algorithm,
relegating details that are not directly relevant to the remainder of the main
article to \shortlong{the full version}{Appendix~\ref{sec:saturation-defs}}.

The saturation algorithm iterates a \emph{saturation function} $\satstep$ that
adds new transitions to a given automaton.  Beginning with $\saauta_0$
representing a target set of configurations, we iterate $\saauta_{\idxi+1} =
\apply{\satstep}{\saauta_\idxi}$ until $\saauta_{\idxi+1} = \saauta_\idxi$.
Once this occurs, we have that $\apply{\lang}{\saauta_\idxi} =
\prestar{\cpds}{\saauta_0}$.  

We define $\satstep$ in terms of a family of auxiliary saturation functions
$\auxsat{\cpdsruler}$ (defined in \shortlong{the full
article}{Appendix~\ref{sec:saturation-defs}}) which return a set of long-form
transitions to be added by saturation.  When $\cpdsruler$ is consuming,
$\apply{\auxsat{\cpdsruler}}{\saauta}$ returns the set of long-form transitions
to be added to $\saauta$ due to the rule $\cpdsruler$.  When $\cpdsruler$ is
generating $\auxsat{\cpdsruler}$ also takes as an argument a long-form
transition $\sat$ of $\saauta$.  Thus $\apply{\auxsat{\cpdsruler}}{\sat,
\saauta}$ returns the set of long-form transitions that should be added to
$\saauta$ as a result of the rule $\cpdsruler$ combined with the transition
$\sat$ (and possibly other transitions of $\saauta$).

For example, if $\cpdsruler = \cpdsrule{\control}{\cha}{\rew{\chb}}{\control'}$
and $\sat =
\satranfull{\sastate_{\control'}}{\chb}{\sastateset_\branch}{\sastateset_1,
\ldots, \sastateset_\cpdsord}$ is a transition of $\saauta$, then
$\apply{\auxsat{\cpdsruler}}{\sat, \saauta}$ contains only the long-form
transition $\sat' =
\satranfull{\sastate_\control}{\cha}{\sastateset_\branch}{\sastateset_1, \ldots,
\sastateset_\cpdsord}$.  The idea is if
$\config{\control'}{\ccompose{\annot{\chb}{\stacku}}{1}{\stackw}}$ is accepted
by $\saauta$ via a run whose first (sequence of) transition(s) is $\sat$,  then
by adding $\sat'$ we will be able to accept
$\config{\control}{\ccompose{\annot{\cha}{\stacku}}{1}{\stackw}}$ via a run
beginning with $\sat'$ instead of $\sat$.  We have
$\config{\control}{\ccompose{\annot{\cha}{\stacku}}{1}{\stackw}} \in
\prestar{\cpds}{\saauta}$ since it can reach
$\config{\control'}{\ccompose{\annot{\chb}{\stacku}}{1}{\stackw}}$ via the rule
$\cpdsruler$.

\begin{nameddefinition}{def:satstep}{The Saturation Function $\satstep$}
    For a CPDS with rules $\cpdsrules$, and given an order-$\cpdsord$ stack
    automaton $\saauta_\idxi$ we define $\saauta_{\idxi+1} =
    \apply{\satstep}{\saauta_\idxi}$. The state-sets of $\saauta_{\idxi+1}$ are
    defined implicitly by the transitions which are those in $\saauta_\idxi$
    plus, for each $\cpdsruler = \cpdsrule{\control}{\cha}{\genop}{\control'}
    \in \cpdsrules$, when
    \begin{compactenum}
        \item $\genop$ is consuming and $\sat \in
            \apply{\auxsat{\cpdsruler}}{\saauta_\idxi}$, then add $\sat$ to
            $\saauta_{\idxi+1}$, 

        \item $\genop$ is generating, $\sat$ is in $\saauta_\idxi$, and $\sat'
              \in \apply{\auxsat{\cpdsruler}}{\sat, \saauta}$, then add $\sat'$
              to $\saauta_{\idxi+1}$.
    \end{compactenum}
\end{nameddefinition}

In ICALP 2012 we showed that saturation adds up to
$\bigo{\exptower{\cpdsord}{\polyof{\sizeof{\controls}}}}$ transitions, for some
polynomial $\poly$, and that this can be reduced to
$\bigo{\exptower{\cpdsord-1}{\polyof{\sizeof{\controls}}}}$ (which is optimal)
by restricting all $\sastateset_\cpdsord$ to have size $1$ when $\saauta_0$ is
``non-alternating at order-$\cpdsord$''.  Since this property holds of all
$\saauta_0$ used here, we use the optimal algorithm for complexity arguments.


\section{Extended Collapsible Pushdown Systems}
\label{sec:extended-cpds}

To analyse concurrent systems, we extend CPDS following Atig~\cite{A12}.  Atig's
extended PDSs allow words from arbitrary languages to be pushed on the stack.
Our notion of extended CPDSs allows sequences of \emph{generating operations}
from a language $\genlang$ to be applied, rather than a single operation per
rule.  We can specify $\genlang$ by any system (e.g.  a Turing machine).

\begin{definition}[Extended CPDSs]
    An order-$\cpdsord$ \emph{extended CPDS ($\cpdsord$-ECPDS)} is a tuple
    $\cpds = \tuple{\controls, \alphabet, \cpdsrules}$ where $\controls$ is a
    finite set of control states, $\alphabet$ is a finite stack alphabet, and
    $\cpdsrules \subseteq \brac{\controls \times \alphabet \times
    \cops{\cpdsord} \times \controls} \cup \brac{\controls \times \alphabet
    \times 2^{\brac{\genrulesfull{\controls}{\alphabet}}^\ast} \times
    \controls}$ is a set of rules.
\end{definition}

As before, we have a transition $\config{\control}{\stackw} \cpdstran
\config{\control'}{\stackw'}$ of an $\cpdsord$-ECPDS via a rule
$\cpdsrule{\control}{\cha}{\genop}{\control'}$ with $\apply{\ctop{1}}{\stackw} =
\cha$ and $\stackw' = \apply{\genop}{\stackw}$.  Additionally, we have a
transition $\config{\control}{\stackw} \cpdstran \config{\control'}{\stackw'}$
when we have a rule $\cpdsrule{\control}{\cha}{\genlang}{\control'}$, a sequence
$\cpdsrule{\control}{\cha}{\genop_1}{\control_1}\cpdsrule{\control_1}{\cha_2}{\genop_2}{\control_2}\ldots\cpdsrule{\control_{\numof-1}}{\cha_\numof}{\genop_\numof}{\control'}
\in \genlang$ and $\stackw' =
\apply{\genop_\numof}{\cdots\apply{\genop_1}{\stackw}}$.  That is, a single
extended rule may apply a sequence of stack updates in one step.  A run of an
ECPDS is a sequence $\config{\control_0}{\stackw_0} \cpdstran
\config{\control_1}{\stackw_1} \cpdstran \cdots$.

\subsection{Reachability Analysis}

We adapt saturation for ECPDSs.  In Atig's algorithm, an essential property is
the decidability of $\genlang \cap \apply{\lang}{\saauta}$ for some order-1
$\controls$-stack automaton $\saauta$ and a language $\genlang$ appearing in a
rule of the extended PDS.  We need analogous machinery in our setting.  For
this, we first define a class of finite automata called \emph{transition}
automata, written $\gentaaut$.  The states of these automata will be long-form
transitions of a stack automaton $\sat =
\satranfull{\sastate}{\cha}{\sastateset_\branch}{\sastateset_1, \ldots,
\sastateset_\cpdsord}$.  Transitions $\tastate \tatran{\cpdsruler} \tastate'$
are labelled by rules.  We write $\tastate \tarun{\opseq} \tastate'$ to denote a
run over $\opseq \in \brac{\genrules}^\ast$.

During the saturation algorithm we will build from $\saauta_\idxi$ a transition
automaton $\gentaaut$.  Then, for each rule
$\cpdsrule{\control}{\cha}{\genlang}{\control'}$ we add to $\saauta_{\idxi+1}$ a
new long-form transition $\sat$ if there is a word $\opseq \in \genlang$ such
that $\tastate \tarun{\opseq} \tastate'$ is a run of $\gentaaut$ and $\sat'$ is
already a transition of $\saauta_\idxi$.

For example, consider $\cpdsrule{\control}{\cha}{\genlang}{\control'}$ where
$\genlang = \set{\cpdsrule{\control}{\cha}{\rew{\chb}}{\control'}}$.  A
transition 
\[ 
    \brac{\satranfull{\sastate_\control}{\cha}{\sastateset_\branch}{\sastateset_1,
    \ldots, \sastateset_\cpdsord}}
    \tatran{\cpdsrule{\control}{\cha}{\rew{\chb}}{\control'}}
    \brac{\satranfull{\sastate_{\control'}}{\chb}{\sastateset_\branch}{\sastateset_1,
    \ldots, \sastateset_\cpdsord}}
\]
will correspond to the fact that the presence of
$\satranfull{\sastate_{\control'}}{\chb}{\sastateset_\branch}{\sastateset_1,
\ldots, \sastateset_\cpdsord}$ in $\saauta_\idxi$ causes
$\satranfull{\sastate_\control}{\cha}{\sastateset_\branch}{\sastateset_1,
\ldots, \sastateset_\cpdsord}$ to be added by $\satstep$.  A run $\tastate_1
\tatran{\optuple_1} \tastate_2 \tatran{\optuple_2} \tastate_3$ comes into play
when e.g. $\genlang = \set{\optuple_1\optuple_2}$.  If the rule were split into
two ordinary rules with intermediate control states, $\satstep$ would first add
$\tastate_2$ derived from $\tastate_3$, and then from $\tastate_2$ derive
$\tastate_1$.  In the case of extended CPDSs, the intermediate transition
$\tastate_2$ is not added to $\saauta_{\idxi+1}$, but its effect is still
present in the addition of $\tastate_1$.  Below, we repeat the above intuition
more formally.  Fix a $\cpdsord$-ECPDS $\cpds = \tuple{\controls, \alphabet,
\cpdsrules}$.

\subparagraph*{Transition Automata}

We build a transition automaton from a given $\controls$-stack automaton
$\saauta$.  Let $\saauta$ have order-$\cpdsord$ to order-$1$ state-sets
$\sastateset_\cpdsord, \ldots, \sastateset_1$ and alphabet $\alphabet$, let
$\tastates{\saauta}$ be the set of all
$\satranfull{\sastate}{\cha}{\sastateset_\branch}{\sastateset_1, \ldots,
\sastateset_\cpdsord}$ with $\sastate \in \sastateset_\cpdsord$, for all
$\opord$, $\sastateset_\opord \subseteq \sastates_\opord$, and for some
$\opord$, $\sastateset_\branch \subseteq \sastates_\opord$.

\begin{nameddefinition}{def:tranaut}{Transition Automata}
    Given an order-$\cpdsord$ $\controls$-stack automaton $\saauta$ with
    alphabet $\alphabet$, and $\tastate, \tastate' \in \tastates{\saauta}$, we
    define the transition automaton $\taaut{\saauta}{\tastate}{\tastate'} =
    \tuple{\tastates{\saauta}, \genrulesfull{\controls}{\alphabet}, \tadelta,
    \tastate, \tastate'}$ such that $\tadelta \subseteq \tastates{\saauta}
    \times \genrulesfull{\controls}{\alphabet} \times \tastates{\saauta}$ is the
    smallest set such that $\sat_1 \satran{\cpdsruler} \sat_2 \in \tadelta$ if
    $\sat_1 \in \apply{\auxsat{\cpdsruler}}{\sat_2, \saauta}$.
\end{nameddefinition}

We define $\apply{\lang}{\taaut{\saauta}{\sat}{\sat'}} = \setcomp{\opseq}{\sat
\tarun{\opseq} \sat'}$.

\subparagraph*{Extended Saturation Function}

We now extend the saturation function following the intuition explained above.
For $\tastate =
\satranfull{\sastate_\control}{\cha}{\sastateset_\branch}{\sastateset_1, \ldots,
\sastateset_\cpdsord}$, let $\apply{\tacha}{\tastate} = \cha$ and
$\apply{\tacontrol}{\tastate} = \control$. 

\begin{definition}[Extended Saturation Function $\satstep$]
    The extended $\satstep$ is $\satstep$ from Definition~\ref{def:satstep} plus
    for each extended rule $\cpdsrule{\control}{\cha}{\genlang}{\control'} \in
    \cpdsrules$ and $\sat, \sat'$, we add $\sat$ to $\saauta_{\idxi+1}$ whenever 
    \begin{inparaenum} 
        \item $\apply{\tacontrol}{\sat} = \control$ and $\apply{\tacha}{\sat} =
              \cha$, 
        \item $\sat'$ is a transition of $\saauta_\idxi$ with
              $\apply{\tacontrol}{\sat'} = \control'$, and 
        \item $\genlang \cap \apply{\lang}{\taaut{\saauta_\idxi}{\sat}{\sat'}}
              \neq \emptyset$.  
    \end{inparaenum}
\end{definition}

\begin{namedtheorem}{thm:ecpds-reachability}{Global Reachability of ECPDS}
    Given an ECPDS $\cpds$ and a $\controls$-stack automaton $\saauta_0$, the
    fixed point $\saauta$ of the extended saturation procedure accepts
    $\prestar{\cpds}{\saauta_0}$.
\end{namedtheorem}

In order for the saturation algorithm to be effective, we need to be able to
decide $\genlang \cap \apply{\lang}{\taaut{\saauta_\idxi}{\sat}{\sat'}} \neq
\emptyset$.  We argue in the \shortlong{full paper}{appendix} that number of
transitions added by extended saturation has the same upper bound as the
unextended case.

\section{Multi-Stack CPDSs}

We define a general model of concurrent collapsible pushdown systems, which we
later restrict.  In the sequel, assume a bottom-of-stack symbol $\sbot$ and
define the ``empty'' stacks $\emptystack{0} = \sbot$ and $\emptystack{\opord+1}
= \sbrac{\emptystack{\opord}}{\opord+1}$.  As standard, we assume that $\sbot$
is neither pushed onto, nor popped from, the stack (though may be copied by
$\scopy{\opord}$).

\begin{definition}[Multi-Stack Collapsible Pushdown Systems]
    An order-$\cpdsord$ \emph{multi-stack collapsible pushdown system
    ($\cpdsord$-MCPDS)} is a tuple $\cpds = \tuple{\controls, \alphabet,
    \cpdsrules_1, \ldots, \cpdsrules_\numstacks}$ where $\controls$ is a finite
    set of control states, $\alphabet$ is a finite stack alphabet, and for each
    $1 \leq \idxi \leq \numstacks$ we have a set of rules $\cpdsrules_\idxi
    \subseteq \controls \times \alphabet \times \cops{\cpdsord} \times
    \controls$.
\end{definition}

A configuration of $\cpds$ is a tuple $\config{\control}{\stackw_1, \ldots,
\stackw_\numstacks}$.  There is a transition
$\config{\control}{\stackw_1,\ldots,\stackw_\numstacks} \cpdstran
\config{\control'}{\stackw_1, \ldots, \stackw_{\idxi-1}, \stackw'_\idxi,
\stackw_{\idxi+1}, \ldots, \stackw_\numstacks}$ via
$\cpdsrule{\control}{\cha}{\genop}{\control'} \in \cpdsrules_\idxi$ when $\cha =
\apply{\ctop{1}}{\stackw_\idxi}$ and $\stackw'_\idxi =
\apply{\genop}{\stackw_\idxi}$.

We also need MCPD\emph{Automata}, which are MCPDSs defining languages over an
input alphabet $\oalphabet$.  For this, we add labelling input characters to the
rules.  Thus, a rule $\cpdarule{\control}{\cha}{\ocha}{\genop}{\control'}$ reads
a character $\ocha \in \oalphabet$.  This is defined formally in \shortlong{the
full paper}{Appendix~\ref{sec:multi-appendix}}.  

We are interested in two problems for a given $\cpdsord$-MCPDS $\cpds$.

\begin{definition}[Control State Reachability Problem]
    Given control states $\controlin, \controlout$ of $\cpds$, decide if there
    is for some $\stackw_1, \ldots, \stackw_\numstacks$ a run
    $\config{\controlin}{\emptystack{\cpdsord}, \ldots, \emptystack{\cpdsord}}
    \cpdstran \cdots \cpdstran \config{\controlout}{\stackw_1, \ldots,
    \stackw_\numstacks}$.
\end{definition}

\begin{definition}[Global Control State Reachability Problem]
    Given a control state $\controlout$ of $\cpds$, construct a representation
    of the set of configurations $\config{\control}{\stackw_1, \ldots,
    \stackw_\numstacks}$ such that there exists for some $\stackw'_1, \ldots,
    \stackw'_\numstacks$ a run $\config{\control}{\stackw_1, \ldots,
    \stackw_\numstacks} \cpdstran \cdots \cpdstran
    \config{\controlout}{\stackw'_1, \ldots, \stackw'_\numstacks}$.
\end{definition}

We represent sets of configurations as follows.  In \shortlong{the full
paper}{Appendix~\ref{sec:multi-appendix}} we show it forms an effective boolean
algebra,  membership is linear time, and emptiness is in PSPACE.

\begin{nameddefinition}{def:regular-set-multi}{Regular Set of Configurations} 
    A regular set $\regset$ of configurations of a multi-stack CPDS $\cpds$ is
    definable via a finite set $\regtuples$ of tuples $\tuple{\control,
    \saauta_1, \ldots, \saauta_\numstacks}$ where $\control$ is a control state
    of $\cpds$ and $\saauta_\idxi$ is a stack automaton with designated initial
    state $\sastate_\idxi$ for each $\idxi$.  We have
    $\config{\control}{\stackw_1, \ldots, \stackw_\numstacks} \in \regset$ iff
    there is some $\tuple{\control, \saauta_1, \ldots, \saauta_\numstacks} \in
    \regtuples$ such that $\stackw_\idxi \in
    \slang{\sastate_\idxi}{\saauta_\idxi}$ for each $\idxi$.  
\end{nameddefinition}

Finally, we often partition runs of an MCPDS $\runsegment =
\runsegment_1\ldots\runsegment_\numof$ where each $\runsegment_\idxi$ is a
sequence of configurations of the MCPDS.  A transition from $\configc$ to
$\configc'$ occurs in segment $\runsegment_\idxi$ if $\configc'$ is a
configuration in $\runsegment_\idxi$.  Thus, transitions from
$\runsegment_\idxi$ to $\runsegment_{\idxi+1}$ are said to belong to
$\runsegment_{\idxi+1}$.

\section{Ordered CPDS}
\label{sec:ordered-reach}

We generalise \emph{ordered multi-stack pushdown systems}~\cite{BCCC96}.
Intuitively, we can only remove characters from stack $\idxi$ whenever all
stacks $\idxj < \idxi$ are empty. 

\begin{definition}[Ordered CPDS]
    An order-$\cpdsord$ \emph{ordered CPDS} ($\cpdsord$-OCPDS) is an
    $\cpdsord$-MCPDS $\cpds = \tuple{\controls, \alphabet, \cpdsrules_1, \ldots,
    \cpdsrules_\numstacks}$ such that a transition from
    $\config{\control}{\stackw_1,\ldots,\stackw_\numstacks}$ using the rule
    $\cpdsruler$ on stack $\idxi$ is permitted iff, when $\cpdsruler$ is
    consuming, for all $1 \leq \idxj < \idxi$ we have $\stackw_\idxj =
    \emptystack{\cpdsord}$.
\end{definition}

\begin{theorem}[Decidability of Reachability Problems]
    For $\cpdsord$-OCPDSs the control state reachability problem and the global
    control state reachability problem are decidable.
\end{theorem}

We outline the proofs below.  In \shortlong{the full
paper}{Appendix~\ref{sec:ordered-appendix}} we show control state reachability
uses  $\bigo{\exptower{\numstacks(\cpdsord-1)}{\numof}}$ time, where $\numof$ is
polynomial in the size of the OCPDS, and we have at most
$\bigo{\exptower{\numstacks\cpdsord}{\numof}}$ tuples in the solution to the
global problem.  First observe that reachability can be reduced to reaching
$\config{\controlout}{\emptystack{\cpdsord}, \ldots, \emptystack{\cpdsord}}$ by
clearing the stacks at the end of the run.

\subparagraph*{Control State Reachability}

Using our notion of ECPDS, we may adapt Atig's inductive algorithm for ordered
PDSs~\cite{A12} for the control state reachability problem.  The induction is
over the number of stacks.  W.l.o.g. we assume that all rules
$\cpdsrule{\control}{\sbot}{\genop}{\control'}$ of $\cpds$ have $\genop =
\cpush{\cha}{\cpdsord}$.

In the base case, we have an $\cpdsord$-OCPDS with a single stack, for which the
global reachability problem is known to be decidable (e.g.~\cite{BEM97}).  

In the inductive case, we have an $\cpdsord$-OCPDS $\cpds$ with $\numstacks$
stacks.  By induction, we can decide the reachability problem for
$\cpdsord$-OCPDSs with fewer than $\numstacks$ stacks.  We first show how to
reduce the problem to reachability analysis of an extended CPDS, and then
finally we show how to decide $\genlang \cap
\apply{\lang}{\taaut{\saauta_\idxi}{\sat}{\sat'}} \neq \emptyset$ using an
$\cpdsord$-OCPDS with $(\numstacks - 1)$ stacks.

Consider the $\numstacks$th stack of $\cpds$.  A run of $\cpds$ can be split
into
$\runsegment_1\altrunsegment_1\runsegment_2\altrunsegment_2\ldots\runsegment_\numof\altrunsegment_\numof$.
During the subruns $\runsegment_\idxi$, the first $(\numstacks-1)$ stacks are
non-empty, and during $\altrunsegment_\idxi$, the first $(\numstacks-1)$ stacks
are empty.  Moreover, during each $\runsegment_\idxi$,  only generating
operations may occur on stack $\numstacks$.

We build an extended CPDS that directly models the $\numstacks$th stack during
the $\altrunsegment_\idxi$ segments where the first $(\numstacks-1)$ stacks are
empty, and uses rules of the form
$\cpdsrule{\control}{\cha}{\genlang}{\control'}$ to encapsulate the behaviour of
the $\runsegment_\idxi$ sections where the first $(\numstacks-1)$ stacks are
non-empty.  The $\genlang$ attached to such a rule is the sequence of updates
applied to the $\numstacks$th stack during $\runsegment_\idxi$.

We begin by defining, from the OCPDS $\cpds$ with $\numstacks$ stacks, an OCPDA
$\leftcpda$ with $(\numstacks-1)$ stacks.  This OCPDA will be used to define the
$\genlang$ described above.  $\leftcpda$ simulates a segment
$\runsegment_\idxi$. Since all updates to stack $\numstacks$ in
$\runsegment_\idxi$ are generating,
$\leftcpda$ need only track its top character, hence only keeps $(\numstacks-1)$
stacks.  The top character of stack $\numstacks$ is kept in the control state,
and the operations that would have occurred on stack $\numstacks$ are output.

\begin{definition}[$\leftcpda$]
    Given an $\cpdsord$-OCPDS $\cpds = \tuple{\controls, \alphabet,
    \cpdsrules_1, \ldots, \cpdsrules_\numstacks}$, we define $\leftcpda$ to be
    an $\cpdsord$-OCPDA with $(\numstacks-1)$ stacks $\tuple{\controls \times
    \alphabet, \alphabet, \cpdsrules'_1 \cup \cpdsrules', \cpdsrules'_2, \ldots,
    \cpdsrules'_{\numstacks-1}}$ over input alphabet $\genrules$ where for all
    $\idxi$ 
    \[
        \cpdsrules'_\idxi = \setcomp{\cpdarule{\tuple{\control,
        \cha}}{\chb}{\cpdsrule{\control}{\cha}{\noop}{\control'}}{\genop}{\tuple{\control',
        \cha}}}{\cha \in \alphabet \land
        \cpdsrule{\control}{\chb}{\genop}{\control'} \in \cpdsrules_\idxi}
        \text{, and}
    \]
    \[
        \begin{array}{rcl}
            \cpdsrules' &=& 
            
            \setcomp{\cpdarule{\tuple{\control,
            \cha}}{\chb}{\cpdsruler}{\noop}{\tuple{\control', \chc}}}{\chb \in
            \alphabet \land \cpdsruler =
            \cpdsrule{\control}{\cha}{\rew{\chc}}{\control'} \in
            \cpdsrules_\numstacks}\ \cup \\

            & & 

            \setcomp{\cpdarule{\tuple{\control,
            \cha}}{\chb}{\cpdsruler}{\noop}{\tuple{\control', \cha}}}{\chb \in
            \alphabet \land \cpdsruler =
            \cpdsrule{\control}{\cha}{\scopy{\opord}}{\control'} \in
            \cpdsrules_\numstacks}\ \cup \\ 
            
            & &
            
            \setcomp{\cpdarule{\tuple{\control,
            \cha}}{\chb}{\cpdsruler}{\noop}{\tuple{\control', \chc}}}{\chb \in
            \alphabet \land \cpdsruler =
            \cpdsrule{\control}{\cha}{\cpush{\chc}{\opord}}{\control'} \in
            \cpdsrules_\numstacks}\ \cup \\

            & &
            
            \setcomp{\cpdarule{\tuple{\control,
            \cha}}{\chb}{\cpdsruler}{\noop}{\tuple{\control', \cha}}}{\chb \in
            \alphabet \land \cpdsruler =
            \cpdsrule{\control}{\cha}{\noop}{\control'} \in
            \cpdsrules_\numstacks} \ .
        \end{array}
    \]
\end{definition}

We define the language
$\cpdalang{\leftcpda}{\control}{\cha}{\control'}{\chb}{\idxi}$ to be the set of
words $\ocha_1\ldots\ocha_\numof$ such that there exists a run of $\leftcpda$
over input $\ocha_1\ldots\ocha_\numof$ from $\config{\tuple{\control,
\cha}}{\stackw_1, \ldots, \stackw_{\numstacks-1}}$ to $\config{\tuple{\control',
\chc}}{\emptystack{\cpdsord}, \ldots, \emptystack{\cpdsord}}$ for some $\chc$,
where $\stackw_\idxi = \apply{\cpush{\chb}{\cpdsord}}{\emptystack{\cpdsord}}$
and $\stackw_\idxj = \emptystack{\cpdsord}$ for all $\idxj \neq \idxi$.  This
language describes the effect on stack $\numstacks$ of a run $\runsegment_\idxj$
from $\control$ to $\control'$.  (Note, by assumption, all $\runsegment_\idxj$
start with some $\cpush{\chb}{\cpdsord}$.)

We now define the extended CPDS $\rightcpds$ that simulates $\cpds$ by keeping
track of stack $\numstacks$ in its stack and using extended rules based on
$\leftcpda$ to simulate parts of the run where the first $(\numstacks-1)$ stacks
are not all empty.  Note, since all
rules operating on $\sbot$ (i.e.
$\cpdsrule{\control}{\sbot}{\genop}{\control'}$) have $\genop =
\cpush{\chb}{\cpdsord}$, rules from $\cpdsrules_1, \ldots,
\cpdsrules_{\numstacks-1}$ may only fire during (or at the start of) the
segments where the first $(\numstacks-1)$ stacks are non-empty (and thus appear
in $\cpdsrules_\genlang$ below).

\begin{definition}[$\rightcpds$]
    Given an $\cpdsord$-OCPDS $\cpds = \tuple{\controls \times \alphabet,
    \alphabet, \cpdsrules_1, \ldots, \cpdsrules_\numstacks}$ with $\numstacks$
    stacks, we define $\rightcpds$ to be an $\cpdsord$-ECPDS such that
    $\rightcpds = \tuple{\controls, \alphabet, \cpdsrules'}$ where $\cpdsrules'
    = \cpdsrules_\numstacks \cup \cpdsrules_\genlang$ and
    \[
        \cpdsrules_\genlang =
        \setcomp{\cpdsrule{\control}{\cha}{\cpdalang{\leftcpda}{\control_1}{\cha}{\control_2}{\chb}{\idxi}}{\control_2}}{\cha
        \in \alphabet \land
        \cpdsrule{\control}{\sbot}{\cpush{\chb}{\cpdsord}}{\control_1} \in
        \cpdsrules_\idxi \land 1 \leq \idxi < \numstacks}
    \]
\end{definition}

\begin{namedlemma}{lem:ecpds-sim-ocpds}{$\rightcpds$ simulates $\cpds$}
    Given an $\cpdsord$-OCPDS $\cpds$ and control states $\controlin,
    \controlout$, we have  $\config{\controlin}{\stackw} \in
    \prestar{\rightcpds}{\saauta}$, where $\saauta$ is the $\controls$-stack
    automaton accepting only the configuration
    $\config{\controlout}{\emptystack{\cpdsord}}$ iff
    $\config{\controlin}{\emptystack{\cpdsord}, \ldots, \emptystack{\cpdsord},
    \stackw} \cpdstran \cdots \cpdstran
    \config{\controlout}{\emptystack{\cpdsord}, \ldots, \emptystack{\cpdsord}}$.
\end{namedlemma}

Lemma~\ref{lem:ecpds-sim-ocpds} only gives an effective decision procedure if we
can decide $\genlang \cap \apply{\lang}{\taaut{\saauta_\idxi}{\sat}{\sat'}} \neq
\emptyset$ for all rules $\cpdsrule{\control}{\cha}{\genlang}{\control'}$
appearing in $\rightcpds$.  For this, we use a standard product construction
between the $\leftcpda$ associated with $\genlang$, and
$\taaut{\saauta_\idxi}{\sat}{\sat'}$.  This gives an ordered CPDS with
$(\numstacks-1)$ stacks, for which, by induction over the number of stacks,
reachability (and emptiness) is decidable.  Note, the initial transition of the
construction sets up the initial stacks of $\leftcpda$.

\begin{nameddefinition}{def:langcheckcpds}{$\langcheckcpds$}
    Given the non-emptiness problem
    $\cpdalang{\leftcpda}{\control_1}{\cha}{\control_2}{\chb}{\idxi} \cap
    \apply{\lang}{\taaut{\saauta_\idxi}{\sat}{\sat'}} \neq \emptyset$, where
    $\apply{\ctop{1}}{\sat} = \cha$, $\leftcpda = \tuple{\controls \times
    \alphabet, \alphabet, \cpdsrules_1, \ldots, \cpdsrules_{\numstacks-1}}$ and
    $\taaut{\saauta_\idxi}{\sat}{\sat'} = \tuple{\tastates{\saauta_\idxi},
    \genrules, \tadelta, \sat, \sat'}$, we define an $\cpdsord$-OCPDS
    $\langcheckcpds = \tuple{\controls^\emptyset, \alphabet,
    \cpdsrules^\emptyset_1, \ldots, \cpdsrules^\emptyset_\idxi \cup
    \cpdsrules_{I/O}, \ldots, \cpdsrules^\emptyset_{\numstacks-1}}$ where, for
    all $1 \leq \idxi \leq (\numstacks - 1)$,
    \begin{align*}
        \controls^\emptyset &= \set{\control_1, \control_2} \uplus
        \setcomp{\tuple{\control, \sat_1}}{\sat_1 \in \tastates{\saauta_\idxi}
        \land \apply{\tacontrol}{\sat_1} = \control} \ ,  \\
        \cpdsrules_{I/O} &=
        \set{\cpdsrule{\control_1}{\sbot}{\cpush{\chb}{\cpdsord}}{\tuple{\control_1,
        \sat}}} \cup \setcomp{\cpdsrule{\tuple{\control_2,
        \sat}}{\sbot}{\noop}{\control_2}}{\sat \in \tastates{\saauta_\idxi}} \ ,
        \text{ and}\\
        \cpdsrules^\emptyset_\idxi &= \setcomp{\cpdsrule{\tuple{\control,
        \sat_1}}{\chc}{\genop}{\tuple{\control',
        \sat_2}}}{\cpdarule{\tuple{\control,
        \apply{\ctop{1}}{\sat_1}}}{\chc}{\optuple}{\genop}{\tuple{\control',
        \apply{\ctop{1}}{\sat_2}}} \in \cpdsrules_\idxi \land \tuple{\sat_1,
        \optuple, \sat_2} \in \sadelta}  
    \end{align*}
\end{nameddefinition}

\begin{namedlemma}{lem:ocpds-lang-emp}{Language Emptiness for OCPDS}
    We have $\cpdalang{\leftcpda}{\control_1}{\cha}{\control_2}{\chb}{\idxi}
    \cap \apply{\lang}{\taaut{\saauta_\idxi}{\sat}{\sat'}} \neq \emptyset$ iff,
    in $\langcheckcpds$ from Definition~\ref{def:langcheckcpds}, we have that
    $\config{\control_2}{\emptystack{\cpdsord}, \ldots \emptystack{\cpdsord}}$
    is reachable from $\config{\control_1}{\emptystack{\cpdsord}, \ldots,
    \emptystack{\cpdsord}}$.
\end{namedlemma}

\subparagraph*{Global Reachability}

We sketch a solution to the global reachability problem, giving a full proof in
\shortlong{the full paper}{Appendix~\ref{sec:ordered-appendix}}.  From
\reflemma{lem:ecpds-sim-ocpds} we gain a representation $\saauta_\numstacks =
\prestar{\rightcpds}{\saauta}$ of the set of configurations
$\config{\control}{\emptystack{\cpdsord}, \ldots, \emptystack{\cpdsord},
\stackw_\numstacks}$ that have a run to
$\config{\controlout}{\emptystack{\cpdsord}, \ldots, \emptystack{\cpdsord}}$.
Now take any $\config{\control}{\emptystack{\cpdsord}, \ldots,
\emptystack{\cpdsord}, \stackw_{\numstacks - 1}, \stackw_\numstacks}$ that
reaches $\config{\controlout}{\emptystack{\cpdsord}, \ldots,
\emptystack{\cpdsord}}$.  The run must pass some
$\config{\control'}{\emptystack{\cpdsord}, \ldots, \emptystack{\cpdsord},
\stackw'_\numstacks}$ with $\config{\control'}{\stackw'_\numstacks}$ accepted by
$\saauta_\numstacks$.  From the product construction above, one can (though not
immediately) extract a tuple $\tuple{\control, \saauta_{\numstacks-1},
\saauta'_\numstacks}$ such that $\stackw_{\numstacks-1}$ is accepted by
$\saauta_{\numstacks-1}$ and $\stackw_\numstacks$ is accepted by
$\saauta'_\numstacks$.  We repeat this reasoning down to stack $1$ and obtain a
tuple of the form $\tuple{\control, \saauta_1, \ldots, \saauta_\numstacks}$.  We
can only obtain a finite set of tuples in this manner, giving a solution to the
global reachability problem.

\section{Scope-Bounded CPDS} 
\label{sec:scope-reach}

Recently, scope-bounded multi-pushdown systems were introduced~\cite{lTN11} and
their reachability problem was shown to be decidable.  Furthermore, reachability
for scope- and phase-bounding was shown to be incomparable~\cite{lTN11}.  Here
we consider scope-bounded CPDS.  

A run $\runsegment = \runsegment_1\ldots\runsegment_\numof$ of an MCPDS is
\emph{context-partitionable} when, for each $\runsegment_\idxi$, if a transition
in $\runsegment_\idxi$ is via $\cpdsruler \in \cpdsrules_\idxj$ on stack
$\idxj$, then all transitions of $\runsegment_\idxi$ are via rules in
$\cpdsrules_\idxj$ on stack $\idxj$.  A \emph{round} is a context-partitioned
run $\runsegment_1\ldots\runsegment_\numstacks$, where during
$\runsegment_\idxi$ only $\cpdsrules_\idxi$ is used.  A
\emph{round-partitionable} run can be partitioned
$\runsegment_1\ldots\runsegment_\numof$ where each $\runsegment_\idxi$ is a
round.  A run of an SBCPDS is such that any character or stack removed from a
stack must have been created at most $\numscopes$ rounds earlier.  For this, we
define pop- and collapse-rounds for stacks.  That is, we mark each stack and
character with the round in which it was created.  When we copy a stack via
$\scopy{\opord}$, the pop-round of the new copy of the stack is the current
round.  However, all stacks and characters within the copy of $\stacku$ keep the
same pop- and collapse-round as in the original $\stacku$.

E.g. take $\sbrac{\stacku}{2}$ where $\stacku = \sbrac{\cha \chb}{1}$, $\stacku$
and $\cha$ have pop-round $2$, and $\chb$ has pop-round $1$.  Suppose in round
$3$ we use $\scopy{2}$ to obtain $\sbrac{\stacku \stacku}{2}$.  The new copy of
$\stacku$ has pop-round $3$ (the current round), but the $\cha$ and $\chb$
appearing in the copy of $\stacku$ still have pop-rounds $2$ and $1$
respectively.  If the scope-bound is $2$, the latest each $\cha$ and the
original $\stacku$ could be popped is in round $4$, but the new $\stacku$ may be
popped in round $5$.  

We will write $\sbstack{\sbpr}{\stackw}$ for a stack $\stackw$ with pop-round
$\sbpr$ and $\sbchar{\sbpr}{\sbcr}{\cha}$ for a character with pop-round $\sbpr$
and collapse-round $\sbcr$.  Pop- and collapse-rounds will be sometimes omitted
for clarity.  Note, the outermost stack will always have pop-round $0$.  In
particular, for all $\ccompose{\stacku}{\opord}{\stackv}$ in the definition
below, the pop-round of $\stackv$ is 0.  

\begin{definition}[Pop- and Collapse-Round]
    Given a round-partitioned run $\runsegment_1 \ldots \runsegment_\numof$ we
    define inductively the pop- and collapse-rounds.  The pop- and
    collapse-round of each stack and character in the first configuration of
    $\runsegment_1$ is $0$.  Take a transition $\config{\control}{\stackw}
    \cpdstran \config{\control'}{\stackw'}$ with $\config{\control'}{\stackw'}$
    in $\runsegment_\idxz$ via a rule
    $\cpdsrule{\control}{\cha}{\genop}{\control'}$.  If $\genop = \noop$ then
    $\stackw = \stackw'$, otherwise when
    \begin{compactenum}
        \item $\genop = \scopy{\opord}$ and $\stackw =
            \ccompose{\sbstack{\sbpr}{\stacku}}{\opord}{\stackv}$, then
            $\stackw' =
            \ccompose{\sbstack{\idxz}{\stacku}}{\opord}{(\ccompose{\sbstack{\sbpr}{\stacku}}{\opord}{\stackv})}$
            where $\sbstack{\idxz}{\stacku} =
            \sbstack{\idxz}{\sbrac{\sbstack{\sbpr_1}{\stacku_1}\ldots\sbstack{\sbpr_\numof}{\stacku_\numof}}{\opord-1}}$
            when  $\sbstack{\sbpr}{\stacku} =
            \sbstack{\sbpr}{\sbrac{\sbstack{\sbpr_1}{\stacku_1}\ldots\sbstack{\sbpr_\numof}{\stacku_\numof}}{\opord-1}}$.

        \item $\genop = \cpush{\chb}{\opord}$, then $\stackw' =
              \ccompose{\sbchar{\idxz}{\sbcr}{\annot{\chb}{\sbstackbr{\sbpr'}{\stacku}}}}{1}{\stackw}$
              where $\sbstack{\sbpr'}{\stacku} =
              \apply{\ctop{\opord+1}}{\apply{\pop{\opord}}{\stackw}}$ and
              $\sbcr$ is the pop-round of $\apply{\ctop{\opord}}{\stackw}$.
              (Note, when $\opord = \cpdsord$, we know $\sbpr' = 0$ since the
              $\ctop{\cpdsord+1}$ stack is outermost.)

        \item $\genop = \pop{\opord}$, when $\stackw =
              \ccompose{\stacku}{\opord}{\stackv}$ then $\stackw' = \stackv$.

        \item We set
              $\apply{\collapse{\opord}}{\ccompose{\annot{\cha}{\sbstackbr{\sbpr}{\stacku'}}}{1}{\ccompose{\stacku}{(\opord+1)}{\stackv}}}=
              \ccompose{\sbstack{\sbpr}{\stacku'}}{(\opord+1)}{\stackv}$ when
              $\stacku$ is order-$\opord$ and $1 \leq \opord < \cpdsord$; and
              $\apply{\collapse{\cpdsord}}{\ccompose{\annot{\cha}{\sbstackbr{0}{\stacku}}}{1}{\stackv}}
              =  \sbstack{0}{\stacku}$ when $\stacku$ is order-$\cpdsord$.

        \item $\genop = \rew{\chb}$ and $\stackw =
              \ccompose{\sbchar{\sbpr}{\sbcr}{\annot{\cha}{\sbstackbr{\sbpr'}{\stacku}}}}{1}{\stackv}$,
              then $\stackw' =
              \ccompose{\sbchar{\sbpr}{\sbcr}{\annot{\chb}{\sbstackbr{\sbpr'}{\stacku}}}}{1}{\stackv}$.
    \end{compactenum}
\end{definition}

\begin{definition}[Scope-Bounded CPDS]
    A $\numscopes$-scope-bounded $\cpdsord$-CPDS ($\cpdsord$-SBCPDS) $\cpds$ is
    an order-$\cpdsord$ MCPDS whose runs are all runs of $\cpds$ that are
    round-partitionable, that is $\runsegment_1\ldots\runsegment_\numof$, such
    that for all $\idxz$, if a transition in $\runsegment_\idxz$ from
    $\config{\control}{\stackw}$ to $\config{\control'}{\stackw'}$ is 
    \begin{compactenum}
        \item a $\pop{\opord}$ transition with $1 < \opord \leq \cpdsord$ and
            $\stackw = \ccompose{\sbstack{\sbpr}{\stacku}}{\opord}{\stackv}$,
            then $\idxz - \numscopes \leq \sbpr$,
        \item a $\pop{1}$ transition with $\stackw =
            \ccompose{\sbchar{\sbpr}{\sbcr}{\annot{\cha}{\stacku}}}{1}{\stackv}$,
            then $\idxz - \numscopes \leq \sbpr$, or
        \item  a $\collapse{\opord}$ transition with $\stackw =
            \ccompose{\sbchar{\sbpr}{\sbcr}{\annot{\cha}{\stacku}}}{1}{\stackv}$,
            then $\idxz - \numscopes \leq \sbcr$.  
    \end{compactenum}
\end{definition}

La Torre and Napoli's decidability proof for the order-$1$ case already uses the
saturation method~\cite{lTN11}.  However, while La Torre and Napoli use a
forwards-reachability analysis, we must use a backwards analysis.  This is
because the forwards-reachable set of configurations is in general not regular.
We thus perform a backwards analysis for CPDS, resulting in a similar approach.
However, the proofs of correctness of the algorithm are quite different.  

\begin{theorem}[Decidability of Reachability Problems]
    For $\cpdsord$-OCPDSs the control state reachability problem and the global
    control state reachability problem are decidable.
\end{theorem}

In \shortlong{the full paper}{Appendix~\ref{sec:ordered-appendix}} we show our
non-global algorithm requires $\bigo{\exptower{\cpdsord-1}{\numof}}$ space,
where $\numof$ is polynomial in $\numscopes$ and the size of the SBCPDS, and we
have at most $\bigo{\exptower{\cpdsord}{\numof}}$ tuples in the global
reachability solution.  La Torre and Parlato give an alternative control state
reachability algorithm at order-$1$ using \emph{thread interfaces}, which allows
sequentialisation~\cite{lTP12} and should generalise order-$\cpdsord$, but, does
not solve the global reachability problem.

\subparagraph*{Control State Reachability}

Fix initial and target control states $\controlin$ and $\controlout$.  The
algorithm first builds a \emph{reachability graph}, which is a finite graph with
a certain kind of path iff $\controlout$ can be reached from $\controlin$.  To
build the graph, we define layered stack automata. These have states
$\salyrst{\sastate_\control}{\idxi}$ for each $1 \leq \idxi \leq \numscopes$
which represent the stack contents $\idxi$ rounds later.  Thus, a layer
automaton tracks the stack across $\numscopes$ rounds, which allows analysis of
scope-bounded CPDSs.

\begin{definition}[$\numscopes$-Layered Stack Automata]
    A \emph{$\numscopes$-layered stack automaton} is a stack automaton $\saauta$
    such that $\sastates_\cpdsord =
    \setcomp{\salyrst{\sastate_\control}{\idxi}}{\control \in \controls \land 1
    \leq \idxi \leq \numscopes}$.
\end{definition}

A state $\salyrst{\sastate_\control}{\idxi}$ is of layer $\idxi$.  A state
$\sastate'$ labelling $\sastate \satran{\sastate'} \sastateset$ has the same
layer as $\sastate$.  We require that there is no $\sastate \satran{\sastate'}
\sastateset$ with $\sastate'' \in \sastateset$ where $\sastate$ is of layer
$\idxi$ and $\sastate''$ is of layer $\idxj < \idxi$.  Similarly, there is no
$\sastate \satrancol{\cha}{\sastateset_\branch} \sastateset$ with $\sastate' \in
\sastateset \cup \sastateset_\branch$ where $\sastate$ is of layer $\idxi$ and
$\sastate'$ is of layer $\idxj < \idxi$.

Next, we define several operations from which the reachability graph is
constructed.  The $\lsapredecessor{\idxj}$ operation connects stack $\idxj$
between two rounds.  We define for stack $\idxj$ 
\[
    \apply{\lsapredecessor{\idxj}}{\saauta, \sastate_\control,
    \sastate_{\control'}} =
    \apply{\lsasaturate{\idxj}}{\apply{\lsaenvmove}{\apply{\lsashift}{\saauta},
    \salyrst{\sastate_{\control_1}}{1}, \salyrst{\sastate_{\control_2}}{2}}} 
\]
where definitions of $\lsashift$, $\lsaenvmove$ and $\lsasaturate{\idxj}$ are
given in \shortlong{the full paper}{Appendix~\ref{sec:scope-appendix}}.
$\lsashift$ moves transitions in layer $\idxi$ to layer $(\idxi+1)$.  E.g.
$\salyrst{\sastate_\control}{1} \satran{\sastate}
\set{\salyrst{\sastate_{\control'}}{2}}$ would become
$\salyrst{\sastate_\control}{2} \satran{\sastate}
\set{\salyrst{\sastate_{\control'}}{3}}$.  Moreover, transitions involving
states in layer $\numscopes$ are removed.  This is because the stack elements in
layer $\numscopes$ will ``go out of scope''.  $\lsaenvmove$ adds a new
transition (analogously to a
$\cpdsrule{\control_1}{\cha}{\rew{\cha}}{\control_2}$ rule) corresponding to the
control state change from $\control_1$ to $\control_2$ effected by the runs over
the other stacks between the current round and the next (hence layers $1$ and
$2$ in the definition above).  $\lsasaturate{\idxj}$ gets by saturation all
configurations of stack $\idxj$ that can reach via $\cpdsrules_\idxj$ the stacks
accepted from the layer-$1$ states of its argument (i.e. saturation using
initial states $\setcomp{\salyrst{\sastate_\control}{1}}{\control \in
\controls}$, which accept stacks from the next round).

The current layer automaton represents a stack across up to $\numscopes$ rounds.
The predecessor operation adds another round on to the front of this
representation.  A key new insight in our proofs is that if a transition goes to
a layer $\idxi$ state, then it represents part of a run where the stack read by
the transition is removed in $\idxi$ rounds time.  Thus, if we add a transition
at layer $0$ (were it to exist) that depends on a transition of layer
$\numscopes$, then the push or copy operation would have a corresponding pop
$(\numscopes+1)$ scopes away.  Scope-bounding forbids this.

\subparagraph*{The Reachability Graph}

The reachability graph $\reachgraph{\cpds}{\controlout} = \tuple{\vertices,
\edges}$ has vertices $\vertices$ and edges $\edges$.  Firstly, $\vertices$
contains some \emph{initial} vertices $\tuple{\control_0, \saauta_1, \control_1,
\ldots, \control_{\numstacks-1}, \saauta_\numstacks, \control_\numstacks}$ where
$\control_\numstacks = \controlout$, and for all $1 \leq \idxi \leq \numstacks$
we have that $\saauta_\idxi$ is the layer automaton
$\apply{\lsasaturate{\idxi}}{\saauta}$ where for all $\stackw$, $\saauta$
accepts $\config{\control_\idxi}{\stackw}$ from
$\salyrst{\sastate_{\control_\idxi}}{1}$.  Furthermore, we require that there is
some $\stackw$ such that $\config{\control_{\idxi-1}}{\stackw}$ is accepted by
$\saauta_\idxi$ from $\salyrst{\sastate_{\control_\idxi}}{1}$.  That is, there
is a run from $\config{\control_{\idxi-1}}{\stackw}$ to $\control_\idxi$.
Intuitively, initial vertices model the final round of a run to $\controlout$
with context switches at $\control_0, \ldots, \control_\numstacks$.  

The complete set $\vertices$ is the set of all tuples $\tuple{\control_0,
\saauta_1, \control_1, \ldots, \control_{\numstacks-1}, \saauta_\numstacks,
\control_\numstacks}$ where there is some $\stackw$ such that
$\config{\control_{\idxi-1}}{\stackw}$ is accepted by $\saauta_\idxi$ from state
$\salyrst{\sastate_{\control_{\idxi-1}}}{1}$.  To ensure finiteness, we can
bound $\saauta_\idxi$ to at most $\sbmax$ states.  The value of $\sbmax$ is
$\bigo{\exptower{\cpdsord-1}{\numof}}$ where $\numof$ is polynomial in
$\numscopes$ and the size of $\cpds$.  We give a full definition of $\sbmax$ and
proof in \shortlong{the full paper}{Appendix~\ref{sec:scope-appendix}}.

We have an edge from a vertex $\tuple{\control_0, \saauta_1, \ldots,
\saauta_\numstacks, \control_\numstacks}$ to $\tuple{\control'_0, \saauta'_1,
\ldots, \saauta'_\numstacks, \control'_\numstacks}$ whenever
$\control_\numstacks = \control'_0$ and for all $\idxi$ we have $\saauta_\idxi =
\apply{\lsapredecessor{\idxi}}{\saauta'_\idxi, \sastate_{\control_\idxi},
\sastate_{\control'_{\idxi-1}}}$.  An edge means the two rounds can be
concatenated into a run since the control states and stack contents match up.

\begin{namedlemma}{lem:scope-reach-graph}
                  {Simulation by $\reachgraph{\cpds}{\controlout}$}  
    Given a scope-bounded CPDS $\cpds$ and control states $\controlin,
    \controlout$, there is a run of $\cpds$ from $\config{\controlin}{\stackw_1,
    \ldots, \stackw_\numstacks}$ to $\config{\controlout}{\stackw'_1, \ldots,
    \stackw'_\numstacks}$ for some $\stackw'_1, \ldots, \stackw'_\numstacks$ iff
    there is a path in $\reachgraph{\cpds}{\controlout}$ to a vertex
    $\tuple{\control_0, \saauta_1, \ldots, \saauta_\numstacks,
    \control_\numstacks}$ with $\control_0 = \controlin$ from an initial vertex where for all
    $\idxi$ we have $\config{\control_{\idxi-1}}{\stackw_\idxi}$ accepted from
    $\salyrst{\sastate_{\control_\idxi}}{1}$ of $\saauta_\idxi$.
\end{namedlemma}

\subparagraph*{Global Reachability}

The $\tuple{\control_0, \saauta_1, \control_1, \ldots,
\control_{\numstacks-1}, \saauta_\numstacks, \control_\numstacks}$ in
$\reachgraph{\cpds}{\controlout}$ reachable from an initial vertex are finite in
number.  We know by Lemma~\ref{lem:scope-reach-graph} that there is such a
vertex accepting all $\config{\control_{\idxi-1}}{\stackw_\idxi}$ iff
$\config{\control_0}{\stackw_1, \ldots, \stackw_\numstacks}$ can reach the
target control state.  Let $\regtuples$ be the set of tuples $\tuple{\control_0,
\saauta_1, \ldots, \saauta_\numstacks}$ for each reachable vertex as above,
where $\saauta_\idxi$ is restricted to the initial state
$\salyrst{\sastate_{\control_{\idxi-1}}}{1}$.  This is a regular solution to the
global control state reachability problem.

\section{Conclusion}

We have shown decidability of global reachability for ordered and scope-bounded
collapsible pushdown systems (and phase-bounded in the \shortlong{full
article}{appendix}).  This leads to a challenge to find a general framework
capturing these systems.  Furthermore, we have only shown upper-bound results.
Although, in the case of phase-bounded systems, our upper-bound matches that of
Seth for CPDSs without collapse~\cite{S09}, we do not know if it is optimal.
Obtaining matching lower-bounds is thus an interesting though non-obvious
problem.  Recently, a more relaxed notion of scope-bounding has been
studied~\cite{lTN12}.  It would be interesting to see if we can extend our
results to this notion.  We are also interested in developing and implementing
algorithms that may perform well in practice.

\vspace{-1ex}

\subparagraph*{Acknowledgments}

Many thanks for initial discussions with Arnaud Carayol and to the
referees for their helpful remarks.  This work was supported by Fond.  Sci.
Math. Paris; AMIS [ANR 2010 JCJC 0203 01 AMIS]; FREC [ANR 2010 BLAN 0202 02
FREC]; VAPF (R\'egion IdF); and the Engineering and Physical Sciences Research
Council [EP/K009907/1].

\vspace{-2ex}

\bibliographystyle{abbrv} 
\bibliography{references.bib}

\newpage
\appendix

\section{Undecidability of MSO Over The Naive Encoding of Order-$2$ Stacks}
\label{sec:mso-undec}

We show that the naive graph representation of an order-$2$ stack leads to the
undecidability of MSO.  By naive graph representation we mean a graph where each
node is a configuration on a run of the CPDS, and we have an edge labelled $S$
between $\configc_1$ and $\configc_2$ if the configurations are neighbouring on
the run.  We have an further edge labelled $1$ if $\configc_2$ was obtained by
popping a character via $\pop{1}$ that was first pushed on to the stack by a
$\cpush{\cha}{\opord}$ at node $\configc_1$.  More formally, we define the
\emph{originating configuration} for each character.

\begin{definition}[Originating Configuration]
    Given a run as a sequence of configurations $\configc_1, \configc_2, \ldots$
    we define inductively the originating configuration of each character.  The
    originating configuration of each character in $\configc_1$ is $1$.  Take a
    transition $\configc_\idxi \cpdstran \configc_{\idxi+1}$ via a rule
    $\cpdsrule{\control}{\cha}{\genop}{\control'}$.  If 
    \begin{enumerate}
        \item $\genop = \scopy{\opord}$, then each character copied inherits its
              originating configuration from the character it is a copy of.  All
              other characters keep the same originating configuration.

        \item $\genop = \cpush{\chb}{\opord}$, all characters maintain the same
              originating configuration except the new $\chb$ character that has
              originating configuration $\idxi$.

        \item $\genop = \rew{\chb}$, all characters maintain the same
              originating configuration except the new $\chb$ character that has
              the originating configuration of the $\cha$ character it is
              replacing.

        \item $\genop = \noop, \pop{\opord}$ or $\collapse{\opord}$, all
              originating configurations are inherited from the previous stack.
    \end{enumerate}
\end{definition}

Thus, from a run $\configc_1, \configc_2, \ldots$ we define a graph
$\tuple{\vertices, \edges_1, \edges_2}$ with vertices $\vertices =
\set{\configc_1, \configc_2, \ldots}$ and edge sets $\edges_1$ and $\edges_2$,
where $\edges_1 = \setcomp{\tuple{\configc_\idxi, \configc_{\idxi+1}}}{1 \leq
\idxi}$ and $\edges_2$ contains all pairs $\tuple{\configc_\idxi,
\configc_\idxj}$ where $\configc_\idxj$ was obtained by a $\pop{1}$ from
$\configc_{\idxj-1}$ and the originating configuration of the character removed
is $\idxi$.

Now, consider the CPDS generating the following run 
\[
    \begin{array}{l}
        \config{p_0}{\sbrac{\sbrac{\sbot}{1}}{2}} \cpdstran
        \config{p_1}{\sbrac{\sbrac{\cha\sbot}{1}}{2}} \cpdstran
        \config{p_2}{\sbrac{\sbrac{\cha\sbot}{1}\sbrac{\cha\sbot}{1}}{2}}
        \cpdstran \config{p_2}{\sbrac{\sbrac{\sbot}{1}\sbrac{\cha\sbot}{1}}{2}}
        \cpdstran  \\

        \ \\
        
        \config{p_0}{\sbrac{\sbrac{\cha\sbot}{1}}{2}} \cpdstran
        \config{p_1}{\sbrac{\sbrac{\cha\cha\sbot}{1}}{2}} \cpdstran
        \config{p_2}{\sbrac{\sbrac{\cha\cha\sbot}{1}\sbrac{\cha\cha\sbot}{1}}{2}}
        \cpdstran
        \config{p_2}{\sbrac{\sbrac{\cha\sbot}{1}\sbrac{\cha\cha\sbot}{1}}{2}}
        \cpdstran \\

        \config{p_2}{\sbrac{\sbrac{\sbot}{1}\sbrac{\cha\cha\sbot}{1}}{2}}
        \cpdstran \\
        
        \ \\ 
        
        \config{p_0}{\sbrac{\sbrac{\cha\cha\sbot}{1}}{2}} \cpdstran
        \config{p_1}{\sbrac{\sbrac{\cha\cha\cha\sbot}{1}}{2}} \cpdstran
        \config{p_2}{\sbrac{\sbrac{\cha\cha\cha\sbot}{1}\sbrac{\cha\cha\cha\sbot}{1}}{2}}
        \cpdstran
        \config{p_2}{\sbrac{\sbrac{\cha\cha\sbot}{1}\sbrac{\cha\cha\cha\sbot}{1}}{2}}
        \\

        \cpdstran
        \config{p_2}{\sbrac{\sbrac{\cha\sbot}{1}\sbrac{\cha\cha\cha\sbot}{1}}{2}}
        \cpdstran
        \config{p_2}{\sbrac{\sbrac{\sbot}{1}\sbrac{\cha\cha\cha\sbot}{1}}{2}}
        \cpdstran \\

        \ \\

        \config{p_0}{\sbrac{\sbrac{\cha\cha\cha\sbot}{1}}{2}} \cpdstran \cdots \
        .
    \end{array}
\]
That is, beginning at $\config{p_0}{\emptystack{2}}$ the CPDS pushes an $\cha$
character, copies the stack with a $\scopy{2}$ and removes all $\cha$s.  After
all $\cha$s are removed, it performs $\pop{2}$ the obtain the stack below
containing only $\cha$.  It pushes another $\cha$ onto the stack and repeats
this process.  After each $\pop{2}$ it adds one more $\cha$ character, performs
a $\scopy{2}$, pops all $\cha$s and so on.  This produces the graph shown below
with $\edges_1$ represented with solid lines, and $\edges_2$ with dashed lines.
Furthermore, nodes from which an $\cha$ is pushed are the target of a dashed
arrow, and nodes reached by popping an $\cha$ are the sources of dashed arrows.

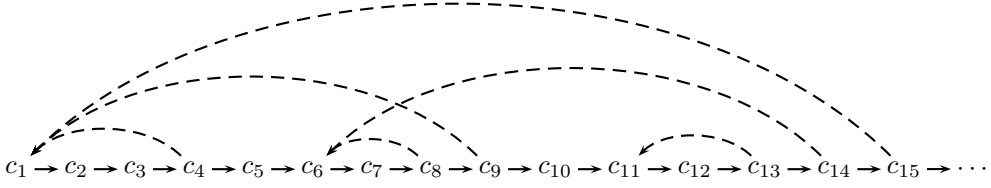
\newcommand\confnode[1]{\rnode{N#1}{$c_{#1}$}} 
\begin{center}
    \begin{psmatrix}[nodealign=true,rowsep=15ex]
        \\

        \multido{\i=1+1}{15}{%
            \confnode{\i} \quad
        }%
        \rnode{end}{$\cdots$}
        \psset{nodesep=.5ex}
        \multido{\i=1+1,\I=2+1}{14}{%
            \ncline{->}{N\i}{N\I}
        }%
        \ncline{->}{N15}{end}

        \psset{linestyle=dashed,arcangle=-45}
        \ncarc{->}{N4}{N1}
        \ncarc{->}{N9}{N1}
        \ncarc{->}{N15}{N1}

        \ncarc{->}{N8}{N6}
        \ncarc{->}{N14}{N6}

        \ncarc{->}{N13}{N11}
    \end{psmatrix}
    \vspace{2ex}
\end{center}

In this graph we can interpret the infinite half-grid.  We restrict the graph to
nodes that are the source of a dashed arrow.  We define horizontal and vertical
edges to obtain the grid below.

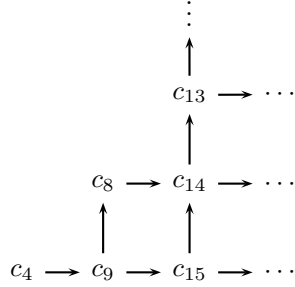
\begin{center}
\begin{psmatrix}[nodealign=true,rowsep=5ex,colsep=5ex]

    & & \rnode{N44}{$\vdots$} & \\
    & & \rnode{N34}{$\configc_{13}$} & \rnode{end3}{$\cdots$} \\
    & \rnode{N23}{$\configc_8$} & \rnode{N24}{$\configc_{14}$} &
    \rnode{end2}{$\cdots$} \\

    \rnode{N12}{$\configc_4$} & \rnode{N13}{$\configc_9$} &
    \rnode{N14}{$\configc_{15}$} & \rnode{end1}{$\cdots$} 

    \psset{nodesep=1ex}

    \ncline{->}{N12}{N13}
    \ncline{->}{N13}{N14}
    \ncline{->}{N23}{N24}
    \ncline{->}{N13}{N23}
    \ncline{->}{N14}{N24}
    \ncline{->}{N24}{N34}

    \ncline{->}{N14}{end1}
    \ncline{->}{N24}{end2}
    \ncline{->}{N34}{end3}
    \ncline{->}{N34}{N44}
\end{psmatrix}
\vspace{1ex}
\end{center}
There is a vertical edge from $\configc$ to $\configc'$ whenever
$\tuple{\configc', \configc} \in \edges_1$.  There is a horizontal edge from
$\configc$ to $\configc'$ whenever we have $\configc''$ such that
\begin{enumerate}
    \item $\tuple{\configc'', \configc} \in
          \edges_2$ and $\tuple{\configc'', \configc'} \in \edges_2$, and

    \item there is a path in $\edges_1$ from $\configc$ to $\configc'$, and

    \item there is no $\configc'''$ on the above path with $\tuple{\configc'',
    \configc'''} \in \edges_2$.
\end{enumerate}
Thus, we can MSO-interpret the infinite half-grid, and hence MSO is undecidable
over this graph.

This naive encoding contains basic matching information about pushes and pops.
It remains an interesting open problem to obtain an encoding of CPDS that is
amenable to MSO based frameworks that give positive decidability results for
concurrent behaviours.

\section{Definition of The Saturation Function}
\label{sec:saturation-defs}

We first introduce two more short-hand notation for sets of transitions.  

The first is a variant on the long-form transitions.  E.g. for the run in
Section~\ref{sec:preliminaries} we can write
$\satranfullk{\sastate_3}{\sastate_1}{\sastateset_2, \sastateset_3}$ to
represent the use of $\sastate_3 \satran{\sastate_2} \sastateset_3$ and
$\sastate_2 \satran{\sastate_1} \sastateset_2$ as the first two transitions in
the run.  That is, for a sequence  $\sastate \satran{\sastate_{\opord-1}}
\sastateset_\opord, \sastate_{\opord-1} \satran{\sastate_{\opord-2}}
\sastateset_{\opord-1}, \ldots, \sastate_{\opord'} \satran{\sastate_{\opord'-1}}
\sastateset_{\opord'}$ in $\sadelta_\opord$ to $\sadelta_{\opord'}$
respectively, we write
$\satranfullk{\sastate}{\sastate_{\opord'-1}}{\sastateset_{\opord'}, \ldots,
\sastateset_\opord}$.

The second notation represents sets of long-form transitions.  We write
$\satranfull{\sastateset}{\cha}{\sastateset_\branch}{\sastateset_1,\ldots,\sastateset_\opord}$
if there is a set $\set{\sat_1, \ldots, \sat_\numof}$ of long-form transitions
such that $\sastateset = \set{\sastate_1,\ldots,\sastate_\numof}$ and for all $1
\leq \idxi \leq \numof$ we have $\sat_\idxi =
\satranfull{\sastate_\idxi}{\cha}{\sastateset^\idxi_\branch}{\sastateset^\idxi_1,\ldots,\sastateset^\idxi_\opord}$
and $\sastateset_\branch = \bigcup_{1 \leq \idxi \leq \numof}
\sastateset^\idxi_\branch \subseteq \sastates_{\opord'}$ for some $\opord'$, and
for all $\opord'$, $\sastateset_{\opord'} = \bigcup_{1 \leq \idxi \leq \numof}
\sastateset^\idxi_{\opord'}$.

\begin{nameddefinition}{def:satstepaux}
                       {The Auxiliary Saturation Function $\auxsat{\cpdsruler}$}

    For a consuming CPDS rule $\cpdsruler =
    \cpdsrule{\control}{\cha}{\genop}{\control'}$ we define for a given stack
    automaton $\saauta$, the set $\apply{\auxsat{\cpdsruler}}{\saauta}$ to be
    the smallest set such that, when
    \begin{enumerate}
        \item $\genop = \pop{\opord}$, for each
              $\satranfullk{\sastate_{\control'}}{\sastate_\opord}{\sastateset_{\opord+1},
              \dots, \sastateset_\cpdsord}$ in $\saauta$, the set
              $\apply{\auxsat{\cpdsruler}}{\saauta}$ contains the transition
              $\satranfull{\sastate_{\control}}{\cha}{\emptyset}{\emptyset,
              \ldots, \emptyset, \set{\sastate_\opord}, \sastateset_{\opord+1},
              \ldots, \sastateset_\cpdsord}$,

        \item $\genop = \collapse{\opord}$, when $\opord = \cpdsord$, the set
              $\apply{\auxsat{\cpdsruler}}{\saauta}$ contains
              $\satranfull{\sastate_{\control}}{\cha}{\set{\sastate_{\control'}}}{
              \emptyset, \ldots, \emptyset}$, and when $\opord < \cpdsord$, for
              each transition
              $\satranfullk{\sastate_{\control'}}{\sastate_\opord}{\sastateset_{\opord+1},
              \dots, \sastateset_\cpdsord}$ in $\saauta$, the set
              $\apply{\auxsat{\cpdsruler}}{\saauta}$ contains the transition
              $\satranfull{\sastate_{\control}}{\cha}{\set{\sastate_\opord}}{\emptyset,
              \ldots, \emptyset, \sastateset_{\opord+1}, \ldots,
              \sastateset_\cpdsord}$, 
    \end{enumerate}

    For a generating CPDS rule $\cpdsruler =
    \cpdsrule{\control}{\cha}{\genop}{\control'}$ we define for a given stack
    automaton $\saauta$ and long-form transition $\sat$ of $\saauta$, the set
    $\apply{\auxsat{\cpdsruler}}{\sat, \saauta}$ to be the smallest set such
    that, when 
    \begin{enumerate}
        \item $\genop = \scopy{\opord}$, $\sat =
              \satranfull{\sastate_{\control'}}{\cha}{\sastateset_\branch}{\sastateset_1,
              \ldots,\sastateset_\opord,\ldots, \sastateset_\cpdsord}$ and
              $\satranfull{\sastateset_\opord}{\cha}{\sastateset'_\branch}{\sastateset'_1,
              \ldots, \sastateset'_\opord}$ is in $\saauta$, the set
              $\apply{\auxsat{\cpdsruler}}{\sat, \saauta}$ contains the
              transition
              \[
                  \satranfull{\sastate_{\control}}{\cha}{\sastateset_\branch
                  \cup \sastateset'_\branch}{\sastateset_1 \cup \sastateset'_1,
                  \ldots, \sastateset_{\opord-1} \cup \sastateset'_{\opord-1},
                  \sastateset'_\opord, \sastateset_{\opord+1}, \ldots,
                  \sastateset_\cpdsord} \ , 
              \]

        \item $\genop = \cpush{\chb}{\opord}$,  for all transitions $\sat =
              \satranfull{\sastate_{\control'}}{\chb}{\sastateset_\branch}{\sastateset_1,
              \ldots, \sastateset_\cpdsord}$ and $\sastateset_1
              \satrancol{\cha}{\sastateset'_\branch} \sastateset'_1$ is in
              $\saauta$ with $\sastateset_\branch \subseteq \sastates_\opord$,
              the set $\apply{\auxsat{\cpdsruler}}{\sat, \saauta}$ contains the
              transition
              \[
                  \satranfull{\sastate_{\control}}{\cha}{\sastateset'_\branch}{\sastateset'_1,
                  \sastateset_2, \ldots, \sastateset_{\opord-1},
                  \sastateset_\opord \cup \sastateset_\branch,
                  \sastateset_{\opord+1}, \ldots, \sastateset_\cpdsord} \ , 
              \]

        \item $\genop = \rew{\chb}$ or $\genop = \noop$, $\sat =
              \satranfull{\sastate_{\control'}}{\chb}{\sastateset_\branch}{\sastateset_1,
              \dots, \sastateset_\cpdsord}$ the set
              $\apply{\auxsat{\cpdsruler}}{\sat, \saauta}$ contains the
              transition
              $\satranfull{\sastate_{\control}}{\cha}{\sastateset_\branch}{\sastateset_1,
              \dots, \sastateset_\cpdsord}$ (where $\chb = \cha$ if $\genop =
              \noop$).
    \end{enumerate}
\end{nameddefinition}

As a remark, omitted from the main body of the paper, during saturation, we add
transitions
$\satranfull{\sastate_\cpdsord}{\cha}{\sastateset_\branch}{\sastateset_1,
\ldots, \sastateset_\cpdsord}$ to the automaton.  Recall this represents a
sequence of transitions $\sastate \satran{\sastate_{\opord-1}}
\sastateset_\opord \in \sadelta_\opord, \sastate_{\opord-1}
\satran{\sastate_{\opord-2}} \sastateset_{\opord-1} \in \sadelta_{\opord-1},
\ldots, \sastate_1 \satrancol{\cha}{\sastateset_\branch} \sastateset_1 \in
\sadelta_1$.  Hence, we first, for each $\cpdsord \geq \opord > 1$, add
$\sastate_\opord \satran{\sastate_{\opord-1}} \sastateset_\opord$ to
$\sadelta_\opord$ if it does not already exist.  Then, we add $\sastate_1
\satrancol{\cha}{\sastateset_\branch} \sastateset_1$ to $\sadelta_1$.  Note, in
particular,  we only add \emph{at most one} $\sastate'$ with $\tuple{\sastate,
\sastate', \sastateset} \in \sadelta_\opord$ for all $\sastate$ and
$\sastateset$.  This ensures termination.

Also, we say a state is \emph{initial} if it is of the form $\sastate_\control
\in \sastateset_{\cpdsord}$ for some control state $\control$ or if it is a
state $\sastate_{\opord} \in \sastateset_{\opord}$ for $\opord < \cpdsord$ such
that there exists a transition $\sastate_{\opord+1} \satran{\sastate_{\opord}}
\sastateset_{\opord+1}$ in $\sadelta_{\opord+1}$.  A pre-condition (that does
not sacrifice generality) of the saturation technique is that there are no
incoming transitions to initial states.

\section{Proofs for Extended CPDS} 
\label{sec:ecpds-appendix}

We provide the proof of~\reftheorem{thm:ecpds-reachability}.  The proof is via
the two lemmas in the sections that follow.  A large part of the proof is
identical to ICALP 2012 and hence not repeated here.

\subsection{Completeness of Saturation for ECPDS}

\begin{namedlemma}{lem:ecpds-completeness}{Completeness of $\satstep$}
    Given an extended CPDS $\cpds$ and an order-$\cpdsord$ stack automaton
    $\saauta_0$, the automaton $\saauta$ constructed by saturation with
    $\satstep$ is such that $\config{\control}{\stackw} \in
    \prestar{\cpds}{\saauta_0}$ implies $\stackw \in
    \slang{\sastate_\control}{\saauta}$.
\end{namedlemma}
\begin{proof}
    We begin with a definition of $\prestar{\cpds}{\saauta_0}$ that permits an
    inductive proof of completeness.  Thus, let $\prestar{\cpds}{\saauta_0} =
    \bigcup\limits_{\ordinal < \omega} \pre{\ordinal}{\cpds}{\saauta_0}$ where
    \[
        \begin{array}{rcl}
            \pre{0}{\cpds}{\saauta_0} &=&
            \setcomp{\config{\control}{\stackw}}{\stackw \in
            \slang{\sastate_\control}{\saauta_0}} \\

            \\ 
            
            \pre{\ordinal+1}{\cpds}{\saauta_0} &=&
            \setcomp{\config{\control}{\stackw}}{ \exists
            \config{\control}{\stackw} \cpdstran \config{\control'}{\stackw'}
            \in \pre{\ordinal}{\cpds}{\saauta_0} }
         \end{array}
    \]
    The proof is by induction over $\ordinal$.  In the base case, we have
    $\stackw \in \slang{\sastate_\control}{\saauta_0}$ and the existence of a
    run of $\saauta_0$, and thus a run in $\saauta$ comes directly from the run
    of $\saauta_0$.  Now, inductively assume $\config{\control}{\stackw}
    \cpdstran \config{\control'}{\stackw'}$ and an accepting run of $\stackw'$
    from $\sastate_{\control'}$ of $\saauta$.
   
    There are two cases depending on the rule used in the transition above.
    Here we consider the case where the rule is of the form
    $\cpdsrule{\control}{\apply{\ctop{1}}{\stackw}}{\genlang}{\control'}$.  The
    case where the rule is a standard CPDS rule is identical to ICALP 2012 and
    hence we do not repeat it here (although a variation of the proof appears in
    the proof of Lemma~\ref{lem:scope-reach-graph-only-if}).

    Take the rule
    $\cpdsrule{\control}{\apply{\ctop{1}}{\stackw}}{\genlang}{\control'}$ and
    the sequence $\cpdsrule{\control_0}{\cha_1}{\genop_1}{\control_1} \ldots,
    \cpdsrule{\control_{\numof-1}}{\cha_\numof}{\genop_\numof}{\control_\numof}
    \in \genlang$ that witnessed the transition, observing that $\control_0 =
    \control$ and $\control_\numof = \control'$.  Now, let $\stackw_\idxi =
    \apply{\genop_\numof}{\cdots\apply{\genop_{\idxi+1}}{\stackw'}}$ for all $0
    \leq \idxi \leq \numof$.  Note, $\stackw = \stackw_0$ and $\stackw' =
    \stackw_\numof$.  
    
    Take $\sat' =
    \satranfull{\sastate_{\control'}}{\chb}{\sastateset_\branch}{\sastateset_1,
    \ldots, \sastateset_\cpdsord}$ to be the first transition on the accepting
    run of $\config{\control'}{\stackw'}$.  Beginning with $\sat_\numof =
    \sat'$, we are going to show that there is a run of
    $\config{\control_\idxi}{\stackw_\idxi}$ beginning with $\sat_\idxi$ and
    thereafter only using transitions appearing in $\saauta$.  Since, by the
    definition of $\satstep$, we add $\sat_0 = \sat$ to $\saauta$, we will
    obtain an accepting run of $\saauta$ for $\config{\control_0}{\stackw_0} =
    \config{\control}{\stackw}$ as required.  We will induct from $\numof$ down
    to $0$.

    The base case $\idxi = \numof$ is trivial, since $\sat_\numof = \sat'$ and
    we already have an accepting run of $\saauta$ over
    $\config{\control_\numof}{\stackw_\numof}$ beginning with $\sat_\numof$.
    Now, assume the case for $\config{\control_\idxi}{\stackw_\idxi}$ and
    $\sat_\idxi$.  We show the case for $\idxi-1$.  Take
    $\tuple{\control_{\idxi-1}, \cha_\idxi, \genop_\idxi, \control_\idxi}$, we
    do a case split on $\genop_\idxi$.  A reader familiar with the saturation
    method for CPDS will observe that the arguments below are very similar to
    the arguments for ordinary CPDS rules.
    \begin{enumerate}
        \item When $\genop_\idxi = \scopy{\opord}$, let $\stackw_{\idxi-1} =
              \ccompose{\stacku_{\opord-1}}{\opord}{\ccompose{\cdots}{\cpdsord}{\stacku_\cpdsord}}$.
              We know 
              \[
                  \stackw_\idxi =
                  \ccompose{\stacku_{\opord-1}}{\opord}{\ccompose{\stacku_{\opord-1}}{\opord}{\ccompose{\stacku_\opord}{(\opord+1)}{\ccompose{\cdots}{\cpdsord}{\stacku_\cpdsord}}}}
                  \ .  
              \] 
              Let $\sat_\idxi =
              \satranfull{\sastate_{\control_\idxi}}{\cha}{\sastateset_\branch}{\sastateset_1,
              \ldots, \sastateset_\opord, \ldots \sastateset_\cpdsord}$ and
              $\satranfull{\sastateset_\opord}{\cha}{\sastateset'_\branch}{\sastateset'_1,
              \ldots, \sastateset'_\opord}$ be the initial transitions used on
              the run of $\stackw_\idxi$ (where the transition from
              $\sastateset_\opord$ reads the second copy of
              $\stacku_{\opord-1}$).  

              From the construction of $\taaut{\saauta}{\sat}{\sat'}$ we have
              have a transition $\sat_{\idxi-1}
              \tatran{\cpdsrule{\control_{\idxi-1}}{\cha_\idxi}{\genop_\idxi}{\control_\idxi}}
              \sat_\idxi$ where
              \[ 
                  \sat_{\idxi-1} =
                  \satranfull{\sastate_{\control_{\idxi-1}}}{\cha}{\sastateset_\branch
                  \cup \sastateset'_\branch}{\sastateset_1 \cup \sastateset'_1,
                  \ldots, \sastateset_{\opord-1} \cup \sastateset'_{\opord-1},
                  \sastateset'_\opord, \sastateset_{\opord+1}, \ldots,
                  \sastateset_\cpdsord} \ .
              \]
              Since we know
              $\ccompose{\stacku_\opord}{(\opord+1)}{\ccompose{\cdots}{\cpdsord}{\stacku_\cpdsord}}$
              is accepted from $\sastateset'_\opord$ via
              $\sastateset_{\opord+1}, \ldots, \sastateset_\cpdsord$, and we
              know that $\stacku_{\opord-1}$ is accepted from $\sastateset_1,
              \ldots, \sastateset_{\opord-1}$ and $\sastateset'_1, \ldots,
              \sastateset'_{\opord-1}$ via $\cha$-transitions labelling
              annotations with $\sastateset_\branch$ and $\sastateset'_\branch$
              respectively, we obtain an accepting run of $\stackw_{\idxi-1}$.

        \item When $\genop_\idxi = \cpush{\chc}{\opord}$, let $\stackw_{\idxi-1}
              =
              \ccompose{\stacku_{\opord-1}}{\opord}{\ccompose{\stacku_\opord}{\opord+1}{\ccompose{\cdots}{\cpdsord}{\stacku_\cpdsord}}}$.
              We know $\stackw_\idxi =
              \apply{\cpush{\chc}{\opord}}{\stackw_{\idxi-1}}$ is 
              \[
                  \ccompose{\annot{\chc}{\stacku_\opord}}{1}{\ccompose{\stacku_{\opord-1}}{\opord}{\ccompose{\cdots}{\cpdsord}{\stacku_\cpdsord}}}
                  \ .
              \]
              Let $\sat_\idxi =
              \satranfull{\sastate_{\control_\idxi}}{\chc}{\sastateset_\branch}{\sastateset_1,
              \ldots, \sastateset_\cpdsord} \quad \text{and} \quad \sastateset_1
              \satrancol{\cha}{\sastateset'_\branch} \sastateset'_1$ be the
              first transitions used on the accepting run of $\stackw_\idxi$.
              The construction of $\taaut{\saauta}{\sat}{\sat'}$ means we have a
              transition $\sat_{\idxi-1}
              \tatran{\cpdsrule{\control_{\idxi-1}}{\cha_\idxi}{\genop_\idxi}{\control_\idxi}}
              \sat_\idxi$ where $\sat_{\idxi-1} =
              \satranfull{\sastate_{\control_{\idxi-1}}}{\cha}{\sastateset'_\branch}{\sastateset'_1,
              \sastateset_2,  \ldots, \sastateset_\opord \cup
              \sastateset_\branch, \ldots, \sastateset_\cpdsord}$.  Thus we can
              construct an accepting run of $\stackw_{\idxi-1}$ (which is
              $\stackw_\idxi$ without the first $\chc$ on top of the top
              order-$1$ stack).  A run from $\sastateset_\opord \cup
              \sastateset_\branch$ exists since $\stacku_\opord$ is also the
              stack annotating $\chc$.

         \item When $\genop_\idxi = \rew{\chc}$ let
               $\satranfull{\sastate_{\control_\idxi}}{\chc}{\sastateset_\branch}{\sastateset_1,
               \ldots, \sastateset_\cpdsord}$ be the first transition on the
               accepting run of $\stackw_\idxi =
               \ccompose{\annot{\chc}{\stacku}}{1}{\stackv}$ for some $\stackv$
               and $\stacku$.  From the construction of
               $\taaut{\saauta}{\sat}{\sat'}$ we know we have a transition
               $\sat_{\idxi-1}
               \tatran{\cpdsrule{\control_{\idxi-1}}{\cha_\idxi}{\genop_\idxi}{\control_\idxi}}
               \sat_\idxi$ where $\sat_{\idxi-1} =
               \satranfull{\sastate_{\control_{\idxi-1}}}{\cha}{\sastateset_\branch}{\sastateset_1,
               \ldots, \sastateset_\cpdsord}$, from which we get an accepting
               run of $\stackw_{\idxi-1} =
               \ccompose{\annot{\cha}{\stacku}}{1}{\stackv}$ as required.  

         \item When $\genop_\idxi = \noop$ let
               $\satranfull{\sastate_{\control_\idxi}}{\cha}{\sastateset_\branch}{\sastateset_1,
               \ldots, \sastateset_\cpdsord}$ be the first transition on the
               accepting run of $\stackw_\idxi =
               \ccompose{\annot{\cha}{\stacku}}{1}{\stackv}$ for some $\stackv$
               and $\stacku$.  From the construction of
               $\taaut{\saauta}{\sat}{\sat'}$ we know we have a transition
               $\sat_{\idxi-1}
               \tatran{\cpdsrule{\control_{\idxi-1}}{\cha_\idxi}{\genop_\idxi}{\control_\idxi}}
               \sat_\idxi$ where $\sat_{\idxi-1} =
               \satranfull{\sastate_{\control_{\idxi-1}}}{\cha}{\sastateset_\branch}{\sastateset_1,
               \ldots, \sastateset_\cpdsord}$, from which we get an accepting
               run of $\stackw_{\idxi-1} =
               \ccompose{\annot{\cha}{\stacku}}{1}{\stackv}$ as required.  
    \end{enumerate}
    Hence, for every $\config{\control}{\stackw} \in \prestar{\cpds}{\saauta_0}$
    we have $\stackw \in \slang{\sastate_\control}{\saauta}$.  
\end{proof}

\subsection{Soundness of Saturation for ECPDS}

As in the previous section, the soundness argument repeats a large part of the
proof given in ICALP 2012.  We first recall the machinery used for soundness,
before giving the soundness proof.

First, assume all stack automata are such that their initial states are not
final.  This is assumed for the automaton $\saauta_{0}$ in and preserved by the
saturation function $\Gamma$.

We assign a ``meaning'' to each state of the automaton.  For this, we define
what it means for an order-$\opord$ stack $\stackw$ to satisfy a state
$q\in\sastates_\opord$, which is denoted $\stackw \models \sastate$.

\begin{definition}[$\stackw \models \sastate$] 
   For any $\sastateset\subseteq \sastates_\opord$ and any order-$\opord$ stack
   $\stackw$, we write $\stackw \models \sastateset$ if $\stackw \models
   \sastate$ for all $\sastate\in\sastateset$, and we define $\stackw \models
   \sastate$ by a case distinction on $\sastate$.  
    \begin{enumerate}
        \item $\sastate$ is an initial state in $\sastates_\cpdsord$. Then for
              any order-$\cpdsord$ stack $w$, we say that $\stackw \models
              \sastate$ if $\config{\sastate}{\stackw} \in
              \prestar{\cpds}{\saauta_0}$.

        \item $\sastate$ is an initial state in $\sastates_\opord$, labeling a
              transition $\sastate_{\opord+1} \satran{\sastate}
              \sastateset_{\opord+1} \in \sadelta_{\opord+1}$. Then for any
              order-$\opord$ stack $\stackw$, we say that  $\stackw \models
              \sastate$ if for all order-$(\opord+1)$ stacks s.t. $\stackv
              \models \sastateset_{\opord+1}$,  then
              $\ccompose{\stackw}{(\opord+1)}{\stackv} \models
              \sastate_{\opord+1}$.  

        \item $\sastate$ is a non-initial state in $\sastates_\opord$. Then for
              any order-$\opord$ stack $w$, we say that $\stackw \models
              \sastate$ if $\saauta_0$ accepts $w$ from $\sastate$.
    \end{enumerate}
\end{definition}

By unfolding the definition, we have that an order-$\opord$ stack
$\stackw_\opord$ satisfies an initial state
$\sastate_\opord\in\sastates_{\opord}$ with
$\satranfullk{\sastate}{\sastate_\opord}{\sastateset_{\opord+1}, \dots,
\sastateset_\cpdsord}$ if for any order-$(\opord+1)$ stack $\stackw_{\opord+1}
\models \sastateset_{\opord+1}$, \ldots, and any order-$\cpdsord$ stack
$\stackw_\cpdsord \models \sastateset_\cpdsord$, we have
$\ccompose{\stackw_\opord}{(\opord+1)}{\ccompose{\cdots}{\cpdsord}{\stackw_\cpdsord}}
\models \sastate$.  

\begin{definition}[Soundness of transitions]
    A transition
    $\satranfull{\sastate}{\cha}{\sastateset_\branch}{\sastateset_1,\ldots,\sastateset_\opord}$
    is sound if for any order-$1$ stack $\stackw_1 \models \sastateset_1$,
    \ldots, and any order-$\opord$ stack $\stackw_\opord \models
    \sastateset_\opord$ and any stack $\stacku \models \sastateset_\branch$, we
    have
    $\ccompose{\annot{\cha}{\stacku}}{1}{\ccompose{\stackw_1}{2}{\ccompose{\cdots}{\opord}{\stackw_\opord}}}
    \models \sastate$.  
\end{definition}

The proof of the following lemma can be found in ICALP 2012~\cite{BCHS12}.

\begin{lemma}[\cite{BCHS12}] \label{lem:soundness-trans-ext}         
    If $\satranfull{\sastate_\control}{\cha}{\sastateset_\branch}{\sastateset_1,
    \ldots, \sastateset_\cpdsord}$ is sound, then any transition
    $\satranfull{\sastate_\opord}{\cha}{\sastateset_\branch}{\sastateset_1,
    \ldots, \sastateset_\opord}$ contained within the transition from
    $\sastate_\control$ is sound.  
\end{lemma}

\begin{definition}[Soundness of stack automata]
    A stack automaton $\saauta$ is sound if the following holds.
    \begin{itemize}
        \item $\saauta$ is obtained from  $\saauta_{0}$ by adding new initial
              states of order $<n$ and transitions starting in an initial state.

        \item In $\saauta$, any transition
              $\satranfull{\sastate}{\cha}{\sastateset_\branch}{\sastateset_1,
              \ldots, \sastateset_\opord}$ for $\opord \leq \cpdsord$ is sound.
    \end{itemize}
\end{definition}

Unsurprisingly, if some order-$\cpdsord$ stack $\stackw$ is accepted by a
\emph{sound} stack automaton $\saauta$ from a state $\sastate_{\control}$ then
$\config{\control}{\stackw}$ belongs to $ \prestar{\cpds}{\saauta_0}$. More
generally, we have the following lemma whose proof can be found in ICALP 2012.

\begin{lemma}[\cite{BCHS12}]
    Let $\saauta$ be a sound stack automaton $\saauta$ and let  $\stackw$ be an
    order-$\opord$ stack. If $\saauta$ accepts $\stackw$ from a state $\sastate
    \in \sastates_{\opord}$ then $\stackw \models \sastate$.  In particular, if
    $\saauta$ accepts an order-$\cpdsord$ stack $\stackw$ from a state
    $\sastate_{\control}\in \sastates_{\cpdsord}$ then
    $\config{\control}{\stackw}$ belongs to $ \prestar{\cpds}{\saauta_0}$.
\end{lemma}

We also recall that the initial automaton $\saauta_{0}$ is sound.

\begin{lemma}[Soundness of $\saauta_{0}$~\cite{BCHS12}] \label{lem:sound-a0}    
    The automaton $\saauta_{0}$ is sound.
\end{lemma}

We are now ready to prove that the soundness of saturation for extended CPDS.

\begin{namedlemma}{lem:ecpds-soundness}{Soundness of $\satstep$}
    The automaton $\saauta$ constructed by saturation with $\satstep$ and
    $\cpds$ from $\saauta_0$ is sound.
\end{namedlemma}
\begin{proof}
    The proof is by induction on the number of iterations of $\satstep$.  The
    base case is the automaton $\saauta_0$ and the result was established in
    Lemma~\ref{lem:sound-a0}.  As in the completeness case, the argument for the
    ordinary CPDS rules is identical to ICALP 2012 and not repeated here
    (although the arguments appear in the proof of
    Lemma~\ref{lem:scope-reach-graph-if}).  
    
    We argue the case for those transitions added because of extended rules
    $\cpdsrule{\control}{\cha}{\genlang}{\control'}$.

    Hence, we consider the inductive step for transitions introduced by extended
    rules of the form $\cpdsrule{\control}{\chc}{\genlang}{\control'}$.  Take
    the $\sat, \sat'$ and
    $\cpdsrule{\control_0}{\cha_1}{\genop_1}{\control_1}\cpdsrule{\control_1}{\cha_2}{\genop_2}{\control_2}\ldots\cpdsrule{\control_{\numof-1}}{\cha_\numof}{\genop_\numof}{\control_\numof}
    \in \genlang \cap \apply{\lang}{\taaut{\saauta_\idxi}{\sat}{\sat'}}$ with
    $\sat'$ being a transition of $\saauta_\idxi$ that led to the introduction
    of $\sat$.  Note $\control = \control_0$ and $\control' = \control_\numof$.      

    Let $\sat_0, \ldots, \sat_\numof$ be the sequence of states on the accepting
    run of $\taaut{\saauta_{\idxi}}{\sat}{\sat'}$.  In particular $\sat_0 =
    \sat$ and $\sat_\numof = \sat'$.  We will prove by induction from $\idxi =
    \numof$ to $\idxi = 0$ that for each $\sat_\idxi$, letting
    \[
        \sat_\idxi =
        \satranfull{\sastate_{\control_\idxi}}{\cha}{\sastateset_\branch}{\sastateset_1,
        \ldots, \sastateset_\cpdsord} \ ,
    \]
    and for all $\stacku \models \sastateset_\branch$, $\stackw_1 \models
    \sastateset_1$, \ldots, $\stackw_\cpdsord \models \sastateset_\cpdsord$ that
    for $\stackw^\idxi =
    \ccompose{\annot{\cha}{\stacku}}{1}{\ccompose{\stackw_1}{2}{\cdots
    \ccompose{}{\cpdsord}{\stackw_\cpdsord}}}$ we have
    $\apply{\genop_\numof}{\cdots\apply{\genop_{\idxi+1}}{\stackw^\idxi}}
    \models \sastate_{\control'}$.  Thus, at $\sat_0 = \sat$ , we have
    $\apply{\genop_\numof}{\cdots\apply{\genop_1}{\stackw^0}} \models
    \sastate_{\control'}$ and thus
    $\config{\control'}{\apply{\genop_\numof}{\cdots\apply{\genop_1}{\stackw^0}}}
    \in \prestar{\cpds}{\saauta_0}$.  Since the above sequence
    \[
        \cpdsrule{\control_0}{\cha_1}{\genop_1}{\control_1}\cpdsrule{\control_1}{\cha_2}{\genop_2}{\control_2}\ldots\cpdsrule{\control_{\numof-1}}{\cha_\numof}{\genop_\numof}{\control_\numof}
    \] 
    is in $\genlang$, we have $\config{\control_0}{\stackw^0} \in
    \prestar{\cpds}{\saauta_0}$ and thus $\stackw^0 \models \sastate_\control$,
    giving soundness of the new transition $\sat_0$.

    The base case is $\sat_\numof = \sat'$.  Since $\sat'$ appears in
    $\saauta_\idxi$, we know it is sound.  That gives us that $\stackw^\numof
    \models \sastate_{\control'}$ as required.  
    
    Now assume that $\sat_\idxi$ satisfies the hypothesis.  We prove that
    $\sat_{\idxi-1}$ does also.  Take the transition $\sat_{\idxi-1}
    \tatran{\cpdsrule{\control_{\idxi-1}}{\cha_\idxi}{\genop_\idxi}{\control_\idxi}}
    \sat_\idxi$.  We perform a case split on $\genop_\idxi$.  Readers familiar
    with ICALP 2012 will notice that the arguments here very much follow the
    soundness proof for ordinary rules.
    \begin{enumerate}
        \item Assume that $\genop_\idxi = \scopy{\opord}$, that we had
              \[ 
                  \sat_\idxi =
                  \satranfull{\sastate_{\control_\idxi}}{\cha}{\sastateset_\branch}{\sastateset_1,
                  \ldots, \sastateset_\cpdsord} \quad \text{and} \quad
                  \satranfull{\sastateset_\opord}{\cha}{\sastateset'_\branch}{\sastateset'_1,
                  \ldots, \sastateset'_\opord} 
              \]
              where the latter set of transition are in $\saauta_\idxi$ and
              therefore sound, and that 
              \[ 
                  \sat_{\idxi-1} =
                  \satranfull{\sastate_{\control_{\idxi-1}}}{\cha}{\sastateset_\branch
                  \cup \sastateset'_\branch}{\sastateset_1 \cup \sastateset'_1,
                  \ldots, \sastateset_{\opord-1} \cup \sastateset'_{\opord-1},
                  \sastateset'_\opord, \sastateset_{\opord+1}, \ldots,
                  \sastateset_\cpdsord} \ .
              \]
              To establish the property for this latter transition, we have to
              prove that for any $\stackw_1 \models \sastateset_1 \cup
              \sastateset'_1, \ldots$, any $\stackw_{\opord-1} \models
              \sastateset_{\opord-1} \cup \sastateset'_{\opord-1}$, any
              $\stackw_\opord \models \sastateset'_\opord,$ any
              $\stackw_{\opord+1} \models \sastateset_{\opord+1}, \ldots$, any
              $\stackw_\cpdsord \models \sastateset_\cpdsord$ and any $\stacku
              \models \sastateset_\branch \cup \sastateset'_\branch$, we have
              for $\stackw^{\idxi-1} =
              \ccompose{\annot{\cha}{\stacku}}{1}{\ccompose{\stackw_1}{2}{\ccompose{\cdots}{\cpdsord}{\stackw_\cpdsord}}}$
              that
              $\apply{\genop_\numof}{\cdots\apply{\genop_\idxi}{\stackw^{\idxi-1}}}
              \models \sastate_{\control'}$. 
              
              Let $\stackv = \apply{\ctop{\opord}}{\stackw^{\idxi-1}} =
              \ccompose{\annot{\cha}{\stacku}}{1}{\ccompose{\stackw_1}{2}{\ccompose{\cdots}{(\opord-1)}{\stackw_{\opord-1}}}}$.

              From the soundness of
              $\satranfull{\sastateset_\opord}{\cha}{\sastateset'_\branch}{\sastateset'_1,
              \ldots, \sastateset'_\opord}$ and as  $\stacku \models
              \sastateset'_\branch, \stackw_1 \models \sastateset'_1, \ldots,
              \stackw_\opord \models \sastateset'_\opord$, we have
              $\ccompose{\stackv}{\opord}{\stackw_\opord} \models
              \sastateset_\opord$.

              Then, from $\stackw_1 \models \sastateset_1, \ldots,
              \stackw_{\opord-1} \models \sastateset_{\opord-1}$, and
              $\ccompose{\stackv}{\opord}{\stackw_\opord} \models
              \sastateset_\opord$, and $\stackw_{\opord+1} \models
              \sastateset_{\opord+1}, \ldots, \stackw_\cpdsord \models
              \sastateset_\cpdsord$ and $\stacku \models \sastateset_\branch$
              and the induction hypothesis for $\sat_\idxi =
              \satranfull{\sastate_{\control_\idxi}}{\cha}{\sastateset_\branch}{\sastateset_1,
              \ldots, \sastateset_\cpdsord}$ we get
              \[
                  \apply{\genop_\numof}{\cdots\apply{\genop_{\idxi+1}}{\apply{\scopy{\opord}}{\stackw}}}
                  =
                  \apply{\genop_\numof}{\cdots\apply{\genop_{\idxi+1}}{\ccompose{\stackv}{\opord}{\ccompose{\stackv}{\opord}{\ccompose{\stackw_\opord}{(\opord+1)}{\ccompose{\cdots}{\cpdsord}{\stackw_\cpdsord}}}}}}
                  \models \sastate_{\control'} 
              \] 
              as required.

        \item Assume that $\genop_\idxi = \cpush{\chb}{\opord}$, that we have
              \[ 
                  \sat_\idxi =
                  \satranfull{\sastate_{\control_\idxi}}{\chb}{\sastateset_\branch}{\sastateset_1,
                  \ldots, \sastateset_\cpdsord} \quad \text{and} \quad
                  \sastateset_1 \satrancol{\cha}{\sastateset'_\branch}
                  (\sastateset'_1)
              \]
              where the latter set of transitions is sound, and that we have
              \[ 
                  \sat_{\idxi-1} =
                  \satranfull{\sastate_{\control_{\idxi-1}}}{\cha}{\sastateset'_\branch}{\sastateset'_1,
                  \sastateset_2, \ldots, \sastateset_\opord \cup
                  \sastateset_\branch, \ldots, \sastateset_\cpdsord} \ .
              \]
              To prove the induction hypothesis for the latter transition, we
              have to prove that for any $\stackw_1 \models \sastateset'_1$, any
              $\stackw_2 \models \sastateset_2, \ldots$, any $\stackw_{\opord-1}
              \models \sastateset_{\opord-1}$, any $\stackw_\opord \models
              \sastateset_\opord \cup \sastateset_\branch$, any
              $\stackw_{\opord+1} \models \sastateset_{\opord+1}, \ldots$, any
              $\stackw_\cpdsord \models \sastateset_\cpdsord$ and any $\stacku
              \models \sastateset'_\branch$, that we have for $\stackw^{\idxi-1}
              =
              \ccompose{\annot{\cha}{\stacku}}{1}{\ccompose{\stackw_1}{2}{\ccompose{\cdots}{\cpdsord}{\stackw_\cpdsord}}}$
              that
              $\apply{\genop_\numof}{\cdots\apply{\genop_\idxi}{\stackw^{\idxi-1}}}
              \models \sastate_{\control'}$.  
              
              From the soundness of
              $\satranfull{\sastateset_1}{\cha}{\sastateset'_\branch}{\sastateset'_1}$
              and as $\stacku \models \sastateset'_\branch$ and $\stackw_1
              \models \sastateset'_1$ we have
              $\ccompose{\annot{\cha}{\stacku}}{1}{\stackw_1} \models
              \sastateset_1$.

              Then, from $\ccompose{\annot{\cha}{\stacku}}{1}{\stackw_1} \models
              \sastateset_1, \stackw_2 \models \sastateset_2, \ldots,
              \stackw_\cpdsord \models \sastateset_\cpdsord$, and
              $\apply{\ctop{\opord+1}}{\apply{\pop{\opord}}{\stackw}} =
              \stackw_\opord \models \sastateset_\branch$, and induction for
              $\sat_\idxi =
              \satranfull{\sastate_{\control_\idxi}}{\chb}{\sastateset_\branch}{\sastateset_1,
              \ldots, \sastateset_\cpdsord}$, we get
              \[
                  \apply{\genop_\numof}{\cdots\apply{\genop_{\idxi+1}}{\apply{\cpush{\chb}{\opord}}{\stackw^{\idxi-1}}}}
                  =
                  \apply{\genop_\numof}{\cdots\apply{\genop_{\idxi+1}}{\ccompose{\annot{\chb}{\stackw_\opord}}{1}{\ccompose{\annot{\cha}{\stacku}}{1}{\ccompose{\stackw_1}{2}{\ccompose{\cdots}{\cpdsord}{\stackw_\cpdsord}}}}}}
                  \models \sastate_{\control'}
              \] 
              as required.

        \item Assume that $\genop = \rew{\chb}$, that we have $\sat_\idxi =
              \satranfull{\sastate_{\control_\idxi}}{\chb}{\sastateset_\branch}{\sastateset_1,
              \dots, \sastateset_\cpdsord}$ and that 
              \[
                  \sat_{\idxi-1} =
                  \satranfull{\sastate_{\control}}{\cha}{\sastateset_\branch}{\sastateset_1,
                  \dots, \sastateset_\cpdsord} \ .
              \]
              To prove the hypothesis for this later transition, we have to
              prove that for any  $\stackw_1 \models \sastateset_1, \ldots,$ for
              any $\stackw_\cpdsord \models \sastateset_\cpdsord$ and any
              $\stacku \models \sastateset_\branch$, we have that for
              $\stackw^{\idxi-1} =
              \ccompose{\annot{\cha}{\stacku}}{1}{\ccompose{\stackw_1}{2}{\ccompose{\cdots}{\cpdsord}{\stackw_\cpdsord}}}$
              we have
              $\apply{\genop_\numof}{\cdots\apply{\genop_\idxi}{\stackw^{\idxi-1}}}
              \models \sastate_{\control'}$. 
              
              From $\stackw_1 \models \sastateset_1, \ldots, \stackw_\cpdsord
              \models \sastateset_\cpdsord$, and $\stacku \models
              \sastateset_\branch$, and the hypothesis for $\sat_\idxi =
              \satranfull{\sastate_{\control_\idxi}}{\chb}{\sastateset_\branch}{\sastateset_1,
              \dots, \sastateset_\cpdsord}$, we get
              \[
                  \apply{\genop_\numof}{\cdots\apply{\genop_{\idxi+1}}{\apply{\rew{\chb}}{\stackw^{\idxi-1}}}}
                  =
                  \apply{\genop_\numof}{\cdots\apply{\genop_{\idxi+1}}{\ccompose{\annot{\chb}{\stacku}}{1}{\ccompose{\stackw_1}{2}{\ccompose{\cdots}{\cpdsord}{\stackw_\cpdsord}}}}}
                  \models \sastate_{\control'}
              \] 
              as required.

        \item Assume that $\genop = \noop$, that we have $\sat_\idxi =
              \satranfull{\sastate_{\control_\idxi}}{\chb}{\sastateset_\branch}{\sastateset_1,
              \dots, \sastateset_\cpdsord}$ and that 
              \[
                  \sat_{\idxi-1} =
                  \satranfull{\sastate_{\control}}{\cha}{\sastateset_\branch}{\sastateset_1,
                  \dots, \sastateset_\cpdsord} \ .
              \]
              To prove the hypothesis for this later transition, we have to
              prove that for any  $\stackw_1 \models \sastateset_1, \ldots,$ for
              any $\stackw_\cpdsord \models \sastateset_\cpdsord$ and any
              $\stacku \models \sastateset_\branch$, we have that for
              $\stackw^{\idxi-1} =
              \ccompose{\annot{\cha}{\stacku}}{1}{\ccompose{\stackw_1}{2}{\ccompose{\cdots}{\cpdsord}{\stackw_\cpdsord}}}$
              we have
              $\apply{\genop_\numof}{\cdots\apply{\genop_\idxi}{\stackw^{\idxi-1}}}
              \models \sastate_{\control'}$. 
              
              From $\stackw_1 \models \sastateset_1, \ldots, \stackw_\cpdsord
              \models \sastateset_\cpdsord$, and $\stacku \models
              \sastateset_\branch$, and the hypothesis for $\sat_\idxi =
              \satranfull{\sastate_{\control_\idxi}}{\cha}{\sastateset_\branch}{\sastateset_1,
              \dots, \sastateset_\cpdsord}$, we get
              \[
                  \apply{\genop_\numof}{\cdots\apply{\genop_{\idxi+1}}{\apply{\rew{\cha}}{\stackw^{\idxi-1}}}}
                  =
                  \apply{\genop_\numof}{\cdots\apply{\genop_{\idxi+1}}{\ccompose{\annot{\cha}{\stacku}}{1}{\ccompose{\stackw_1}{2}{\ccompose{\cdots}{\cpdsord}{\stackw_\cpdsord}}}}}
                  \models \sastate_{\control'}
              \] 
              as required.
    \end{enumerate}
    This completes the proof.    
\end{proof}

\subsection{Complexity of Saturation for ECPDS}

We argue that saturation for ECPDS maintains the same complexity as saturation
for CPDS.

\begin{proposition}
    The saturation construction for an order-$\cpdsord$ CPDS $\cpds$ and an
    order-$\cpdsord$ stack automaton $\saauta_0$ runs in $\cpdsord$-EXPTIME.
\end{proposition} 
\begin{proof}
    The number of states of $\saauta$ is bounded by
    $\exptower{(\cpdsord-1)}{\numof}$ where $\numof$ is the size of $\cpds$ and
    $\saauta_0$:  each state in $\sastates_\opord$ was either in $\saauta_0$ or
    comes from a transition in $\sadelta_{\opord+1}$.  Since the automata are
    alternating, there is an exponential blow up at each order except at
    order-$\cpdsord$.  Each iteration of the algorithm adds at least one new
    transition.  Only $\exptower{\cpdsord}{\numof}$ transitions can be added.
\end{proof}

The complexity can be reduced by a single exponential when runs of the stack
automata are ``non-alternating at order-$\cpdsord$''.  In this case an
exponential is avoided by only adding a transition
$\satranfull{\sastate_\control}{\cha}{\sastateset_\branch}{\sastateset_1,
\ldots, \sastateset_\cpdsord}$ when $\sastateset_\cpdsord$ contains at most one
element.

We refer the reader to ICALP 2012 for a full discussion of non-alternation since
it relies on the original notion of collapsible pushdown system that we have not
defined here.  ICALP 2012 describes the connection between our notion of CPDS
(using annotations) and the original notion, as well as defining non-alternation
at order-$\cpdsord$ and arguing completeness for the restricted saturation step.
It is straightforward to extend this proof to include ECPDS as in the proof of
\reflemma{lem:ecpds-completeness} above.


\section{Definitions and Proofs for Multi-Stack CPDS}
\label{sec:multi-appendix}

\subsection{Multi-Stack Collapsible Pushdown Automata}

We formally define mutli-stack collapsible pushdown automata.

\begin{definition}[Multi-Stack Collapsible Pushdown Automata]
    An order-$\cpdsord$ \emph{multi-stack collapsible pushdown automaton
    ($\cpdsord$-OCPDA)} over input alphabet $\oalphabet$ is a tuple $\cpds =
    \tuple{\controls, \alphabet, \cpdsrules_1, \ldots, \cpdsrules_\numstacks}$
    where $\controls$ is a finite set of control states, $\alphabet$ is a finite
    stack alphabet, $\oalphabet$ is a finite set of output symbols, and for each
    $1 \leq \idxi \leq \numstacks$ we have a set of rules $\cpdsrules_\idxi
    \subseteq \controls \times \alphabet \times \oalphabet \times
    \cops{\cpdsord} \times \controls$.
\end{definition}

Configurations of an OCPDA are defined identically to configurations for OCPDS.
We have a transition
\[
    \config{\control}{\stackw_1,\ldots,\stackw_\numstacks} \cpdatran{\ocha}
    \config{\control'}{\stackw_1, \ldots, \stackw_{\idxi-1}, \stackw'_\idxi,
    \stackw_{\idxi+1}, \ldots, \stackw_\numstacks}
\]
whenever $\cpdsruler = \cpdarule{\control}{\cha}{\ocha}{\genop}{\control'} \in
\cpdsrules_\idxi$ with $\cha = \apply{\ctop{1}}{\stackw}$, $\stackw'_\idxi =
\apply{\genop}{\stackw_\idxi}$.

\subsection{Regular Sets of Configurations}

We prove several properties about \refdefinition{def:regular-set-multi}. 

\begin{property}
    Regular sets of configurations of a multi-stack CPDS 
    \begin{enumerate}
        \item \label{item:bool-alg} form an effective boolean algebra, 
        
        \item \label{item:emptiness} the emptiness problem is decidable in
              PSPACE, 
              
        \item \label{item:membership} the membership problem is decidable in
              linear time.  
    \end{enumerate}
\end{property}
\begin{proof}
    We first prove $(\ref{item:bool-alg})$.  We recall from~\cite{BCHS12} that
    stack automata form an effective boolean algebra.  Given two regular sets
    $\regtuples_1$ and $\regtuples_2$, we can form $\regtuples = \regtuples_1
    \cup \regtuples_2$ as the simple union of the two sets of tuples.  We obtain
    the intersection of $\regtuples_1$ and $\regtuples_2$ by defining
    $\regtuples = \regtuples_1 \cap \regtuples_2$ via a product construction.
    That is, 
    \[
        \regtuples = \setcomp{
                         \tuple{\control, \saauta_1 \cap \saauta'_1, \ldots,
                         \saauta_\numstacks \cap \saauta'_\numstacks}
                     }{\begin{array}{c}        
                         \tuple{\control, \saauta_1, \ldots, \saauta_\numstacks}
                         \in \regtuples_1 \ \land \\

                         \tuple{\control, \saauta'_1, \ldots,
                         \saauta'_\numstacks} \in \regtuples_2
                     \end{array}} \ .
    \]        
    It remains to define the complement $\negation{\regtuples}$ of a set
    $\regtuples$.  Let $\regtuples = \regtuples_1 \cup \cdots \cup
    \regtuples_\numof$ where each $\regtuples_\idxi$ is a singleton set of
    tuples.  Observe that $\negation{\regtuples} = \negation{\regtuples_1}
    \cap \cdots \cap \negation{\regtuples_\numof}$.  Hence, we define for a
    singleton $\regtuples_\idxi$ its complement $\negation{\regtuples_\idxi}$.
    Let $\saauta$ be a stack automaton accepting all stacks.  Furthermore, let
    $\regtuples_\idxi$ contain only $\tuple{\control, \saauta_1, \ldots,
    \saauta_\numof}$.  We define
    \[
        \begin{array}{rcl}
            \negation{\regtuples_\idxi} &=& \setcomp{\tuple{\control', \saauta,
            \ldots, \saauta}}{\control \neq \control' \in \controls} \ \cup \\

            & & \setcomp{\tuple{\control, \saauta, \ldots, \saauta,
            \negation{\saauta_\idxj}, \saauta, \ldots, \saauta}}{1 \leq \idxj
            \leq \numstacks} \ .
        \end{array}
    \]
    That is, either the control state does not match, or at least one of the
    $\numstacks$ stacks does not match.

    We now prove $(\ref{item:emptiness})$.  We know from~\cite{BCHS12} that the
    emptiness problem for a stack automaton is PSPACE.  By checking all tuples
    to find some tuple $\tuple{\control, \saauta_1, \ldots,
    \saauta_\numstacks}$ such that $\saauta_\idxi$ is non-empty for all $\idxi$,
    we have a PSPACE algorithm for determining the emptiness of a regular set
    $\regtuples$.

    Finally, we show $(\ref{item:membership})$, recalling from~\cite{BCHS12}
    that the membership problem for stack automata is linear time.  To check
    whether $\config{\control}{\stackw_1, \ldots, \stackw_\numstacks}$ is
    contained in $\regtuples$ we check each tuple $\tuple{\control, \saauta_1,
    \ldots, \saauta_\numstacks} \in \regtuples$ to see if $\stackw_\idxi$ is
    contained in $\saauta_\idxi$ for all $\idxi$.  This requires linear time.
\end{proof}

\section{Proofs for Ordered CPDS}
\label{sec:ordered-appendix}

\subsection{Proofs for Simulation by $\rightcpds$}

We prove \reflemma{lem:ecpds-sim-ocpds} via
Lemma~\ref{lem:ecpds-sim-ocpds-only-if} and Lemma~\ref{lem:ecpds-sim-ocpds-if}
below.

\begin{lemma}\label{lem:ecpds-sim-ocpds-only-if}
    Given an $\cpdsord$-OCPDS $\cpds$ and control states $\controlin,
    \controlout$, we have 
    \[
        \config{\controlin}{\emptystack{\cpdsord}, \ldots,
        \emptystack{\cpdsord}, \stackw} \cpdstran \cdots \cpdstran
        \config{\controlout}{\emptystack{\cpdsord}, \ldots,
        \emptystack{\cpdsord}} \ .
    \]
    only if $\config{\controlin}{\stackw} \in \prestar{\rightcpds}{\saauta}$,
    where $\saauta$ is the $\controls$-stack automaton accepting only the
    configuration $\config{\controlout}{\emptystack{\cpdsord}}$.
\end{lemma}
\begin{proof}
    Take such a run
    \[
        \config{\controlin}{\emptystack{\cpdsord}, \ldots,
        \emptystack{\cpdsord}, \stackw} \cpdstran \cdots \cpdstran
        \config{\controlout}{\emptystack{\cpdsord}, \ldots,
        \emptystack{\cpdsord}} 
    \]
    of $\cpds$.  Observe that the run can be partitioned into
    $\altrunsegment_0\runsegment_1\altrunsegment_1\ldots\runsegment_\numof\altrunsegment_\numof$
    where during each $\altrunsegment_\idxi$, the first $(\numstacks - 1)$
    stacks are $\emptystack{\cpdsord}$, and, during each $\runsegment_\idxi$,
    there is at least one stack in the first $(\numstacks - 1)$ stacks that is
    not $\emptystack{\cpdsord}$.  Let $\control^1_\idxi$ be the control state of
    the first configuration of $\altrunsegment_\idxi$, $\control^2_\idxi$ be the
    control state in the final configuration of $\altrunsegment_\idxi$,
    $\control^3_\idxi$ be the control state at the beginning of each
    $\runsegment_\idxi$, and $\control^4_\idxi$ be the control state at the end
    of each $\runsegment_\idxi$.  Note, $\control^4_\numof = \controlout$ and
    $\control^1_1 = \controlin$.  Next, let $\cpdsruler_\idxi$ be the rule fired
    between the final configuration of $\altrunsegment_{\idxi-1}$ and the first
    configuration of $\runsegment_\idxi$ (if it exists).  Finally, let
    $\stackw_\idxi$ be the contents of stack $\numstacks$ in the final
    configuration of each $\altrunsegment_\idxi$.  Note $\stackw_\numof =
    \stackw$.

    We proceed by backwards induction from $\idxi = \numof$ down to $\idxi = 0$.
    Trivially it is the case that $\config{\control^4_\numof}{\stackw_\numof}
    \in \prestar{\rightcpds}{\saauta}$.

    In the inductive step, first assume
    $\config{\control^4_\idxi}{\stackw_\idxi} \in
    \prestar{\rightcpds}{\saauta}$.  We have the final configuration of
    $\altrunsegment_\idxi$ is $\config{\control^4_\idxi}{\emptystack{\cpdsord},
    \ldots, \emptystack{\cpdsord}, \stackw_\idxi}$.  Let
    $\config{\control^3_\idxi}{\emptystack{\cpdsord}, \ldots,
    \emptystack{\cpdsord}, \stackw'}$ be the first configuration of
    $\altrunsegment_\idxi$.  Note, since we assume all rules of the form
    $\cpdsrule{\control_1}{\sbot}{\genop}{\control_2}$ have $\genop =
    \cpush{\cha}{\cpdsord}$ for some $\cha$, and during $\altrunsegment_\idxi$
    the first $(\numstacks - 1)$ stacks are empty, we know that no rule from
    $\cpdsrules_1, \ldots, \cpdsrules_{\numstacks - 1}$ was used during
    $\altrunsegment_\idxi$.  Thus, $\altrunsegment_\idxi$ is a run of
    $\rightcpds$ using only rules from $\cpdsrules_\numstacks$.  Hence, we have
    $\config{\control^3_\idxi}{\stackw'} \in \prestar{\rightcpds}{\saauta}$.  

    Now consider $\runsegment_\idxi$ with
    $\config{\control^3_\idxi}{\emptystack{\cpdsord}, \ldots,
    \emptystack{\cpdsord}, \stackw'}$ appended to the end.  Suppose we have that
    $\cpdsruler_{\idxi-1} =
    \cpdsrule{\control^4_{\idxi-1}}{\sbot}{\cpush{\chb}{\cpdsord}}{\control^1_\idxi}
    \in \cpdsrules_\idxj$.  We thus have a run 
    \[
        \config{\control^1_\idxi}{\stackw'_1, \ldots, \stackw'_{\numstacks-1},
        \stackw_{\idxi-1}} \cpdatran{\cpdsruler^1} \cdots
        \cpdatran{\cpdsruler^{\numof-1}} \config{\control^2_\idxi}{\stackw''_1,
        \ldots, \stackw''_\numstacks} \cpdatran{\cpdsruler^\numof}
        \config{\control^3_\idxi}{\emptystack{\cpdsord}, \ldots,
        \emptystack{\cpdsord}, \stackw'}
    \]
    where $\stackw'_\idxj =
    \apply{\cpush{\chb}{\cpdsord}}{\emptystack{\cpdsord}}$ and
    $\stackw'_{\idxj'} = \emptystack{\cpdsord}$ for all $\idxj' \neq \idxj$.
    Since it is not the case that the first $(\numstacks - 1)$ stacks are empty,
    we know that only generating rules from $\cpdsrules_\numstacks$ can be used
    during this run.  Let $\apply{\ctop{1}}{\stackw_{\idxi-1}} = \cha$.  From
    this run we can immediately project a sequence
    $\cpdsrule{\control^0}{\cha^1}{\genop^1}{\control^1}\cpdsrule{\control^1}{\cha^2}{\genop^2}{\control^2}\ldots\cpdsrule{\control^{\numof'-1}}{\cha^\numof}{\genop^{\numof'}}{\control^{\numof'}}
    \in
    \cpdalang{\leftcpda}{\control^1_\idxi}{\cha}{\control^3_\idxi}{\chb}{\idxj}$
    such that we have $\stackw' =
    \apply{\genop^{\numof'}}{\cdots\apply{\genop^1}{\stackw_{\idxi-1}}}$,
    $\control^0 = \control^1_\idxi$ and $\control^{\numof'} = \control^3_\idxi$.
    Since we have $\config{\control^3_\idxi}{\stackw'} \in
    \prestar{\rightcpds}{\saauta}$ and a rule
    $\cpdsrule{\control^4_{\idxi-1}}{\cha}{\cpdalang{\leftcpda}{\control^1_\idxi}{\cha}{\control^3_\idxi}{\chb}{\idxj}}{\control^3_\idxi}$
    in $\rightcpds$, we thus have
    $\config{\control^4_{\idxi-1}}{\stackw_{\idxi-1}} \in
    \prestar{\rightcpds}{\saauta}$ as required.
    
    Hence, when $\idxi = 0$, we have $\config{\controlin}{\stackw} \in
    \prestar{\rightcpds}{\saauta}$, completing the proof.    
\end{proof}

\begin{lemma}\label{lem:ecpds-sim-ocpds-if}
    Given an $\cpdsord$-OCPDS $\cpds$ and control states $\controlin,
    \controlout$, we have 
    \[
        \config{\controlin}{\emptystack{\cpdsord}, \ldots,
        \emptystack{\cpdsord}, \stackw} \cpdstran \cdots \cpdstran
        \config{\controlout}{\emptystack{\cpdsord}, \ldots,
        \emptystack{\cpdsord}} \ .
    \]
    whenever $\config{\controlin}{\stackw} \in \prestar{\rightcpds}{\saauta}$,
    where $\saauta$ is the $\controls$-stack automaton accepting only the
    configuration $\config{\controlout}{\emptystack{\cpdsord}}$.
\end{lemma}
\begin{proof}
    Since $\config{\controlin}{\stackw} \in \prestar{\rightcpds}{\saauta}$ we
    have a run of $\rightcpds$ of the form
    $\runsegment_1\ldots\runsegment_\numof$ where the rules used to connect the
    last configuration of $\runsegment_\idxi$ to $\runsegment_{\idxi+1}$ are of
    the form $\cpdsrule{\control'_\idxi}{\cha}{\genlang}{\control_{\idxi + 1}}$
    and no other rules of this form are used otherwise.  Thus, let
    $\control'_\idxi$ denote the control state at the end of $\runsegment_\idxi$
    and $\control_\idxi$ denote the control state in the first configuration of
    $\runsegment_\idxi$.  Similarly, let $\stackw'_\idxi$ denote the stack
    contents at the end of $\runsegment_\idxi$ and $\stackw_\idxi$ the stack
    contents at the beginning.

    We proceed by induction from $\idxi = \numof$ down to $\idxi = 1$.  In the
    base case, we immediately have a run from
    $\config{\control_\numof}{\emptystack{\cpdsord}, \ldots,
    \emptystack{\cpdsord}, \stackw_\numof}$ to
    $\config{\control'_\numof}{\emptystack{\cpdsord}, \ldots,
    \emptystack{\cpdsord}}$.  Now, assume the we have a run from
    $\config{\control'_\idxi}{\emptystack{\cpdsord}, \ldots,
    \emptystack{\cpdsord}, \stackw'_\idxi}$ to the final configuration.  Since
    we have a run to this configuration from
    $\config{\control_\idxi}{\stackw_\idxi}$ to
    $\config{\control'_\idxi}{\stackw'_\idxi}$ in $\rightcpds$ that uses only
    ordinary rules, we can execute the same run from
    $\config{\control_\idxi}{\emptystack{\cpdsord}, \ldots,
    \emptystack{\cpdsord}, \stackw_\idxi}$ to reach
    $\config{\control'_\idxi}{\emptystack{\cpdsord}, \ldots,
    \emptystack{\cpdsord}, \stackw'_\idxi}$.

    Now consider the rule
    $\cpdsrule{\control'_{\idxi-1}}{\cha}{\genlang}{\control_\idxi}$ that
    connects $\runsegment_{\idxi-1}$ and $\runsegment_\idxi$.  We have $\genlang
    = \cpdalang{\leftcpda}{\control^1_\idxi}{\cha}{\control_\idxi}{\chb}{\idxj}$
    for some $\control^1_\idxi$, $\chb$ and $\idxj$, and there is a rule
    $\cpdsrule{\control'_{\idxi-1}}{\sbot}{\cpush{\chb}{\cpdsord}}{\control^1_\idxi}
    \in \cpdsrules_\idxj$ of $\cpds$.  Furthermore, there is a sequence
    $\cpdsrule{\control^0}{\cha^1}{\genop^1}{\control^1}\cpdsrule{\control^1}{\cha^2}{\genop^2}{\control^2}\ldots\cpdsrule{\control^{\numof'-1}}{\cha^\numof}{\genop^{\numof'}}{\control^{\numof'}}
    \in \genlang$ such that $\stackw_\idxi =
    \apply{\genop^{\numof'}}{\cdots\apply{\genop^1}{\stackw'_{\idxi-1}}}$,
    $\control^0 = \control^1_\idxi$, and $\control^{\numof'} = \control_\idxi$.
  
    From the definition of $\leftcpda$, this sequence immediately describes a
    run
    \[
        \begin{array}{rcl}
            \config{\control'_{\idxi-1}}{\emptystack{\cpdsord}, \ldots,
            \emptystack{\cpdsord}, \stackw'_{\idxi-1}} &\cpdstran&
            \config{\control^1_\idxi}{\emptystack{\cpdsord}, \ldots,
            \apply{\cpush{\chb}{\cpdsord}}{\emptystack{\cpdsord}}, \ldots,
            \emptystack{\cpdsord}, \stackw'_{\idxi-1}} \\
            
            &\cpdstran& \cdots \\
            
            &\cpdstran& \config{\control_\idxi}{\emptystack{\cpdsord}, \ldots,
            \emptystack{\cpdsord}, \stackw_\idxi} 
        \end{array}
    \]
    of $\cpds$.  Thus we have a run from
    $\config{\control'_{\idxi-1}}{\emptystack{\cpdsord}, \ldots,
    \emptystack{\cpdsord}, \stackw'_{\idxi-1}}$ to the final configuration, to
    complete the inductive case.

    Finally, when $\idxi = 1$, we repeat the first half of the argument above to
    obtain a run from $\config{\control_1}{\emptystack{\cpdsord}, \ldots,
    \emptystack{\cpdsord}, \stackw_1}$, and since $\control_1 = \controlin$ and
    $\stackw_1 = \stackw$ we have a run of $\cpds$ as required.  
\end{proof}

\subsection{Proofs for Language Emptiness for OCPDS} 

We prove \reflemma{lem:ocpds-lang-emp} below.

\begin{proof}
    By standard product construction arguments, a run of $\langcheckcpds$ can be
    projected into runs of $\leftcpda$ and $\taaut{\saauta_\idxi}{\sat}{\sat'}$
    and vice-versa.  We need only note that in any control state
    $\tuple{\control, \sat_1}$ of $\langcheckcpds$, the corresponding state in
    $\leftcpda$ is always $\tuple{\control, \apply{\ctop{1}}{\sat_1}}$.
\end{proof}

\subsection{Global Reachability}

We provide an inductive proof of global reachability for ordered CPDS.

\begin{proof} 
    Take $\saauta_\numstacks = \prestar{\rightcpds}{\saauta}$ from
    \reflemma{lem:ecpds-sim-ocpds}.  Furthermore, let $\saauta_\sbot$ be the
    stack automaton accepting only $\emptystack{\cpdsord}$ from its initial
    state.  For each control state $\control$, we have that $\tuple{\control,
    \saauta_\sbot, \ldots, \saauta_\sbot, \saauta_\numstacks}$ represents all
    configurations $\config{\control}{\emptystack{\cpdsord}, \ldots,
    \emptystack{\cpdsord}, \stackw_\numstacks}$ for which there is a run to
    $\config{\controlout}{\emptystack{\cpdsord}, \ldots, \emptystack{\cpdsord}}$
    when $\saauta_\numstacks$ is restricted to have initial state
    $\sastate_\control$.  

    Hence, inductively assume for $\idxi+1$ that we have a finite set of tuples
    $\regtuples$ such that for each configuration
    $\config{\control}{\emptystack{\cpdsord}, \ldots, \emptystack{\cpdsord},
    \stackw_{\idxi+1}, \ldots, \stackw_\numstacks}$ for which there is a run to
    $\config{\controlout}{\emptystack{\cpdsord}, \ldots, \emptystack{\cpdsord}}$
    there is a tuple $\tuple{\control, \saauta_\sbot, \ldots, \saauta_\sbot,
    \saauta_{\idxi+1}, \ldots, \saauta_\numstacks}$ such that $\stackw_\idxj$ is
    accepted by $\saauta_\idxj$ for each $\idxj$.

    Now consider any configuration $\config{\control}{\emptystack{\cpdsord},
    \ldots, \emptystack{\cpdsord}, \stackw_{\idxi}, \ldots, \stackw_\numstacks}$
    that can reach the final configuration.  We know the run goes via some
    $\config{\control'}{\emptystack{\cpdsord}, \ldots, \emptystack{\cpdsord},
    \stackw'_{\idxi+1}, \ldots, \stackw'_\numstacks}$ accepted by some tuple
    $\tuple{\control', \saauta_\sbot, \ldots, \saauta_\sbot, \saauta_{\idxi+1},
    \ldots, \saauta_\numstacks} \in \regtuples$.  Furthermore, we know from the
    proof of correctness of the extended saturation algorithm, that there is a
    run of the $\idxi$ stack OCPDS $\langcheckcpds$ from
    $\config{\tuple{\control, \sat_{\idxi+1}, \ldots,
    \sat_\numstacks}}{\emptystack{\cpdsord}, \ldots, \emptystack{\cpdsord},
    \stackw_{\idxi}}$ to $\config{\tuple{\control', \sat'_{\idxi+1}, \ldots,
    \sat'_\numstacks}}{\emptystack{\cpdsord}, \ldots, \emptystack{\cpdsord}}$
    where 
    \begin{enumerate}
        \item $\sat'_\idxj$ is the initial transition of $\saauta_\idxj$
              accepting $\stackw'_\idxj$, and 
              
        \item the sequence of stack operations to the $\idxj$th stack $\genop_1,
              \ldots, \genop_\numof$ connected to this run give $\stackw'_\idxj =
              \apply{\genop_\numof}{\cdots\apply{\genop_1}{\stackw_\idxj}}$, and

        \item $\stackw_\idxj$ can be accepted by first taking transition
              $\sat_\idxj$ and thereafter only transitions in $\saauta_\idxj$.  
    \end{enumerate}

    Thus, let $\saauta_\idxi$ be $\prestar{\langcheckcpds}{\saauta}$ where
    $\saauta$ accepts $\config{\tuple{\control', \sat'_{\idxi+1}, \ldots,
    \sat'_\numstacks}}{\emptystack{\cpdsord}}$.  Restrict $\saauta_\idxi$ to
    have initial state $\sastate_{\tuple{\control, \sat_{\idxi+1}, \ldots,
    \sat_\numstacks}}$ and let $\saauta^{\sat_\idxj}_\idxj$ be the automaton
    $\saauta_\idxj$ with the transition $\sat_\idxj$ added from a new state,
    which is designated as the initial state.  Thus, for each configuration
    $\config{\control}{\emptystack{\cpdsord}, \ldots, \emptystack{\cpdsord},
    \stackw_\idxi, \ldots, \stackw_\numstacks}$, there is a tuple
    $\tuple{\control, \saauta_\sbot, \ldots, \saauta_\sbot, \saauta_\idxi,
        \saauta^{\sat_{\idxi+1}}_{\idxi + 1}, \ldots,
    \saauta^{\sat_\numstacks}_\numstacks}$ such that $\stackw_\idxi$ is accepted
    by $\saauta_\idxi$ and $\stackw_\idxj$ is accepted by
    $\saauta^{\sat_\idxj}_\idxj$ for all $\idxj > \idxi$.  This results in a
    finite set of tuples $\regtuples'$ satisfying the induction hypothesis.

    Thus, after $\idxi = 1$ we obtain a finite set of tuples $\regtuples$ of the
    form $\tuple{\control, \saauta_1, \ldots, \saauta_\numstacks}$ representing
    all configurations that can reach
    $\config{\controlout}{\emptystack{\cpdsord}, \ldots,
    \emptystack{\cpdsord}}$, as required.
\end{proof}

\subsection{Complexity}

Assume $\cpdsord > 1$.  Our control state reachability algorithm requires
$\exptower{\numstacks(\cpdsord-1)}{\numof}$ time, where $\numof$ is polynomial
in the size of the OCPDS.  Beginning with stack $\numstacks$, the saturation
algorithm can add at most $\bigo{\exptower{\cpdsord-1}{\numof}}$ transitions
over the same number of iterations.  Each of these iterations may require
analysis of some $\langcheckcpds$ which has
$\bigo{\exptower{\cpdsord-1}{\numof}}$ control states and thus the
stack-automaton constructed by saturation over $\langcheckcpds$ may have up to
$\bigo{\exptower{2(\cpdsord-1)}{\numof}}$ transitions.  By continuing in this
way, we have at most $\bigo{\exptower{(\numstacks-1)(\cpdsord-1)}{\numof}}$
control states when there is only one stack remaining, and thus the number of
transitions, and the total running time of the algorithm is
$\bigo{\exptower{\numstacks(\cpdsord-1)}{\numof}}$.  This also gives us at most
$\bigo{\exptower{\numstacks\cpdsord}{\numof}}$ tuples in the solution to the
global reachability problem.

\section{Phase-Bounded CPDS} 
\label{sec:phase-appendix}

Phase-bounding~\cite{lTMP07} for multi-stack pushdown systems is a restriction
where each computation can be split into a fixed number of phases.  During each
phase, characters can only be removed from one stack, but push actions may occur
on any stack. 

\begin{definition}[Phase-Bounded CPDS]
    Given a fixed number $\numphases$ of phases, an order-$\cpdsord$
    \emph{phase-bounded CPDS} ($\cpdsord$-PBCPDS) is an $\cpdsord$-MCPDS with
    the restriction that each run $\runsegment$ can be partitioned into
    $\runsegment_1 \ldots \runsegment_\numphases$ and for all $\idxi$, if some
    transition in $\runsegment_\idxi$ by $\cpdsruler \in \cpdsrules_\idxj$ on
    stack $\idxj$ for some $\idxj$ is consuming, then all consuming transitions
    in $\runsegment_\idxi$ are by some $\cpdsruler' \in \cpdsrules_\idxj$ on
    stack $\idxj$.
\end{definition}

We give a direct\footnote{For PDS, phase-bounded reachability can be reduced to
ordered PDS.  We do not know if this holds for CPDS, and prefer instead to give
a direct algorithm.} algorithm for deciding the reachability problem over
phase-bounded CPDSs.  We remark that Seth~\cite{S10} presented a saturation
technique for order-$1$ phase-bounded pushdown systems.  Our algorithm was
developed independently of Seth's, but our product construction can be compared
with Seth's automaton $T_i$.

\begin{theorem}[Decidability of the Reachability Problems]
    For $\cpdsord$-PBCPDSs the control state reachability problem and the global
    control state reachability problem are decidable.
\end{theorem}

In Appendix~\ref{sec:phase-complexity} we show that our control state
reachability algorithm will require
$\bigo{\exptower{\numstacks(\cpdsord-1)}{\numof}}$ time, where $\numof$ is
polynomial in the size of the PBCPDS, and we have at most
$\bigo{\exptower{\numstacks\cpdsord}{\numof}}$ tuples in the solution to the
global reachability problem.

\subparagraph*{Control State Reachability}

A run of the PBCPDS will be $\pbrunseg{1}\ldots\pbrunseg{\numphases}$,  assuming
(w.l.o.g.) that all phases are used.  We can guess (or enumerate) the sequence
$\pbcontrol{0}\pbcontrol{1}\ldots\pbcontrol{\numphases}$ of control states
occurring at the boundaries of each $\pbrunseg{\idxi}$.  That is,
$\pbrunseg{\idxi}$ ends with control state $\pbcontrol{\idxi}$,
$\pbcontrol{\numphases}$ is the target control state, and $\pbcontrol{0}$ is the
initial control state.  We also guess for each $\idxi$, the stack
$\pbstack{\idxi}$ that may perform consuming operations between
$\pbcontrol{\idxi-1}$ and $\pbcontrol{\idxi}$.  Our algorithm iterates from
$\idxi = \numphases$ down to $\idxi = 0$.

We begin with the stack automata $\pbsa{1}{\numphases}, \ldots,
\pbsa{\numstacks}{\numphases}$ which each accept
$\config{\pbcontrol{\numphases}}{\stackw}$ for all stacks $\stackw$.  Note we
can vary these automata to accept any regular set of stacks we wish. 

Thus, $\pbsa{1}{\idxi}, \ldots, \pbsa{\numstacks}{\idxi}$ will characterise a
possible set of stack contents at the end of phase $\idxi$.  We show below how
to construct $\pbsa{1}{\idxi-1}, \ldots, \pbsa{\numstacks}{\idxi-1}$ given
$\pbsa{1}{\idxi}, \ldots, \pbsa{\numstacks}{\idxi}$.  This is repeated until we
have $\pbsa{1}{0}, \ldots, \pbsa{\numstacks}{0}$.  We then check, for each
$\idxj$, that $\config{\pbcontrol{0}}{\emptystack{\cpdsord}}$ is accepted by
$\pbsa{\idxj}{0}$.  This is the case iff we have a positive instance of the
reachability problem.

We construct $\pbsa{1}{\idxi-1}, \ldots, \pbsa{\numstacks}{\idxi-1}$ from
$\pbsa{1}{\idxi}, \ldots, \pbsa{\numstacks}{\idxi}$.  For each $\idxj \neq
\pbstack{\idxi}$ we build $\pbsa{\idxj}{\idxi-1}$ by adding to
$\pbsa{\idxj}{\idxi}$ a brand new set of initial states $\sastate_\control$ and
a guessed transition $\sat_\idxj =
\satranfull{\sastate_{\pbcontrol{\idxi-1}}}{\cha}{\sastateset_\branch}{\sastateset_1,
\ldots, \sastateset_\cpdsord}$ with $\sastateset_\branch, \sastateset_1, \ldots,
\sastateset_\cpdsord$ being states of $\pbsa{\idxj}{\idxi}$ and
$\sastate_{\pbcontrol{\idxi-1}}$ being one of the new states.  The idea is
$\sat_\idxj$ will be the initial transition accepting
$\config{\pbcontrol{\idxi-1}}{\stackw}$ where $\stackw$ is stack $\idxj$ at the
beginning of phase $\idxi$.  By guessing an accompanying $\sat'_\idxj$ of
$\pbsa{\idxj}{\idxi}$ we can build
$\taaut{\pbsa{\idxj}{\idxi}}{\sat_\idxj}{\sat'_\idxj}$ (by instantiating
\refdefinition{def:tranaut} with $\saauta = \pbsa{\idxj}{\idxi}$, $\sat =
\sat_\idxj$ and $\sat' = \sat'_\idxj$) for which there will be an accepting run
if the updates to stack $\idxj$ during phase $\idxi$ are concordant with the
introduction of transition $\sat_\idxj$.

Thus, for each $\idxj \neq \pbstack{\idxi}$ we have $\pbsa{\idxj}{\idxi-1}$ and
$\taaut{\pbsa{\idxj}{\idxi}}{\sat_\idxj}{\sat'_\idxj}$.  We now consider the
$\pbstack{\idxi}$th stack.  We build a CPDS $\pbcpds{\idxi}$ that accurately
models stack $\pbstack{\idxi}$ and tracks each
$\taaut{\pbsa{\idxj}{\idxi}}{\sat_\idxj}{\sat'_\idxj}$ in its control state.  We
ensure that $\pbcpds{\idxi}$ has a run from
$\config{\pbcontrol{\idxi-1}}{\stackw}$ to
$\config{\pbcontrol{\idxi}}{\stackw'}$ for some $\stackw$ and $\stackw'$ iff
there is a corresponding run over the $\pbstack{\idxi}$th stack of $\cpds$ that
updates the remaining stacks $\idxj$ in concordance with each guessed
$\sat_\idxj$.  Thus, we define $\pbsa{\pbstack{\idxi}}{\idxi-1}$ to be the
automaton recognising $\prestar{\pbcpds{\idxi}}{\pbsa{\pbstack{\idxi}}{\idxi}}$
constructed by saturation.  The construction of $\pbcpds{\idxi}$ (given below)
follows the standard product construction of a CPDS with several finite-state
automata.

Note $\pbcpds{\idxi}$ is looking for a run from $\pbcontrol{\idxi-1}$ to
$\pbcontrol{\idxi}$ concordant with runs of $\sat_\idxj$ to $\sat'_\idxj$ for
each $\idxj$.  To let $\pbcpds{\idxi}$ start in $\pbcontrol{\idxi-1}$ and finish
in $\pbcontrol{\idxi}$, we have an initial transition from $\pbcontrol{\idxi-1}$
to $\tuple{\pbcontrol{\idxi-1}, \sat_1, \ldots, \sat_\numstacks}$.  Thereafter,
the components are updated as in a standard product construction.  When
$\tuple{\pbcontrol{\idxi}, \sat'_1, \ldots, \sat'_\numstacks}$ is reached, there
is a final transition to $\pbcontrol{\idxi}$.  To ease notation, we use dummy
variables $\sat_{\pbstack{\idxi}} = \sat'_{\pbstack{\idxi}} =
\tastate^{\pbstack{\idxi}} = \tastate^{\pbstack{\idxi}}_1$ for the transition
automaton component of the $\pbstack{\idxi}$th stack (for which we do not have a
$\sat$ and $\sat'$ to track).

In the definition below, the first line of the definition of $\cpdsrules^\idxi$
gives the initial and final transitions, the second line models rules operating
on stack $\pbstack{\idxi}$, and the final line models generating operations
occurring on the $\idxj$th stack for $\idxj \neq \pbstack{\idxi}$.

\begin{definition}[$\pbcpds{\idxi}$]
    Given for all $1 \leq \idxj \neq \pbstack{\idxi} \leq \numstacks$ a
    transition automaton $\gentaaut_\idxj =
    \taaut{\pbsa{\idxj}{\idxi}}{\sat_\idxj}{\sat'_\idxj}$ and a phase-bounded
    CPDS $\cpds = \tuple{\controls, \alphabet, \cpdsrules_1, \ldots,
    \cpdsrules_\numstacks}$ and control states $\pbcontrol{\idxi-1}$,
    $\pbcontrol{\idxi}$, we define the CPDS $\pbcpds{\idxi} =
    \tuple{\set{\pbcontrol{\idxi-1}, \pbcontrol{\idxi}} \cup \controls^\idxi,
    \cpdsrules^\idxi, \alphabet}$ where, letting $\sat_{\pbstack{\idxi}} =
    \sat'_{\pbstack{\idxi}} = \tastate^{\pbstack{\idxi}} =
    \tastate^{\pbstack{\idxi}}_1$ be dummy transitions for technical
    convenience, and letting $\tastate^\idxj$ for all $\idxj \neq
    \pbstack{\idxi}$ range over all states of $\gentaaut_\idxj$, we have
    \begin{itemize}
        \item $\controls^\idxi$ contains all states $\tuple{\control,
              \tastate^1, \ldots, \tastate^\numstacks}$ where $\control \in
              \controls$, and

        \item the rules $\cpdsrules^\idxi$ of $\pbcpds{\idxi}$ are 
              \[
                  \begin{array}{l}
                      \setcomp{\cpdsrule{\pbcontrol{\idxi-1}}{\cha}{\noop}{\tuple{\pbcontrol{\idxi-1},
                      \sat_1, \ldots, \sat_\numstacks}},
                      \cpdsrule{\tuple{\pbcontrol{\idxi}, \sat'_1, \ldots,
                      \sat'_\numstacks}}{\cha}{\noop}{\pbcontrol{\idxi}}}{\cha
                      \in \alphabet}\ \cup\\

                      \setcomp{\cpdsrule{\tuple{\control, \tastate^1, \ldots,
                      \tastate^\numstacks}}{\cha}{\genop}{\tuple{\control',
                      \tastate^1_1, \ldots, \tastate^\numstacks_1}}}{
                      \begin{array}{l}
                        \cpdsrule{\control}{\cha}{\genop}{\control'} \in
                        \cpdsrules_{\pbstack{\idxi}} \\ 

                        \forall \idxj' \neq \idxj \ . \ \tastate^{\idxj'}
                        \tatran{\cpdsrule{\control}{\_}{\noop}{\control'}}
                        \tastate^{\idxj'}_1 
                      \end{array}} \ \cup \\

                      \setcomp{\cpdsrule{\control_1}{\cha}{\genop}{\control_2}}{
                      \begin{array}{c}
                          \control_1= \tuple{\control, \tastate^1, \ldots,
                          \tastate^\idxj, \ldots \tastate^\numstacks} \land

                          \control_2 = \tuple{\control', \tastate^1_1, \ldots,
                          \tastate^\idxj_1, \ldots \tastate^\numstacks_1} \\

                          \land\ \cpdsrule{\control}{\chb}{\genop}{\control'}
                          \in \cpdsrules_\idxj \land 

                          \tastate^\idxj
                          \tatran{\cpdsrule{\control}{\chb}{\genop}{\control'}}
                          \tastate^\idxj_1 \ \land \\ 

                          \forall \idxj' \neq \idxj \ . \ \tastate^{\idxj'}
                          \tatran{\tuple{\control, \_, \noop, \control'}}
                          \tastate^{\idxj'}_1 
                      \end{array}} \ .
                  \end{array} 
              \]
    \end{itemize}
\end{definition}

We state the correctness of our reduction, deferring the proof to
Appendix~\ref{sec:phase-proof}.

\begin{namedlemma}{lem:phase-reach}{Simulation of a PBCPDS}
    Given a phase-bounded CPDS $\cpds$ control states $\pbcontrol{0}$ and
    $\pbcontrol{\numphases}$, there is a run of $\cpds$ from
    $\config{\pbcontrol{0}}{\stackw_1, \ldots, \stackw_\numstacks}$ to
    $\config{\pbcontrol{\numphases}}{\stackw'_1, \ldots, \stackw'_\numstacks}$
    iff for each $1 \leq \idxj \leq \numstacks$, we have that
    $\config{\pbcontrol{0}}{\stackw_\idxj}$ is accepted by $\pbsa{\idxj}{0}$.
\end{namedlemma}

\subsection{Global Reachability}

$\saauta^1_0, \ldots, \saauta^\numstacks_0$ were obtained by a finite sequence
of non-deterministic choices ranging over a finite number of values.  Let
$\regtuples$ be the therefore finite set of tuples $\tuple{\pbcontrol{0},
\saauta_1, \ldots, \saauta_\numstacks}$ for each sequence as above, where
$\saauta_\idxi$ is $\saauta^\idxi_0$ with initial state
$\sastate_{\pbcontrol{0}}$.  From Lemma~\ref{lem:phase-reach}, we have a regular
solution to the global control state reachability problem as required.

\subsection{Proofs for Control-State Reachability}
\label{sec:phase-proof}

In this section we prove \reflemma{lem:phase-reach} via
Lemma~\ref{lem:phase-reach-only-if} and Lemma~\ref{lem:phase-reach-if} below.

\begin{lemma}\label{lem:phase-reach-only-if}
    Given a phase-bounded CPDS $\cpds$ control states $\pbcontrol{0}$ and
    $\pbcontrol{\numphases}$, there is a run of $\cpds$ from
    $\config{\pbcontrol{0}}{\stackw_1, \ldots, \stackw_\numstacks}$ to
    $\config{\pbcontrol{\numphases}}{\stackw'_1, \ldots, \stackw'_\numstacks}$
    only if for each $1 \leq \idxj \leq \numstacks$, we have that
    $\config{\pbcontrol{0}}{\stackw_\idxj}$ is accepted by $\pbsa{\idxj}{0}$.
\end{lemma}
\begin{proof}
    Take a run of $\cpds$ from $\config{\pbcontrol{0}}{\stackw^1_0, \ldots,
    \stackw^\numstacks_0}$ to
    $\config{\pbcontrol{\numphases}}{\stackw^1_\numphases, \ldots,
    \stackw^\numstacks_\numphases}$ and split it into phases
    $\runsegment_1\ldots\runsegment_\numphases$.  Let $\control_\idxi$ be the
    control state at the end of each $\runsegment_\idxi$, and $\control_0$ be
    the control state at the beginning of $\runsegment_1$.  Similarly, let
    $\stackw^\idxj_\idxi$ be the stack contents of stack $\idxj$ at the end of
    $\runsegment_\idxi$.  We include, for convenience, the transition from the
    end of $\runsegment_\idxi$ to the beginning of $\runsegment_{\idxi+1}$ in
    $\runsegment_{\idxi+1}$.  Thus, the last configuration of
    $\runsegment_\idxi$ is also the first configuration of
    $\runsegment_{\idxi+1}$.

    We proceed by induction from $\idxi = \numphases$ down to $\idxi = 1$.  In
    the base case we know by definition that
    $\config{\control_\numphases}{\stackw^\idxj_\numphases}$ is accepted by
    $\pbsa{\idxj}{\numphases}$.  

    Hence, assume $\config{\control_{\idxi+1}}{\stackw^\idxj_{\idxi+1}}$ is
    accepted by $\pbsa{\idxj}{\idxi+1}$.  We show the case for $\idxi$.  First
    consider $\pbstack{\idxi}$.  Take the run
    \[
        \config{\control_\idxi}{\stackw^1_\idxi, \ldots,
        \stackw^\numstacks_\idxi} \cpdstran \cdots \cpdstran
        \config{\control_{\idxi+1}}{\stackw^1_{\idxi+1}, \ldots,
        \stackw^\numstacks_{\idxi+1}} \ .
    \]
    We want to find a run
    \[
        \config{\control_\idxi}{\stackw^{\pbstack{\idxi}}_\idxi} \cpdstran
        \config{\tuple{\control_\idxi, \sat_1, \ldots,
        \sat_\numstacks}}{\stackw^{\pbstack{\idxi}}_\idxi} \cpdstran \cdots
        \cpdstran \config{\tuple{\control_{\idxi+1}, \sat'_1, \ldots,
        \sat'_\numstacks}}{\stackw^{\pbstack{\idxi}}_{\idxi+1}} \cpdstran
        \config{\control_1}{\stackw^{\pbstack{\idxi}}_{\idxi+1}}
    \]
    of $\pbcpds{\idxi}$, giving us that
    $\config{\control_\idxi}{\stackw^{\pbstack{\idxi}}_\idxi}$ is accepted by
    $\pbsa{\pbstack{\idxi}}{\idxi}$.  This is almost by definition, except we
    need to prove for each $\idxj \neq \pbstack{\idxi}$ that there is a sequence
    $\sat^0, \ldots, \sat^\numof$ that is also the projection of the run of
    $\pbcpds{\idxi}$ to the $(\idxj+1)$th component (that is, the state of the
    $\idxj$th transition automaton).  In particular, we require $\sat^0 =
    \sat_\idxj$ and $\sat^\numof = \sat'_\idxj$.  The proof proceeds in exactly
    the same manner as the case of
    $\cpdsrule{\control}{\cha}{\genlang}{\control'}$ in the proof of
    \reflemma{lem:ecpds-completeness} for ECPDS.  Namely, from the sequence of
    operations $\genop^0, \ldots, \genop^\numof$ taken from the run $\sat^0,
    \ldots, \sat^\numof$, we obtain a sequence of stacks such that at each
    $\idxz$ there is an accepting run of the $\idxz$th stack constructed from
    $\sat^\idxz$ and thereafter only transitions of $\pbsa{\idxj}{\idxi+1}$.
    Thus, since $\sat_\idxj$ is added to $\pbsa{\idxj}{\idxi+1}$ to obtain
    $\pbsa{\idxj}{\idxi}$, we additionally get an accepting run of
    $\pbsa{\idxj}{\idxi}$ over $\config{\control_\idxi}{\stackw^\idxj_\idxi}$.
    We do not repeat the arguments here.

    Finally, then, when $\idxi$ reaches $1$, we repeat the arguments above to
    conclude $\config{\control_0}{\stackw^\idxj_0}$ is accepted by
    $\pbsa{\idxj}{0}$ for each $\idxj$, giving the required lemma.
\end{proof}

\begin{lemma}\label{lem:phase-reach-if}
    Given a phase-bounded CPDS $\cpds$ control states $\pbcontrol{0}$ and
    $\pbcontrol{\numphases}$, there is a run of $\cpds$ from
    $\config{\pbcontrol{0}}{\stackw_1, \ldots, \stackw_\numstacks}$ to
    $\config{\pbcontrol{\numphases}}{\stackw'_1, \ldots, \stackw'_\numstacks}$
    whenever for each $1 \leq \idxj \leq \numstacks$, we have that
    $\config{\pbcontrol{0}}{\stackw_\idxj}$ is accepted by $\pbsa{\idxj}{0}$.
\end{lemma}
\begin{proof}
    Assume for each $1 \leq \idxj \leq \numstacks$, we have that
    $\config{\pbcontrol{0}}{\stackw_\idxj}$ is accepted by $\pbsa{\idxj}{0}$.  

    Thus, we can inductively assume for each $\idxj$ we have
    $\config{\pbcontrol{\idxi}}{\stackw^\idxj_\idxi}$ accepted by
    $\pbsa{\idxj}{\idxi}$ and a run of $\cpds$ of the form
    \[
        \config{\pbcontrol{0}}{\stackw_1, \ldots, \stackw_\numstacks} \cpdstran
        \cdots \cpdstran \config{\pbcontrol{\idxi}}{\stackw^1_\idxi, \ldots,
        \stackw^\numstacks_\idxi} \ .
    \]
    Taking $\stackw^\idxj_0 = \stackw_\idxj$ trivially gives us the base case.
    We prove the case for $(\idxi+1)$.
    
    From the induction hypothesis, we have in particular that
    $\config{\pbcontrol{\idxi}}{\stackw^{\pbstack{\idxi}}_\idxi}$ is accepted by
    $\pbsa{\pbstack{\idxi}}{\idxi}$ and hence we have a run of
    $\pbcpds{\idxi+1}$ of the form
    \[
        \config{\control_\idxi}{\stackw^{\pbstack{\idxi}}_\idxi} \cpdstran
        \config{\tuple{\control_\idxi, \sat_1, \ldots,
        \sat_\numstacks}}{\stackw^{\pbstack{\idxi}}_\idxi} \cpdstran \cdots
        \cpdstran \config{\tuple{\control_{\idxi+1}, \sat'_1, \ldots,
        \sat'_\numstacks}}{\stackw^{\pbstack{\idxi}}_{\idxi+1}} \cpdstran
        \config{\control_1}{\stackw^{\pbstack{\idxi}}_{\idxi+1}}
    \]
    such that $\config{\control_1}{\stackw^{\pbstack{\idxi}}_{\idxi+1}}$ is
    accepted by $\pbsa{\pbstack{\idxi}}{\idxi+1}$.  From this run, due to the
    definition of $\pbcpds{\idxi}$ we can build a run 
    \[
        \config{\control_\idxi}{\stackw^1_\idxi, \ldots,
        \stackw^\numstacks_\idxi} \cpdstran \cdots \cpdstran
        \config{\control_{\idxi+1}}{\stackw^1_{\idxi+1}, \ldots,
        \stackw^\numstacks_{\idxi+1}}
    \]
    of $\cpds$ where for all $\idxj \neq \pbstack{\idxi}$, we define
    $\stackw^\idxj_{\idxi+1} =
    \apply{\genop^\numof}{\cdots\apply{\genop^1}{\stackw^\idxj_\idxi}}$ where
    \[
        \cpdsrule{\control^0}{\cha^1}{\genop^1}{\control^1}\cpdsrule{\control^1}{\cha^2}{\genop^2}{\control^2}\ldots\cpdsrule{\control^{\numof-1}}{\cha^\numof}{\genop^\numof}{\control^\numof} 
    \] 
    is the sequence of labels on the run of
    $\taaut{\pbsa{\idxj}{\idxi}}{\sat_\idxj}{\sat'_\idxj}$.  We have to prove
    for all $\idxj \neq \pbstack{\idxi}$ that
    $\config{\control_{\idxi+1}}{\stackw^\idxj_{\idxi+1}}$ is accepted by
    $\pbsa{\idxj}{\idxi+1}$.  For the proof observe that the introduction of
    $\sat_\idxj$ to $\pbsa{\idxj}{\idxi+1}$ to form $\pbsa{\idxj}{\idxi}$
    followed the saturation technique for extended CPDS for a rule
    $\cpdsrule{\control_\idxi}{\cha}{\genlang}{\control_{\idxi+1}}$ where
    $\genlang$ is the language of possible sequences of the form above.  Thus,
    from the soundness of the saturation method for extended CPDS, we have that
    there must be the required run of $\pbsa{\idxj}{\idxi+1}$ over
    $\config{\control_{\idxi+1}}{\stackw^\idxj_{\idxi+1}}$ beginning with
    transition $\sat'_\idxj$.

    Alternatively, we can argue similarly to the proof of
    \reflemma{lem:ecpds-completeness}, but in the reverse direction.  That is,
    we start with the observation that the accepting run of
    $\config{\control_\idxi}{\stackw^\idxj_\idxi}$ uses $\sat_\idxj = \sat^0$
    for the first transition, and thereafter only transitions from
    $\pbsa{\idxj}{\idxi+1}$.  We prove this by induction for the stack obtained
    by applying $\genop^1$ and $\sat^1$, then for the stack obtained by applying
    $\genop^2$ and $\sat^2$.  This continues until we reach
    $\stackw^\idxj_{\idxi+1}$, and since $\sat^\numof = \sat'_\idxj$ with
    $\sat'_\idxj$ being a transition of $\pbsa{\idxj}{\idxi+1}$, we get the
    accepting run we need.  We remark that this is how the soundness proof for
    the standard saturation algorithm would proceed if we were able to assume
    that each new transition is only used at the head of any new runs the
    transition introduces (but in general this is not the case because new
    transitions may introduce loops).  We leave the construction of this proof
    as an exercise for the interested reader, for which they may follow the
    proof of the extended rule case for \reflemma{lem:ecpds-soundness}.

    Thus, finally, by induction, we obtain a run to
    $\config{\control_\numphases}{\stackw_1, \ldots, \stackw_\numstacks}$ such
    that $\config{\control_\numphases}{\stackw_\idxj}$ is accepted by
    $\pbsa{\idxj}{\numphases}$. 
\end{proof}

\subsection{Complexity}
\label{sec:phase-complexity}

Assume $\cpdsord > 1$.  Our control state reachability algorithm requires
$\exptower{\numphases(\cpdsord-1)}{\numof}$ time, where $\numof$ is polynomial
in the size of the PBCPDS.  Beginning with phase $\numphases$, the saturation
algorithm can add at most $\bigo{\exptower{\cpdsord-1}{\numof}}$ transitions
over the same number of iterations to
$\saauta^{\pbstack{\numphases}}_{\numphases-1}$.  Thus we assume each
$\saauta^\idxj_\idxi$ to have at most $\bigo{\exptower{(\numphases -
\idxi)(\cpdsord-1)}{\numof}}$ transitions.  The largest automaton
$\saauta^\idxj_{\idxi-1}$ construction is when $\idxj = \pbstack{\idxi}$.  For
this we build a CPDS with $\bigo{\exptower{(\numphases -
\idxi)(\cpdsord-1)}{\numof}}$ control states and thus
$\saauta^{\pbstack{\idxi}}_{\idxi-1}$ has at most $\bigo{\exptower{(\numphases -
\idxi + 1)(\cpdsord-1)}{\numof}}$ transitions.  Hence, when $\idxi = 0$, we have
at most $\bigo{\exptower{\numphases(\cpdsord-1)}{\numof}}$ transitions, which
also gives the run time of the algorithm.  This also implies we have at most
$\bigo{\exptower{\numphases\cpdsord}{\numof}}$ tuples in the solution to the
global reachability problem.

\section{Proofs for Scope-Bounded CPDS}
\label{sec:scope-appendix}

\subsection{Operations on Layer Automata}

\subparagraph*{Shift of a Layer Automaton}

The idea behind $\lsashift$ is that all transitions in layer $\idxi$ are moved
up to layer $(\idxi+1)$ and transitions involving states in layer $\numscopes$
are removed.  Intuitively this is because the stack elements in layer
$\numscopes$ will ``go out of scope'' when the context switch corresponding to
the $\lsashift$ occurs.  In more detail, states of layer $\idxi$ are renamed to
become states of layer $(\idxi+1)$, with all states of layer $\numscopes$ being
deleted.  Similarly, all transitions that involved a layer $\numscopes$ state
are also removed.  

We define $\apply{\lsashift}{\saauta}$ of an order-$\cpdsord$ $\numscopes$-layer
stack automaton
\[
    \saauta = \tuple{ 
                  \sastates_\cpdsord,\ldots,\sastates_1, 
                  \alphabet,
                  \sadelta_\cpdsord,\ldots,\sadelta_1,
                  \emptyset ,\ldots, \emptyset
              } 
\] 
to be 
\[
    \saauta' = \tuple{ 
                   \sastates'_\cpdsord,\ldots,\sastates'_1, 
                   \alphabet,
                   \sadelta'_\cpdsord,\ldots,\sadelta'_1,
                   \emptyset ,\ldots, \emptyset
               } 
\]
where defining 
\[
    \apply{\lsashift}{\sastate} = \begin{cases}
    
        \sastate & \text{if $\sastate \in \sastates_\opord$, $\cpdsord > \opord$
        and $\sastate$ is layer $\idxi < \numscopes$} \\ 

        \salyrst{\sastate_\control}{\idxi+1} & \text{if $\sastate =
        \salyrst{\sastate_\control}{\idxi} \in \sastates_\cpdsord$ and $\idxi <
        \numscopes$} \\

        \text{undefined} & \text{otherwise} 
    \end{cases}
\]
and extending $\lsashift$ point-wise to sets of states, we have 
\[
    \sadelta'_\cpdsord = \setcomp{\apply{\lsashift}{\sastate} \satran{\sastate'}
    \apply{\lsashift}{\sastateset}}{\sastate \satran{\sastate'} \sastateset \in
    \sadelta_\cpdsord \text{ and $\sastate$ is layer $\idxi < \numscopes$}} 
\]
and for all $\cpdsord > \opord > 1$ 
\[
    \sadelta'_\opord = \setcomp{\sastate \satran{\sastate'}
    \apply{\lsashift}{\sastateset}}{\sastate \satran{\sastate'} \sastateset \in
    \sadelta_\opord \text{ and $\sastate$ is layer $\idxi < \numscopes$}} 
\]
and
\[
    \sadelta'_1 = \setcomp{\sastate
    \satrancol{\sastate'}{\apply{\lsashift}{\sastateset_\branch}}
    \apply{\lsashift}{\sastateset}}{\sastate
    \satrancol{\sastate'}{\sastateset_\branch} \sastateset \in \sadelta_1 \text{
    and $\sastate$ is layer $\idxi < \numscopes$}} \ .
\]
In all cases above, transitions are only created if the applications of
$\lsashift$ result in a defined state or set of states.  This operation will
erase all layer $\numscopes$ states, and all transitions that go to a layer
$\numscopes$ state.  All other states will be shifted up one layer.  E.g.  layer
$1$ states become layer $2$.

\subparagraph*{Environment Moves}

Given an automaton $\saauta$, define $\apply{\lsaenvmove}{\saauta, \sastate,
\sastate'}$ of an order-$\cpdsord$ $\numscopes$-layer stack automaton to be
$\saauta'$ obtained from $\saauta$ by adding for each transition
$\satranfull{\sastate'}{\cha}{\sastateset_\branch}{\sastateset_1, \ldots,
\sastateset_\cpdsord}$ the transition
$\satranfull{\sastate}{\cha}{\sastateset_\branch}{\sastateset_1, \ldots,
\sastateset_\cpdsord}$.  This operation can be thought of as a saturation rule
that captures the effect of an external context, and could be considered as
rules $\cpdsrule{\control}{\cha}{\noop}{\control'}$ for each $\cha \in
\alphabet$.

\subparagraph*{Saturating a Layer Automaton}

Given a layer automaton $\saauta$, we define
$\apply{\lsasaturate{\idxj}}{\saauta}$ to be the result of applying the
saturation procedure with the CPDS $\tuple{\controls, \alphabet,
\cpdsrules_\idxj}$ and the stack automaton $\saauta$ with initial state-set
$\setcomp{\salyrst{\sastate_\control}{1}}{\control \in \controls}$.

\subsection{Size of the Reachability Graph}

We define $\sbmax$. 

\begin{lemma} \label{lem:sbmax}
    The maximum number of states in any layer automaton constructable by
    repeated applications of $\lsapredecessor{\idxj}$ is
    $\exptower{\cpdsord-2}{\polyof{\numscopes, \sizeof{\controls}}}$ states for
    some computable polynomial $\poly$.
\end{lemma}
\begin{proof} 
    A $\numscopes$-layer automaton may have in $\sastate_\cpdsord$ only the
    states $\salyrst{\sastate_\control}{\idxi}$ for $1 \leq \idxi \leq
    \numscopes$ and $\control \in \controls$, and thus at most $\numscopes
    \sizeof{\controls} = \const$ states.  There may be at most $\const$
    transitions from any state at order-$\cpdsord$ using the restricted
    saturation algorithm where $\sastateset_\cpdsord$ has cardinality $1$ for
    any transition added, and thus at most $\const \cdot \const$ states at
    order-$(\cpdsord-1)$ (noting that the shift operation deletes all states
    that would become non-initial if they were to remain).  
    
    Next, there may be at most $2^{\const \cdot \const}$ transitions from any
    state at order-$(\cpdsord-1)$, and thus at most $\const \cdot \const \cdot
    2^{\const \cdot \const}$ states at order-$(\cpdsord-2)$ (noting that the
    shift operation deletes all states that would become non-initial if they
    were to remain).  
    
    Thus, we can repeat this argument down to order-$1$ and obtain
    $\exptower{\cpdsord-2}{\polyof{\numscopes, \sizeof{\controls}}}$ states for
    some computable polynomial $\poly$.
\end{proof}

Take the automaton accepting any $\config{\control_\idxi}{\stackw}$ from
$\salyrst{\sastate_{\control_\idxi}}{1}$.  This automaton has order-$\cpdsord$
states of the form $\salyrst{\sastate_\control}{\idxi}$, and at most a single
transition from each of the layer $1$ states to $\emptyset$.  Each of these
transitions is labelled by a state with at most one transition to $\emptyset$,
and so on until order-$1$.

\begin{definition}[$\sbmax$]
    Following Lemma~\ref{lem:sbmax}, we take $\sbmax =
    \exptower{\cpdsord-2}{\polyof{\numscopes, \const}}$ for some computable
    polynomial $\poly$.
\end{definition}

\subsection{Proofs for Control State Reachability}

In this section, we prove \reflemma{lem:scope-reach-graph}.  The proof is split
in to two directions, given in Lemma~\ref{lem:scope-reach-graph-only-if} and
Lemma~\ref{lem:scope-reach-graph-if} below.

\begin{lemma} \label{lem:scope-reach-graph-only-if}
    Given a scope-bounded CPDS $\cpds$ and control states $\controlin$ and
    $\controlout$, there is a run of $\cpds$ from
    $\config{\controlin}{\stackw_1, \ldots, \stackw_\numstacks}$ to
    $\config{\controlout}{\stackw'_1, \ldots, \stackw'_\numstacks}$ for some
    $\stackw'_1, \ldots, \stackw'_\numstacks$ only if there is a path in
    $\reachgraph{\cpds}{\controlout}$ from an initial vertex to a vertex
    \[ 
        \tuple{\control_0, \saauta_1, \control_1, \ldots,
        \control_{\numstacks-1}, \saauta_\numstacks, \control_\numstacks} 
    \]
    where for all $\idxi$ we have $\config{\control_{\idxi-1}}{\stackw_\idxi}$
    accepted from the $1$st layer of $\saauta_\idxi$ and $\control_0 =
    \controlin$.
\end{lemma}
\begin{proof}
    Take a run of the scope-bounded CPDS from $\config{\controlin}{\stackw_1,
    \ldots, \stackw_\numstacks}$ to $\config{\controlout}{\stackw'_1, \ldots,
    \stackw'_\numstacks}$.  We proceed by induction over the number of rounds in
    the run.  In the following we will override the $\stackw_\idxi$ and
    $\stackw'_\idxi$ in the statement of the lemma to ease notation.

    In the base case, take a single round 
    \[
        \config{\control_0}{\stackw_1, \ldots, \stackw_\numstacks} \cpdsrun
        \config{\control_1}{\stackw'_1, \stackw_2, \ldots, \stackw_\numstacks}
        \cpdsrun \cdots \cpdsrun \config{\control_\numstacks}{\stackw'_1,
        \ldots, \stackw'_\numstacks}
    \]
    where $\control_\idxi$ is the control state after the run on stack $\idxi$,
    and $\stackw'_\idxi$ is the $\idxi$th stack at the end of this run.  Take an
    initial vertex 
    \[
        \tuple{\control_0, \saauta_1, \control_1, \ldots,
        \control_{\numstacks-1}, \saauta_\numstacks, \control_\numstacks} \ .
    \]
    We know $\saauta_\idxi$ is constructed by saturation from an automaton
    accepting $\config{\control_\idxi}{\stackw'_\idxi}$ and thus
    $\config{\control_{\idxi-1}}{\stackw_\idxi}$ is accepted by $\saauta_\idxi$
    from the $1$st layer.  This vertex then gives us a path in the reachability
    graph to a vertex where for all $\idxi$ we have
    $\config{\control_{\idxi-1}}{\stackw_\idxi}$ accepted from the $1$st layer
    of $\saauta_\idxi$.

    Now consider the inductive step where we have a round
    \[
        \config{\control_0}{\stackw_1, \ldots, \stackw_\numstacks} \cpdsrun
        \config{\control_1}{\stackw'_1, \stackw_2, \ldots, \stackw_\numstacks}
        \cpdsrun \cdots \cpdsrun \config{\control_\numstacks}{\stackw'_1,
        \ldots, \stackw'_\numstacks}
    \]
    and a run from $\config{\control_\numstacks}{\stackw'_1, \ldots,
    \stackw'_\numstacks}$ to the destination control state.  By induction we
    have a vertex in the reachability graph 
    \[
        \tuple{\control'_0, \saauta'_1, \control'_1, \ldots,
        \control'_{\numstacks-1}, \saauta'_\numstacks, \control'_\numstacks} 
    \]
    with $\control_\numstacks = \control'_0$ that is reachable from an initial
    vertex and has for all $\idxi$ that
    $\config{\control'_{\idxi-1}}{\stackw'_\idxi}$ is accepted from the $1$st
    layer of $\saauta'_\idxi$.

    By definition of the reachability graph, there exists an edge to this vertex
    from a vertex
    \[
        \tuple{\control_0, \saauta_1, \control_1, \ldots,
        \control_{\numstacks-1}, \saauta_\numstacks, \control_\numstacks} \ .
    \]
    such that $\saauta_\idxi = \apply{\lsapredecessor{\idxi}}{\saauta'_\idxi,
    \sastate_{\control_\idxi}, \sastate_{\control'_{\idxi-1}}}$.

    Since the run of $\cpds$ is scope-bounded, we know there is an accepting run
    of $\stackw'_\idxi$ from $\salyrst{\sastate_{\control'_{\idxi-1}}}{1}$ in
    $\saauta'_\idxi$ that does not use any layer $\numscopes$ states (by the
    further condition described below and since layer $\numscopes$ corresponds
    to the round out of scope for elements of $\stackw'_\idxi$).  Therefrom, we
    have an accepting run of $\stackw'_\idxi$ from
    $\salyrst{\sastate_{\control'_{\idxi-1}}}{2}$ in
    $\apply{\lsashift}{\saauta'_\idxi}$.  Thus, there is an accepting run of
    $\stackw'_\idxi$ from $\salyrst{\control_\idxi}{1}$ after the application of
    $\lsaenvmove$.  Since there is a run over stack $\idxi$ from
    $\config{\control_{\idxi-1}}{\stackw_\idxi}$ to
    $\config{\control_\idxi}{\stackw'_\idxi}$ we therefore have an accepting run
    of $\stackw_\idxi$ from $\salyrst{\sastate_{\control_{\idxi-1}}}{1}$ in
    $\saauta_\idxi$.

    In addition to the above, we need a further property that reflects the scope
    boundedness.  In particular, if no character or stack with pop- or
    collapse-round $0$ is removed during the $\idxz$th round, then there is a
    run over $\stackw_\idxi$ that uses only transitions $\sastate
    \satran{\sastate'} \sastateset$ to read stacks $\stacku$ such that no layer
    $\idxz$ state is in $\sastateset$ and, similarly, for characters $\cha$, the
    run uses only transitions $\sastate
    \satrancol{\cha}{\sastateset_\branch}{\sastateset}$ to read the instance of
    $\cha$ where no layer $\idxz$ state appears in $\sastateset$ and no layer
    $\idxz$ state appears in $\sastateset_\branch$.  
    
    Note that the base case is for the automata accepting any stack, only
    containing transitions to the empty set, for which the property is trivial.
    In the inductive step, we prove this property by further induction over the
    length of the run from $\config{\control_\idxi}{\stackw_\idxi}$ to
    $\config{\control_{\idxi+1}}{\stackw'_\idxi}$.  In the base case we have a
    run of length $0$ and the property holds since, by induction, we can assume
    that $\saauta'_\idxi$ has the property (with the round numbers shifted) and
    it is maintained by the $\lsashift$ and $\lsaenvmove$.  Hence, assume we
    have a run beginning $\config{\control}{\stackw} \cpdstran
    \config{\control'}{\stackw'}$ and the required run over $\stackw'$.  We do a
    case split on the stack operation $\genop$ associated with the transition.
    \begin{enumerate}
        \item If $\genop = \pop{\opord}$ then we have $\stackw =
              \ccompose{\stacku}{\opord}{\stackv}$ and $\stackw' = \stackv$.  If
              $\idxz = 1$ and $\stacku$ has pop-round $0$ (i.e.  appears in
              $\stackw_\idxi$), then this case cannot occur because the
              transition we're currently analysing appears in round $1$ and by
              assumption $\stacku$ is not removed in round $1$.  Hence, assume
              $\idxz > 1$.  We had a run over $\stackw'$ from
              $\satranfullk{\salyrst{\sastate_{\control'}}{1}}{\sastate_\opord}{\sastateset_{\opord+1},
              \dots, \sastateset_\cpdsord}$ in $\saauta_\idxi$ respecting the
              property, and by saturation we have a run over $\stackw$ beginning
              with 
              \[
                  \satranfull{\salyrst{\sastate_{\control}}{1}}{\cha}{\emptyset}{\emptyset,
                  \ldots, \emptyset, \set{\sastate_\opord},
                  \sastateset_{\opord+1}, \ldots, \sastateset_\cpdsord}
              \]
              that also respects the property, since $\sastate_\opord$ is layer
              $1$ and $\idxz \neq 1$.

        \item When $\genop = \scopy{\opord}$ we have $\stackw =
              \ccompose{\stacku}{\opord}{\stackv}$ and $\stackw' =
              \ccompose{\stacku}{\opord}{\ccompose{\stacku}{\opord}{\stackv}}$.
              Let
              $\satranfull{\salyrst{\sastate_{\control'}}{1}}{\cha}{\sastateset_\branch}{\sastateset_1,
              \ldots, \sastateset_\opord, \ldots \sastateset_\cpdsord}$ and
              $\satranfull{\sastateset_\opord}{\cha}{\sastateset'_\branch}{\sastateset'_1,
              \ldots, \sastateset'_\opord}$ be the initial transitions used on
              the run of $\stackw'$.  We know neither these transitions, nor the
              runs from these transitions, pass a layer $\idxz$ state on any
              component with pop- or collapse-round $0$.  Furthermore, we know
              the first $\stacku$ has pop-round $1$.  The second $\stacku$ may
              have pop-round $0$.  If it does, we know $\sastateset'_\opord$
              does not contain any layer $\idxz$ states.

              From the saturation algorithm, we have a transition
              \[ 
                  \satranfull{\salyrst{\sastate_{\control}}{1}}{\cha}{\sastateset_\branch
                  \cup \sastateset'_\branch}{\sastateset_1 \cup \sastateset'_1,
                  \ldots, \sastateset_{\opord-1} \cup \sastateset'_{\opord-1},
                  \sastateset'_\opord, \sastateset_{\opord+1}, \ldots,
                  \sastateset_\cpdsord} \ .
              \]
              from which we have an accepting run of $\stackw$ that satisfies
              the property.

        \item If $\genop = \collapse{\opord}$, $\stackw =
              \ccompose{\annot{\cha}{\stacku'}}{1}{\ccompose{\stacku}{(\opord+1)}{\stackv}}$
              and $\stackw' = \ccompose{\stacku'}{(\opord+1)}{\stackv}$.  When
              $\opord = \cpdsord$, we have an accepting run of $\stackw'$
              respecting the property, and from the saturation, an accepting run
              of $\stackw$ beginning with a transition
              $\satranfull{\salyrst{\sastate_{\control}}{1}}{\cha}{\set{\salyrst{\sastate_{\control'}}{1}}}{
              \emptyset, \ldots, \emptyset}$ and $\stackw' = \stacku'$.  When
              $\idxz = 1$ and $\cha$ has collapse-round $0$, this case cannot
              occur because the transition we're currently analysing appears in
              round $1$ (similarly to the $\pop{\opord}$ case).  Otherwise
              $\idxz > 1$ and we have a run over $\stackw$ respecting the
              property.  

              When $\opord < \cpdsord$, we have an accepting run of $\stackw'$
              in beginning with
              $\satranfullk{\salyrst{\sastate_{\control'}}{1}}{\sastate_\opord}{\sastateset_{\opord+1},
              \dots, \sastateset_\cpdsord}$ that respects the property.  By
              saturation, we have an accepting run of $\stackw$ beginning with a
              transition
              $\satranfull{\salyrst{\sastate_{\control}}{1}}{\cha}{\set{\sastate_\opord}}{\emptyset,
              \ldots, \emptyset, \sastateset_{\opord+1}, \ldots,
              \sastateset_\cpdsord}$.  If the collapse-round of $\cha$ is $0$
              and $\idxz = 1$, this case cannot occur.  Otherwise, the run over
              $\stackw$ satisfies the property since the run over $\stackw'$
              does and $\sastate_\opord$ is layer $1$ and $\idxz > 1$.

        \item When $\genop = \cpush{\chc}{\opord}$, let $\stackw =
              \ccompose{\stacku_{\opord-1}}{\opord}{\ccompose{\stacku_\opord}{\opord+1}{\ccompose{\cdots}{\cpdsord}{\stacku_\cpdsord}}}$.
              We know $\stackw' = \apply{\cpush{\chc}{\opord}}{\stackw}$ is 
              \[
                  \ccompose{\annot{\chc}{\stacku_\opord}}{1}{\ccompose{\stacku_{\opord-1}}{\opord}{\ccompose{\cdots}{\cpdsord}{\stacku_\cpdsord}}}
                  \ .
              \]
              Let
              $\satranfull{\salyrst{\sastate_{\control'}}{1}}{\chc}{\sastateset_\branch}{\sastateset_1,
              \ldots, \sastateset_\cpdsord} \quad \text{and} \quad \sastateset_1
              \satrancol{\cha}{\sastateset'_\branch} \sastateset'_1$ be the
              first transitions used on the accepting run of $\stackw'$.  If the
              pop-round of $\cha$ is $0$, we know there are no layer $\idxz$
              states in $\sastateset'_1$.  Similarly if the pop-round of
              $\stacku_\opord$ is $0$ we know that there are no layer $\idxz$
              states in $\sastateset_\branch$.  The saturation algorithm means
              we have
              $\satranfull{\salyrst{\sastate_{\control}}{1}}{\cha}{\sastateset'_\branch}{\sastateset'_1,
              \sastateset_2,  \ldots, \sastateset_\opord \cup
              \sastateset_\branch, \ldots, \sastateset_\cpdsord}$ leading to an
              accepting run that respects the property.

        \item If $\genop = \rew{\chb}$ then $\stackw =
              \ccompose{\annot{\cha}{\stacku}}{1}{\stackv}$ and $\stackw' =
              \ccompose{\annot{\chb}{\stacku}}{1}{\stackv}$.  Note none of the
              pop- or collapse-rounds are changed, and the run of $\stackw'$
              beginning
              $\satranfull{\salyrst{\sastate_{\control'}}{1}}{\chb}{\sastateset_\branch}{\sastateset_1,
              \dots, \sastateset_\cpdsord}$ and satisfying the property implies
              a run of $\stackw$ beginning
              $\satranfull{\salyrst{\sastate_{\control}}{1}}{\cha}{\sastateset_\branch}{\sastateset_1,
              \dots, \sastateset_\cpdsord}$ and also satisfying the property. 

        \item If $\genop = \noop$ then $\stackw =
              \ccompose{\annot{\cha}{\stacku}}{1}{\stackv}$ and $\stackw' =
              \ccompose{\annot{\cha}{\stacku}}{1}{\stackv}$.  Note none of the
              pop- or collapse-rounds are changed, and the run of $\stackw'$
              beginning
              $\satranfull{\salyrst{\sastate_{\control'}}{1}}{\cha}{\sastateset_\branch}{\sastateset_1,
              \dots, \sastateset_\cpdsord}$ and satisfying the property implies
              a run of $\stackw$ beginning
              $\satranfull{\salyrst{\sastate_{\control}}{1}}{\cha}{\sastateset_\branch}{\sastateset_1,
              \dots, \sastateset_\cpdsord}$ and also satisfying the property. 
    \end{enumerate}

    Finally then, by induction over the number of rounds, we reach the first
    round beginning with $\config{\control_0}{\stackw_1, \ldots,
    \stackw_\numstacks}$ and we know there is a path from an initial vertex to a
    vertex
    \[ 
        \tuple{\control_0, \saauta_1, \control_1, \ldots,
        \control_{\numstacks-1}, \saauta_\numstacks, \control_\numstacks} 
    \]
    with $\control_0 = \control$ and for all $\idxi$ we have
    $\config{\control_{\idxi-1}}{\stackw_\idxi}$ accepted from the $1$st layer
    of $\saauta_\idxi$.
\end{proof}

\begin{lemma} \label{lem:scope-reach-graph-if}
    Given a scope-bounded CPDS $\cpds$ and control states $\controlin$ and
    $\controlout$, there is a run of $\cpds$ from
    $\config{\controlin}{\stackw_1, \ldots, \stackw_\numstacks}$ to
    $\config{\controlout}{\stackw'_1, \ldots, \stackw'_\numstacks}$ for some
    $\stackw'_1, \ldots, \stackw'_\numstacks$ whenever there is a path in
    $\reachgraph{\cpds}{\controlout}$ from an initial vertex to a vertex
    \[ 
        \tuple{\control_0, \saauta_1, \control_1, \ldots,
        \control_{\numstacks-1}, \saauta_\numstacks, \control_\numstacks} 
    \]
    with $\control_0 = \controlin$ and for all $\idxi$ we have
    $\config{\control_{\idxi-1}}{\stackw_\idxi}$ accepted from the $1$st layer
    of $\saauta_\idxi$.
\end{lemma}
\begin{proof}
    Note, in the following proof, we override the $\stackw_\idxi$ and
    $\stackw'_\idxi$ in the statement of the lemma.  Take a path in the
    reachability graph.  The proof goes by induction over the length of the
    path.  When the path is of length $0$ we have a single vertex
    $\tuple{\control_0, \saauta_1, \control_1, \ldots, \control_{\numstacks-1},
    \saauta_\numstacks, \control_\numstacks}$.  Take any configuration
    $\config{\control_{\idxi-1}}{\stackw_\idxi}$ accepted by $\saauta_\idxi$.
    We know $\saauta_\idxi$ accepts all configurations that can reach
    $\config{\control_\idxi}{\stackw}$ for some $\stackw$.  Therefore, from the
    initial configuration 
    \[
        \config{\control_0}{\stackw_1, \ldots, \stackw_\numstacks}
    \]
    we first apply the run over the $1$st stack to $\control_1$ to obtain 
    \[
        \config{\control_1}{\stackw'_1, \stackw_2, \ldots, \stackw_\numstacks}
    \]
    for some $\stackw'_1$.  Then we apply the run over the $2$nd stack to
    $\control_2$ and so on until we reach
    \[
        \config{\control_\numstacks}{\stackw'_1, \ldots, \stackw'_\numstacks}
    \]
    for some $\stackw'_1, \ldots, \stackw'_\numstacks$.  This witnesses the
    reachability property as required. 
    
    Now consider the inductive case where we have a path beginning with an edge
    of the reachability graph from 
    \[ 
        \tuple{\control_0, \saauta_1, \control_1, \ldots,
        \control_{\numstacks-1}, \saauta_\numstacks, \control_\numstacks} 
    \]
    to
    \[ 
        \tuple{\control'_0, \saauta'_1, \control'_1, \ldots,
        \control'_{\numstacks-1}, \saauta'_\numstacks, \control'_\numstacks} \ .
    \]
    By induction we have a run from 
    \[
        \config{\control_\numstacks}{\stackw'_1, \ldots, \stackw'_\numstacks}
    \]
    to the final control state for any $\stackw'_\idxi$ accepted by
    $\saauta'_\idxi$ from $\salyrst{\sastate_{\control_{\idxi-1}}}{1}$.

    Now, similarly to the base case, take any configuration
    $\config{\control_{\idxi-1}}{\stackw_\idxi}$ accepted by $\saauta_\idxi$.
    We know $\saauta_\idxi$ accepts all configurations that can reach
    $\config{\control_\idxi}{\stackw}$ for some $\stackw$ accepted from
    $\salyrst{\sastate_{\control'_{\idxi-1}}}{2}$ in
    $\apply{\lsashift}{\saauta'_\idxi}$ and therefore, from
    $\salyrst{\sastate_{\control'_{\idxi-1}}}{1}$ in $\saauta'_\idxi$.  Hence,
    from the initial configuration 
    \[
        \config{\control_0}{\stackw_1, \ldots, \stackw_\numstacks}
    \]
    we first apply the run over the $1$st stack to $\control_1$ to obtain 
    \[
        \config{\control_1}{\stackw'_1, \stackw_2, \ldots, \stackw_\numstacks}
    \]
    for some $\stackw'_1$.  Then we apply the run over the $2$nd stack to
    $\control_2$ and so on until we reach
    \[
        \config{\control_\numstacks}{\stackw'_1, \ldots, \stackw'_\numstacks}
    \]
    for some $\stackw'_1, \ldots, \stackw'_\numstacks$ and then, by induction,
    we have a run from this configuration to the target control state as
    required.
    
    We need to prove a stronger property that we can in fact build a
    scope-bounded run.  In particular, we show that, for all stacks $\stacku$ in
    $\stackw_\idxi$, if the accepting run of $\stackw_\idxi$ uses only
    transitions $\sastate \satran{\sastate'} \sastateset$ to read $\stacku$ such
    that no layer $\idxz$ state is in $\sastateset$, then there is a run to the
    final control state such that $\stacku$ is not popped during round $\idxz$.
    Similarly, for characters $\cha$, if the accepting run uses only transitions
    $\sastate \satrancol{\cha}{\sastateset_\branch}{\sastateset}$ to read the
    instance of $\cha$ where no layer $\idxz$ state appears in $\sastateset$,
    then $\cha$ is not popped in round $\idxz$.  Similarly, if no layer $\idxz$
    state appears in $\sastateset_\branch$, then collapse is not called on that
    character during round $\idxz$.  We observe the property is trivially true
    for the base case where the automata accept any stack using only transitions
    to $\emptyset$.  The inductive case is below.
    
    We start from $\config{\control}{\stackw} =
    \config{\control_\idxi}{\stackw_\idxi}$.  First assign each stack and
    character in $\stackw$ pop- and collapse-round $0$.  Noting that $\saauta$
    is obtained by saturation from $\saauta'$ (after a $\lsashift$ and
    $\lsaenvmove$ --- call this automaton $\saautb$), we aim to exhibit a run
    from $\config{\control}{\stackw}$ to
    $\config{\control_{\idxi+1}}{\stackw_{\idxi+1}}$ (in fact we choose
    $\stackw_{\idxi+1}$ via this procedure) such that all stacks and characters
    in $\stackw_{\idxi+1}$ with pop- or collapse-round $0$ do not pass layer
    $\idxz$ states in $\saautb$.  Since we have a run over $\stackw_{\idxi+1}$
    in $\saauta'_\idxi$ that does not pass layer $1$ states for parts of the
    stack with pop- or collapse-round $0$, we know by induction we have a run
    from $\config{\control_{\idxi+1}}{\stackw_{\idxi+1}}$ that is scope bounded.

    To generate such a run we follow the counter-example generation algorithm
    in~\cite{BCHS13}.  We refer the reader to this paper for a precise
    exposition of the algorithm.  Furthermore, that this routine terminates is
    non-trivial and requires a subtle well-founded relation over stacks, which
    is also shown in~\cite{BCHS13}.

    Beginning with the run over $\config{\control_\idxi}{\stackw_\idxi}$ that
    has the property of not passing layer $\idxz$ states, we have our base case.
    Now assume we have a run to $\config{\control}{\stackw}$ such that the run
    over $\stackw$ has no transitions to layer $\idxz$ states reading stacks or
    characters with pop- or collapse-rounds of $0$.  We take the first
    transition of such a run, which was introduced by the saturation algorithm
    because of a rule $\cpdsrule{\control}{\cha}{\genop}{\control'}$ and certain
    transitions of the partially saturated $\saautb$.  Let
    $\config{\control'}{\stackw'}$ be the configuration reached via this rule.
    We do a case split on $\genop$.  
    \begin{enumerate}
        \item If $\genop = \pop{\opord}$, then we have $\stackw =
              \ccompose{\stacku}{\opord}{\stackv}$ and the accepting run of
              $\stackw$ begins with 
              \[
                  \satranfull{\salyrst{\sastate_{\control}}{1}}{\cha}{\emptyset}{\emptyset,
                  \ldots, \emptyset, \set{\sastate_\opord},
                  \sastateset_{\opord+1}, \ldots, \sastateset_\cpdsord}
              \]
              where
              $\satranfullk{\salyrst{\sastate_{\control'}}{1}}{\sastate_\opord}{\sastateset_{\opord+1},
              \dots, \sastateset_\cpdsord}$ was already in $\saautb$.  This
              gives us an accepting run of $\stackv$ beginning with this
              transition.  Note that $\sastate_\opord$ is of layer $1$.  Thus,
              if $\stacku$ has pop-round $0$ and $\idxz = 1$, this case cannot
              occur.  Otherwise, we have that the run of $\stackv$ visits a
              subset of the states in the run over $\stackw$ and thus maintains
              the property.

        \item If $\genop = \scopy{\opord}$, then we have $\stackw =
              \ccompose{\stacku}{\opord}{\stackv}$ and $\stackw' =
              \ccompose{\stacku}{\opord}{\ccompose{\stacku}{\opord}{\stackv}}$.
              Furthermore, we had an accepting run of $\stackw$ using the
              initial transition
              \[ 
                  \satranfull{\salyrst{\sastate_{\control}}{1}}{\cha}{\sastateset_\branch
                  \cup \sastateset'_\branch}{\sastateset_1 \cup \sastateset'_1,
                  \ldots, \sastateset_{\opord-1} \cup \sastateset'_{\opord-1},
                  \sastateset'_\opord, \sastateset_{\opord+1}, \ldots,
                  \sastateset_\cpdsord} 
              \]
              and an accepting run of $\saautb$ on $\stackw'$ using the initial
              transitions
              $\satranfull{\salyrst{\sastate_{\control'}}{1}}{\cha}{\sastateset_\branch}{\sastateset_1,
              \ldots,\sastateset_\opord,\ldots, \sastateset_\cpdsord}$ and
              $\satranfull{\sastateset_\opord}{\cha}{\sastateset'_\branch}{\sastateset'_1,
              \ldots, \sastateset'_\opord}$ from which we have an accepting run
              over $\stackw'$.  Note that, to prove the required property, we
              observe that for all elements of $\stackw'$ obtaining their pop-
              and collapse-rounds from $\stackw$, the targets of the transitions
              used to read them already appear in the run of $\stackw$, hence
              the run satisfies the property.  The only new part of the run is
              to $\sastateset'_\opord$ after reading the new copy of $\stacku$,
              which has pop-round $1$.  Thus the property is maintained.

        \item If $\genop = \collapse{\opord}$ then we have $\stackw =
              \ccompose{\annot{\cha}{\stacku'}}{1}{\ccompose{\stacku}{(\opord+1)}{\stackv}}$
              and $\stackw' = \ccompose{\stacku'}{(\opord+1)}{\stackv}$.  When
              $\opord = \cpdsord$, the accepting run of $\stackw$ begins with a
              transition
              $\satranfull{\salyrst{\sastate_{\control}}{1}}{\cha}{\set{\salyrst{\sastate_{\control'}}{1}}}{
              \emptyset, \ldots, \emptyset}$ and $\stackw' = \stacku'$.  When
              $\idxz = 1$ and $\cha$ has collapse-round $0$, this case cannot
              occur because the initial transition goes to a layer $\idxz$
              state.  Otherwise, we have a run over $\stackw'$ that is a subrun
              of that over $\stackw$, and thus the property is transferred.  

              When $\opord < \cpdsord$, the accepting run of $\stackw$ begins
              with
              $\satranfull{\salyrst{\sastate_{\control}}{1}}{\cha}{\set{\sastate_\opord}}{\emptyset,
              \ldots, \emptyset, \sastateset_{\opord+1}, \ldots,
              \sastateset_\cpdsord}$ and we have an accepting run of $\stackw'$
              in $\saautb$ beginning with
              $\satranfullk{\salyrst{\sastate_{\control'}}{1}}{\sastate_\opord}{\sastateset_{\opord+1},
              \dots, \sastateset_\cpdsord}$.  If the collapse-round of $\cha$ is
              $0$ and $\idxz = 1$, this case cannot occur because
              $\sastate_\opord$ is layer $\idxz$.  Otherwise, the run over
              $\stackw'$ is a subrun of that over $\stackw$ and the property is
              transferred.

        \item If $\genop = \cpush{\chb}{\opord}$ then $\stackw' =
              \ccompose{\annot{\chb}{\stacku}}{1}{\stackw}$ where $\stacku =
              \apply{\ctop{\opord+1}}{\apply{\pop{\opord}}{\stackw}}$ and the
              collapse-round of $\chb$ is the pop-round of
              $\apply{\ctop{\opord}}{\stackw}$.  The run of $\stackw$ begins
              with a transition
              \[ 
                  \satranfull{\salyrst{\sastate_{\control}}{1}}{\cha}{\sastateset'_\branch}{\sastateset'_1,
                  \sastateset_2, \ldots, \sastateset_{\opord-1},
                  \sastateset_\opord \cup \sastateset_\branch,
                  \sastateset_{\opord+1}, \ldots, \sastateset_\cpdsord} 
              \]
              and there is a run over $\stackw'$ in $\saautb$ beginning with
              $\satranfull{\salyrst{\sastate_{\control'}}{1}}{\chb}{\sastateset_\branch}{\sastateset_1,
              \ldots, \sastateset_\cpdsord}$ and $\sastateset_1
              \satrancol{\cha}{\sastateset'_\branch} \sastateset'_1$.  Note
              that, to prove the required property, we observe that for all
              elements of $\stackw'$ obtaining their pop- and collapse-rounds
              from $\stackw$, the targets of the transitions used to read them
              already appear in the run of $\stackw$, hence the run satisfies
              the property.  The only new parts of the run are to
              $\sastateset'_1$ after reading $\chb$, which has pop-round $1$,
              and the transition to $\sastateset_\branch$ on the collapse branch
              of $\chb$.  Note, however, that $\chb$ has the collapse-round
              equal to the pop-round of $\apply{\ctop{\opord}}{\stackw}$ and
              hence we know that $\sastateset_\branch$ has no layer $\idxz$
              states if the collapse-round of $\chb$ is $0$.  Thus the property
              is maintained.

        \item If $\genop = \rew{\chb}$ then $\stackw =
              \ccompose{\annot{\cha}{\stacku}}{1}{\stackv}$ and $\stackw' =
              \ccompose{\annot{\chb}{\stacku}}{1}{\stackv}$.  Note none of the
              pop- or collapse-rounds are changed, and the run of $\stackw$
              beginning
              $\satranfull{\salyrst{\sastate_{\control}}{1}}{\cha}{\sastateset_\branch}{\sastateset_1,
              \dots, \sastateset_\cpdsord}$ and satisfying the property implies
              a run of $\stackw'$ in $\saautb$ beginning
              $\satranfull{\salyrst{\sastate_{\control'}}{1}}{\chb}{\sastateset_\branch}{\sastateset_1,
              \dots, \sastateset_\cpdsord}$ and also satisfying the property. 

        \item If $\genop = \noop$ then $\stackw =
              \ccompose{\annot{\cha}{\stacku}}{1}{\stackv}$ and $\stackw' =
              \ccompose{\annot{\cha}{\stacku}}{1}{\stackv}$.  Note none of the
              pop- or collapse-rounds are changed, and the run of $\stackw$
              beginning
              $\satranfull{\salyrst{\sastate_{\control}}{1}}{\cha}{\sastateset_\branch}{\sastateset_1,
              \dots, \sastateset_\cpdsord}$ and satisfying the property implies
              a run of $\stackw'$ in $\saautb$ beginning
              $\satranfull{\salyrst{\sastate_{\control'}}{1}}{\cha}{\sastateset_\branch}{\sastateset_1,
              \dots, \sastateset_\cpdsord}$ and also satisfying the property. 
    \end{enumerate}
    Thus we are done.
\end{proof}

\subsection{Complexity}

Solving the control state reachability problem requires finding a path in the
reachability graph.  Since each vertex can be stored in
$\bigo{\exptower{\cpdsord-1}{\polyof{\numscopes,\numof}}}$ space, where $\poly$
is a polynomial and $\numof$ the number of control states, and we require
$\bigo{\exptower{\cpdsord-1}{\polyof{\numscopes,\numof}}}$ time to decide the
edge relation, we have via Savitch's algorithm, a
$\bigo{\exptower{\cpdsord-1}{\polyof{\numscopes,\numof}}}$ space procedure for
deciding the control state reachability problem.  We also observe that the
solution to the global control state reachability problem may contain at most
$\bigo{\exptower{\cpdsord}{\polyof{\numscopes,\numof}}}$ tuples.

\end{document}